\documentclass{lmcs}
\pdfoutput=1

\usepackage{lastpage}
\lmcsdoi{15}{2}{15}
\lmcsheading{}{\pageref{LastPage}}{}{}%
{Jun.~06,~2018}{May~23,~2019}{}

\usepackage{graphicx, hyperref}
\usepackage{amssymb}
\usepackage{xspace}
\usepackage{pifont} % needed for qed mark
\usepackage{misc}
\usepackage{thmtools,thm-restate}
\usepackage{tikz, tkz-graph, models, misc}
\usetikzlibrary{arrows, shapes, snakes, backgrounds, petri}

\begin{document}

\title[Rewritability in MDDLog, MMSNP, and Expressive DLs]{Rewritability in Monadic Disjunctive Datalog, MMSNP, and Expressive Description Logics}
\titlecomment{{\lsuper*} The current paper is the extended version of an invited conference abstract \cite{ICDT2017}. It contains detailed proofs of all results and new material regarding dichotomies and deciding \PTime query evaluation. }

\begin{abstract}
  We study rewritability of monadic disjunctive Datalog programs,
  (the complements of) MMSNP sentences, and ontology-mediated queries (OMQs)
  based on expressive description logics of the \ALC family and on
  conjunctive queries. We show that rewritability into FO and into
  monadic Datalog (MDLog) are decidable, and that rewritability into
  Datalog is decidable when the original query satisfies a certain
  condition related to equality. % The latter gives a potentially
  % incomplete procedure for the general case.
  We establish {\sc 2NExpTime}-completeness for all studied problems except
  rewritability into MDLog for which there remains a gap between
  {\sc 2NExpTime} and {\sc 3ExpTime}. We also analyze the shape of
  rewritings, which in the case of MMSNP correspond to obstructions, and
  give a new construction of canonical Datalog programs that is more
  elementary than existing ones and also applies to non-Boolean queries.
\end{abstract}

\author[C.~Feier]{Cristina Feier\rsuper{a,*}}	%required
\address{\lsuper{a}University of Bremen, Department of Computer Science, Germany}	%required
%\email{lastname@uni-bremen.de}  %optional
\thanks{\lsuper{*}The author was supported by ERC Consolidator Grant 647289 CODA}	%optional

\author[A.~Kuusisto]{Antti Kuusisto\rsuper{b,*}}	%optional
\address{\lsuper{b}Tampere University, Mathematics, Finland}	%optional
%\email{name2@email2; ditto for email addresses}  %optional

\author[C.~Lutz]{Carsten Lutz\rsuper{b,*}}	%optional
%\urladdr{name3@url3\quad\rm{(optionally, a web-page can be specified)}}  %optional

\maketitle

\section{Introduction}

In data access with ontologies, the premier aim is to answer queries
over incomplete and heterogeneous data while taking advantage of the
domain knowledge provided by an ontology 
\cite{DBLP:journals/tods/BienvenuCLW14,DBLP:conf/rweb/CalvaneseGLLPRR09,
  DBLP:conf/rweb/BienvenuO15}.  Since traditional database systems are
often unaware of ontologies, it is common to rewrite the emerging
ontology-mediated queries (OMQs) into more standard database query
languages. For example, the DL-Lite family of description logics (DLs)
was designed as an ontology language specifically so that any OMQ
$Q=(\Tmc,\Sigma,q)$ where \Tmc is a DL-Lite ontology, $\Sigma$ a data
signature, and $q$ a conjunctive query, can be rewritten into an
equivalent first-order (FO) query that can then be executed using a
standard SQL database system
\cite{calvanese-2007,Artale09thedl-lite}. In more expressive ontology
languages, it is not guaranteed that for every OMQ, there is an
equivalent FO query. For example, this is the case for DLs of the \EL
and Horn-\ALC families %\cite{rosati07on,DBLP:conf/aaai/EiterOSTX12},
%\cite{baader-2005,DBLP:journals/jar/HustadtMS07},
and for DLs of the expressive \ALC family; please see
\cite{DL-Textbook} for a general introduction to DLs.  In many members
of the \EL and Horn-\ALC families, however, rewritability into monadic
Datalog (MDLog) is guaranteed, thus enabling the use of Datalog
engines for query answering. In \ALC and above, not even
Datalog-rewritability is generally ensured. Since ontologies emerging
from practical applications tend to be structurally simple, though,
there is reason to hope that (FO-, MDLog-, and Datalog-) rewritings do
exist in many practically relevant cases even when the ontology is
formulated in an expressive language. This has in fact been
experimentally confirmed for FO-rewritability in the \EL family of DLs
\cite{HLSW-IJCAI15}, and it has led to the implementation of rewriting
tools that, although incomplete, are able to compute rewritings in
many practical cases
\cite{perezurbina10tractable,conf/rr/KaminskiNG14,DBLP:journals/ws/TrivelaSCS15}.

Fundamental problems that emerge from this situation are to understand
the exact limits of rewritability and to provide (complete) algorithms
that decide the rewritability of a given OMQ and that compute a
rewriting when it exists. These problems have been adressed in
\cite{BienvenuLW13,HLSW-IJCAI15,BHLW-IJCAI16,LutzSabellekIJCAI17} for
DLs from the \EL and Horn-\ALC families. For DLs from the \ALC family,
first results were obtained in \cite{DBLP:journals/tods/BienvenuCLW14}
where a connection between OMQs and constraint satisfaction problems
(CSPs) was established and then used to transfer decidability results
from CSPs to OMQs. In fact, rewritability is an important topic in CSP
(where it is called definability) as it constitutes a central tool for
analyzing the complexity of CSPs
\cite{DBLP:journals/siamcomp/FederV98,DBLP:journals/lmcs/LaroseLT07,DBLP:conf/lics/EgriLT07,DBLP:conf/lics/DalmauL08}. In
particular, it is known that deciding the rewritability of (the
complement of) a given CSP into FO and into Datalog is {\sc
  NP}-complete
\cite{DBLP:journals/lmcs/LaroseLT07,DBLP:journals/logcom/Barto16,Chen2017}
and rewritability into MDLog is {\sc NP}-hard and in \ExpTime
\cite{Chen2017}. In
\cite{DBLP:journals/tods/BienvenuCLW14}, these results were used to
show that FO- and Datalog-rewritability of OMQs $(\Tmc,\Sigma,q)$ is
decidable and in fact {\sc NExpTime}-complete when \Tmc is formulated
in \ALC or a moderate extension thereof and $q$ is an atomic query
(AQ), that is, a monadic query of the simple form $A(x)$. For
MDLog-rewritability, one can show {\sc NExpTime}-hardness and
containment in {\sc 2ExpTime}.

The aim of this paper is to study the above questions for OMQs where
the ontology is formulated in an expressive DL from the \ALC family
and where the actual query is a conjunctive query (CQ) or a union of
conjunctive queries (UCQ). As observed in
\cite{DBLP:journals/tods/BienvenuCLW14}, transitioning in OMQs from
AQs to UCQs corresponds to the transition from CSP to its logical
generalization MMSNP introduced by Feder and Vardi
\cite{DBLP:journals/siamcomp/FederV98} and studied, for example, in
\cite{DBLP:journals/siamcomp/MadelaineS07,DBLP:journals/corr/abs-0904-2521,DBLP:conf/cp/Madelaine10,DBLP:journals/siamdm/BodirskyCF12}. More
precisely, while the OMQ language $(\ALC,\text{AQ})$ that consists of
all OMQs $(\Tmc,\Sigma,q)$ where \Tmc is formulated in \ALC and $q$ is
an AQ has the same expressive power as the complement of CSP (with
multiple templates and a single constant), the OMQ language
$(\ALC,\text{UCQ})$ has the same expressive power as the complement of
MMSNP (with free variables)---which in turn is essentially a notational variant of
monadic disjunctive Datalog (MDDLog). It should be noted, however,
that while all these formalisms are equivalent in expressive power,
they differ significantly in succinctness
\cite{DBLP:journals/tods/BienvenuCLW14}; in particular, the best known
translation of OMQs from $(\ALC,\text{UCQ})$ into MMSNP/MDDLog
involves a double exponential blowup. In contrast to the CSP case,
FO-, MDLog-, and Datalog-rewritability of (complemented) MMSNP
sentences was not known to be decidable. In this paper, we establish
decidability of FO- and MDLog-rewritability in $(\ALC,\text{UCQ})$ and
related OMQ languages, in MDDLog, and in complemented MMSNP. We
show that FO-rewritability is {\sc 2NExpTime}-complete in all three
cases, and that MDLog-rewritability is in {\sc 3ExpTime}; a {\sc
  2NExpTime} lower bound was established in \cite{KR-submitted}.  Let
us discuss our results on the complexity of FO-rewritability from
three different perspectives. From the OMQ perspective, the transition
from AQs to UCQs results in an increase of complexity from {\sc
  NExpTime} to {\sc 2NExpTime}. From the monadic Datalog perspective,
adding disjunction (transitioning from monadic Datalog to MDDLog)
results in a moderate increase of complexity from {\sc 2ExpTime}
\cite{DBLP:conf/lics/BenediktCCB15} to {\sc 2NExpTime}. And from the
CSP perspective, the transition from CSPs to MMSNP results in a rather
dramatic complexity jump from {\sc NP} to {\sc 2NExpTime}.

For Datalog-rewritability, we obtain only partial results. In
particular, we show that Datalog-rewritability is decidable and {\sc
  2NExpTime}-complete for MDDLog programs that, in a
certain technical sense made precise in the paper, \emph{have
  equality}.  For the general case, we only obtain a potentially
incomplete procedure.  It is well possible that the procedure is in
fact complete, but proving this remains an open issue for now. These
results also apply to analogously defined classes of MMSNP sentences
and OMQs that have equality.

While we mainly focus on deciding whether a rewriting exists rather
than actually computing it, we also %do a first step towards computing
%rewritings by analyzing 
analyze the shape that rewritings can take. Since the shape turns out
to be rather restricted, this is important information for algorithms
(complete or incomplete) that seek to compute rewritings. In the
CSP/MMSNP world, this corresponds to analyzing obstruction sets for
MMSNP, in the style of CSP obstructions
\cite{DBLP:conf/bonnco/Nesetril08,DBLP:conf/dagstuhl/BulatovKL08,Atserias08}
and not to be confused with colored forbidden patterns sometimes used
to characterize MMSNP \cite{DBLP:journals/siamcomp/MadelaineS07}. More
precisely, we show that an OMQ $(\Tmc,\Sigma,q)$ from
$(\ALC,\text{UCQ})$ is FO-rewritable if and only if it is rewritable
into a UCQ in which each CQ has treewidth $(1,\max\{2,n_q\})$, where
$n_q$ is the size of $q$;\footnote{What we mean here is that $q$ has a
  tree decomposition in which every bag has at most $\max\{2,n_q\}$
  elements and in which neighboring bags overlap in at most one
  element.} similarly, the complement of an MMSNP sentence $\varphi$
is FO-definable if and only if it admits a finite set of finite
obstructions of treewidth $(1,k)$ where $k$ is the diameter of
$\varphi$ (the maximum number of variables in a negated conjunction in
its body, in Feder and Vardi's terminology). We also show that an OMQ
$(\Tmc,\Sigma,q)$ is MDLog-rewritable if and only if it is rewritable
into an MDLog program of diameter $\max\{2,n_q\}$ where the diameter
of an MDLog program is the maximum number of variables in a rule;
equivalently, the  complemented of an MMSNP sentence $\varphi$ is
MDLog-definable if and only if it admits a (potentially infinite) set
of finite obstructions of treewidth $(1,k)$ where $k$ is the diameter
of $\varphi$. For the case of rewriting into Datalog, we give a new
and direct construction of canonical Datalog-rewritings of MMSNP
sentences. It has been observed in
\cite{DBLP:journals/siamcomp/FederV98} that for every CSP and all
$\ell,k$, it is possible to construct a canonical Datalog program
$\Pi$ of width $\ell$ and diameter $k$ (the width is the maximum arity
of IDB relations) in the sense that if any such program is a rewriting
of the CSP, then so is $\Pi$; moreover, even when there is no
$(\ell,k)$-Datalog rewriting, then $\Pi$ is the best possible
approximation of such a rewriting. The existence of canonical
Datalog-rewritings for (complemented) MMSNP sentences was already
known from \cite{DBLP:journals/jcss/BodirskyD13}. However, the
construction given there is quite complex, proceeding via an infinite
template that is obtained by applying an intricate construction due to
Cherlin, Shelah, and Shi \cite{CherlinEtAl99}, and resulting in
canonical programs that are rather difficult to understand and to
analyze. In contrast, our construction is elementary and essentially
parallels the CSP case; it also applies to MMSNP formulas with free
variables, where the canonical program takes a rather special form
that involves \emph{parameters}, similar in spirit to the parameters
to least fixed-point operators in FO(LFP) \cite{BenediktLICS16}.

Our main technical tool is the translation of  MMSNP sentences into 
CSPs exhibited by Feder and Vardi
\cite{DBLP:journals/siamcomp/FederV98}; actually, the target of the
translation is a generalized CSP, meaning that there are multiple
templates. The translation is not equivalence preserving and involves
a double exponential blowup, but it was designed so as to preserve
complexity up to polynomial time reductions. Here, we are primarily
interested in the semantic relationship between the original MMSNP
sentence and the constructed CSP. It turns out that the translation
does not quite preserve rewritability. In particular, when the
original MMSNP sentence has a rewriting, then the natural way of
constructing from it a rewriting for the CSP is sound only on
instances of high girth. However, FO- and MDLog-rewritings of CSPs
that are sound on high girth (and unconditionally complete) can be
converted into rewritings that are unconditionally sound (and
complete). The same is true for Datalog-rewritings when the CSP is
derived from an MMSNP sentence that has equality, but it remains open
whether it is true for Datalog-rewritings of unrestricted CSPs.

With our translations in place, we can also make relevant observations
regarding the (data) complexity of query evaluation in MMSNP, in
MDDLog, and of OMQs. This is especially interesting in the light of
the recently obtained breakthrough in CSPs that there is a dichotomy
between \PTime and {\sc NP} in the complexity of CSPs
\cite{Bulatov17,Zhuk17} and that it is decidable and {\sc NP}-complete
whether a CSP is in \PTime \cite{Chen2017}.  We
show that this implies a dichotomy between \PTime and {\sc coNP} for
MDDLog and for all OMQ languages mentioned above. We also prove that
\PTime query evaluation is decidable and 2\NExpTime-complete in
the mentioned query languages, and that the same is true for MMSNP.

The structure of this paper is as follows. In
Section~\ref{sec:prelim}, we introduce the essentials of disjunctive
Datalog and its relevant fragments as well as CSP and MMSNP; in fact,
we shall always work with Boolean MDDLog rather than with complemented
MMSNP.  In Section~\ref{sec:shape}, we summarize the main properties
of Feder and Vardi's translation of MMSNP into CSP. We use this in
Section~\ref{sec:dec} to show that FO- and MDLog-rewritability of
Boolean MDDLog programs and of the complement of an MMSNP sentences is
decidable, also establishing the announced complexity results. In
Section~\ref{sect:shapeObstrExplosion}, we analyze the shape of FO-
and MDLog-rewritings and of obstructions for MMSNP sentences. We also
establish an MMSNP analogue of an essential combinatorial lemma for
CSPs which says that it is possible to replace a structure by a
structure of high high girth while preserving certain homomorphisms;
the MMSNP analogue achieves high `decomposition girth' (defined in the
paper) and preserves the truth of certain MMSNP sentences.  In
Section~\ref{sect:dlog-rewr}, we study Datalog-rewritability of MDDLog
programs that have equality and construct canonical Datalog
programs. Section~\ref{sec:answer} is concerned with lifting our
results from the Boolean case to the general case, concerning the
complexity of deciding rewritability, the shape of rewritings, and the
construction of canonical Datalog programs. In this section, Datalog
programs with parameters play a central role. In
Section~\ref{sect:omqs}, we introduce OMQs and further lift our
results to that setting, finally arriving at our goal to study
fundamental rewritability questions for OMQ languages based on (unions
of) conjunctive queries. Section~\ref{sect:dicho} is then concerned
with dichotomies and the complexity of deciding \PTime query
evaluation. We conclude in Section~\ref{sec:concl}.

% {\color{blue} we need a more detailed intro}

% {\bf Related work.} FO-rewritability was first studied in an OMQ 
% context for the inexpressive DL-Lite family of DLs 
% \cite{DBLP:conf/dlog/Perez-UrbinaMH09,DBLP:conf/icalp/KikotKPZ12}.  FO-rewritability in OMQ languages based on 
% more expressive Horn DLs has been investigated in 
% \cite{BieLutWo-DL12,BienvenuLW13,HLSW-IJCAI15}.  Rewritability of Horn 
% DLs into Datalog was considered in \cite{rosati07on,perezurbina10tractable,DBLP:conf/aaai/EiterOSTX12,conf/rr/KaminskiNG14,DBLP:journals/ws/TrivelaSCS15}. 

\section{Preliminaries}
\label{sec:prelim}
%We assume prior knowledge of DL languages such as $\ALC$, $\shi$, and query languages like union of conjunctive queries (UCQs), Boolean conjunctive queries(BCQs), etc.

We introduce disjunction Datalog, CSP, and MMSNP. To avoid overloading
this section, the introduction of ontology languages and
ontology-mediated queries is deferred to Section~\ref{sect:omqs}.

\smallskip

A \emph{schema} is a finite collection $\Sbf=(S_1,\dots,S_k)$ of
relation symbols with associated arity. An \emph{\Sbf-fact} is an
expression of the form $S(a_1,\ldots, a_n)$ where $S\in\Sbf$ is an
$n$-ary relation symbol, and $a_1, \ldots, a_n$ are elements of some
fixed, countably infinite set $\mn{const}$ of \emph{constants}. For an
$n$-ary relation symbol~$S$, $\mn{pos}(S)$ is $\{1, \ldots, n\}$.  An
\emph{\Sbf-instance} $I$ is a finite set of \Sbf-facts.  The
\emph{active domain} $\adom(I)$ of~$I$ is the set of all constants
that occur in the facts in~$I$. We use the symbols \abf, \bbf, \cbf to
denote tuples of constants and, slightly abusing notation, write $\abf
\subseteq \mn{dom}(I)$ to mean that \abf is a tuple of constants from
$\mn{dom}(I)$ when we do not want to be precise about the length of
the tuple. For an instance $I$ and a schema \Sbf, we write $I|_\Sbf$
to denote the restriction of $I$ to the relation symbols in \Sbf.

A \emph{tree decomposition} of an instance $I$ is a pair $(T, (B_v)_{v \in V})$, where $T=(V, E)$ is an undirected tree and $(B_v)_{v \in V} $ is a family of subsets of $\adom(I)$ such that the following conditions are satisfied:
\begin{enumerate} 
\item for all $a \in \adom(I)$, $\{v \in V\mid a \in B_v\}$ is
  nonempty and connected in~$T$;
\item for every fact $R(a_1, \ldots a_r)$ in $I$, there is a $v \in V$ such that $a_1, \ldots, a_r \in B_v$.
\end{enumerate}

Unlike in the traditional setup \cite{FlumG06}, we are interested in two parameters of tree decompositions instead of only one.  We call $(T, (B_v)_{v \in V})$ an \emph{$(\ell,k)$-tree 
decomposition} if for all distinct $v,v' \in V$, 
%
%\begin{enumerate}
%\item 
$|B_v \cap B_{v'}| \leq \ell$ and 
%\item 
$|B_v| \leq k$. 
%\end{enumerate}
An instance $I$ \emph{has treewidth} $(\ell,k)$ if it admits an $(\ell,k)$-tree decomposition. % Note that this coincides with the notion of $(\ell,k)$ tree width for relational structures from \cite{DBLP:journals/jcss/BodirskyD13}.

We now define the notion of girth. %A  subset $C$ of  an instance $I$ is a \emph{cycle} if  either:
%\begin{enumerate}
%\item  it consists of a single fact $R(a_1,\dots,a_n)$ with $a_i =a_j$ for some distinct $i,j \leq n$, 
%\item  it consists of two facts $R_1(\vect{a}_1)$ and $R_2(\vect{a}_2)$  with $|\vect{a}_1 \cap \vect{a}_2| \geq 2$, or
%\item  it contains at least three facts and the facts in $C$ can be ordered into $R_1(\vect{a}_1),\dots,R_n(\vect{a}_n)$, such that there exist distinct constants $a_1,\dots,a_n$ with $\{a_i\} = \vect{a}_i \cap \vect{a}_{i+1}$ for all $1 \leq i \leq n$ where $\abf_{n+1} := \abf_1$.
%\end{enumerate}
%The \emph{length} of a cycle is one for cycles of the former kind, two for cycles of the second kind, and $n$ for cycles of the latter kind.
A finite structure $I$ has a \emph{cycle} of length~\mbox{$n>0$} if
there are
distinct facts $R_0(\vect{a}_0),\dots, R_{n-1}(\vect{a}_{n-1}) \in I$
%$\abf_i=a_{i,1} \cdots a_{i,m_i}$, 
and positions
$p_0,p'_0 \in \mn{pos}(R_0),\dots,p_{n-1},p'_{n-1} \in \mn{pos}(R_{n-1})$
%$p_i,p'_i \in \mn{pos}(R_i)$, $0 \leq i < n$ 
such that
\begin{itemize}
\item $p_i \neq p'_i$ for $0 \leq i < n$ and
\item the constant in the $p'_i$-th position of $\abf_i$ is identical
  to the constant in the $p_{i+1}$-th position of $\abf_{i+1}$
% $a_{i, p'_i}=a_{i\oplus1, p_{i\oplus1}}$
 for $0 \leq i < n$ and with $p_n := p_0$ and $p'_n := p'_0$.
%  $\oplus$ denotes addition modulo $n$.
\end{itemize}
The \emph{girth} of $I$ is the length of the shortest cycle in it and
$\omega$ if $I$ has no cycle (in which case we say that $I$ is a
\emph{tree}).

\medskip
\noindent %\emph{Constraint Satisfaction Problems.} 
A \emph{constraint satisfaction problem (CSP)} is defined by an
instance $T$ over a schema $\Sbf_E$, called
\textit{template}.\footnote{Adopting Datalog terminology, we generally
  call the schema in which the data is formulated the \emph{EDB
    schema} and denote it with $\Sbf_E$.} The problem associated with
$T$, denoted $\text{CSP}(T)$, is to decide whether an input instance
$I$ over $\Sbf_E$ admits a homomorphism to~$T$, denoted
$I \rightarrow T$. We use coCSP$(T)$ to denote the complement problem,
that is, deciding whether $I \not \rightarrow T$.  A \emph{generalized
  CSP} is defined by a set of templates $S$ over the same schema
$\Sbf_E$ and asks for a homomorphism from the input $I$ to at least
one of the templates $T \in S$, denoted $I \rightarrow S$. Note that a
(generalized) CSP can be viewed as a Boolean query over
$\Sbf_E$-instances.

%Erdoes lemma
%FO rewritability of CSPs/generalized CSPs

\medskip
An \emph{MMSNP sentence} $\theta$ over schema $\Sbf_E$ has the form $\exists X_1 \cdots \exists X_n \forall x_1 \cdots \forall x_m \vp$  with $X_1,\dots, X_n$ monadic second-order variables,
$x_1,\dots,x_m$ first-order variables, and $\vp$ a conjunction of quantifier-free formulas of the form
\[ 
\beta_1 \vee \cdots \vee \beta_n 
\leftarrow 
\alpha_1 \wedge \cdots\wedge \alpha_m
\]
with $n,m \geq 0$ and where each $\alpha_i$ takes the form $X_i(x_j)$
or $R(\vect{x})$ %, or $x_j=x_k$
with $R \in \Sbf_E$ and each $\beta_i$ takes the form $X_i(x_j)$. The
\emph{diameter} of $\theta$ is the maximum number of variables in some
implication in $\vp$. This presentation is syntactically different
from, but semantically equivalent to the original definition from
\cite{DBLP:journals/siamcomp/FederV98}, which does not use the
implication symbol and instead restricts the allowed polarities of
atoms. % We also relax the usual definition of MMSNP by allowing free FO variables -- in this case we speak of MMSNP formulas instead of sentences.
Both forms can be interconverted in polynomial time, see
\cite{DBLP:journals/tods/BienvenuCLW14}.  The semantics of MMSNP is
the standard semantics of second-order logic. More information can be
found, e.g., in
\cite{DBLP:journals/siamdm/BodirskyCF12,DBLP:journals/jcss/BodirskyD13}.
Note that, just like CSPs, MMSNP sentences can be viewed as Boolean
queries. 

\medskip A \emph{conjunctive query (CQ)} takes the form
$q(\xbf)=\exists \vect{y} \, \vp(\vect{x}, \vect{y})$ where $\vp$ is a
conjunction of relational atoms and \xbf, \ybf denote tuples of
variables; the equality relation may be used. Whenever convenient, we
shall confuse $q(\xbf)$ with the set of atoms in $\varphi$. The
variables in \xbf are the \emph{answer variables} in $q(\xbf)$. The
\emph{arity} of $q(\xbf)$ is the number of its answer variables and
$q(\xbf)$ is \emph{Boolean} if it has arity zero. We say that
$q(\xbf)$ is \emph{over} $\Sbf_E$ if $\vp$ uses only relation symbols
from $\Sbf_E$. An \emph{answer} to $q$ on an $\Sbf_E$-instance $I$ is
a tuple of constants $\abf$ such that $I \models q(\abf)$ in the
standard sense of first-order logic. It is well-known that $I \models
q(\abf)$ if and only if there is a homomorphism from $\vp$ viewed as
an instance to $I$ that takes $\xbf$ to $\abf$.  A Boolean CQ $q$ is
\emph{true} on an instance $I$, denoted $I \models q$, if the empty
tuple is an answer to $q$ on $I$. A CQ $q$ is a \emph{contraction} of
a CQ $q'$ if it can be obtained from $q'$ by identifying variables.  A
\emph{union of conjunctive queries (UCQ)} is a disjunction of CQs with
the same answer variables.  The semantics of UCQs is defined in the
expected way.

A \emph{disjunctive Datalog rule}
$\rho$ has the form
\[S_1(\vect{x}_1) \vee \cdots \vee S_m(\vect{x}_m) \leftarrow
R_1(\vect{y}_1)\land \cdots\land R_n(\vect{y}_n)\]
with $n> 0$ and $m \geq 0$. % \footnote{Empty rule heads (denoted
  % $\bot$) are sometimes disallowed. We admit them only in our upper bound
  % proofs, but do not use them for lower bounds, thus achieving maximum
  % generality. {\color{blue}from KR paper where this was needed. do we need this here, too?} }
%(note that here and elsewhere, we use bold font to indicate tuples of variables
%or constants).
We refer to $S_1(\vect{x}_1) \vee \cdots \vee S_m(\vect{x}_m)$ as the
\emph{head} of $\rho$, and to $R_1(\vect{y}_1) \wedge \cdots \wedge
R_n(\vect{y}_n)$ as the \emph{body}. Every variable that occurs in the
head of $\rho$ is required to also occur in the body of $\rho$.
A \emph{disjunctive Datalog (DDLog) program} $\Pi$ is a finite set of
disjunctive Datalog rules with a selected \emph{goal relation}
{\mn{goal}} that does not occur in rule bodies and appears only in
non-disjunctive \emph{goal rules}
$\mn{goal}(\vect{x}) \leftarrow R_1(\vect{x}_1) \wedge \cdots \wedge
R_n(\vect{x}_n)$.  The \emph{arity of $\Pi$} is the arity of
the %selected
\mn{goal} relation; we say that $\Pi$ is \emph{Boolean} if it has
arity zero.  Relation symbols that occur in the head of at least one
rule of $\Pi$ are \emph{intensional (IDB) relations}, and all
remaining relation symbols in $\Pi$ are \emph{extensional (EDB)
  relations}.  Note that, by definition, \mn{goal} is an IDB
relation. When all relations in $\Pi$ are from schema $\Sbf_E$,
then we say that $\Pi$ is \emph{over EDB schema} $\Sbf_E$. The
\emph{IDB schema} of $\Pi$ is the set of all IDB relations in $\Pi$. 

We will sometimes use body atoms of the form $\mn{true}(x)$ that are
vacuously true for all elements of the active domain. This is just
syntactic sugar since any rule with body atom $\mn{true}(x)$ can
equivalently be replaced by a set of rules obtained by replacing
$\mn{true}(x)$ in all possible ways with an atom $R(x_1,\dots,x_n)$
where $R$ is a relation symbol from $\Sbf_E$ and where $x_i=x$ for
some $i$ and all other $x_i$ are fresh variables.

% To make some constructions more transparent, we generally
% require that no variable repetitions are allowed in \mn{goal} atoms;
% all results obtained in this paper, however, can be lifted to the
% general case.
%
A DDLog program is called \emph{monadic} or an \emph{MDDLog program}
if all its IDB relations with the possible exception of \mn{goal} have
arity at most one.  The \emph{size} of a DDLog program $\Pi$ is the
number of symbols needed to write it (where relation symbols and
variable names count one), its \emph{width} is the maximum arity of
non-\mn{goal} IDB relations used in it, and its \emph{diameter} is the
maximum number of variables that occur in a rule in $\Pi$.

A \emph{Datalog rule} is a disjunctive Datalog rule in which the rule
head contains exactly one disjunct. Datalog (DLog) programs and
monadic Datalog (MDLog) programs are then defined in the expected way.
We call a Datalog program an \emph{$(\ell,k)$-Datalog program} if its
width is $\ell$ and its diameter is $k$.

For $\Pi$ an $n$-ary DDLog program over schema $\Sbf_E$, an
$\Sbf_E$-instance $I$, and $a_1, \ldots, a_n \in \adom(I)$, we write
$I \models \Pi(a_1, \ldots, a_n)$ if
$\Pi \cup I \models \mn{goal}(a_1, \ldots , a_n)$ where the variables
in all rules of $\Pi$ are universally quantified and thus $\Pi$ is a
set of first-order (FO) sentences; please see
\cite{Abiteboul:1995:FDL:551350,DBLP:journals/tods/EiterGM97} for more
information on the semantics of (disjunctive) Datalog. A query
$q(\xbf)$ over $\Sbf_E$ of arity $n$ is
\begin{itemize}

\item \emph{sound for} $\Pi$ if for all $\Sbf_E$-instances $I$ and
  $\abf \subseteq \adom(I)$, $I \models q(\abf)$ implies $I \models
  \Pi(\abf)$;

\item \emph{complete for} $\Pi$ if for all $\Sbf_E$-instances $I$ and
  $\abf \subseteq \adom(I)$, $I \models \Pi(\abf)$ implies $I \models
  q(\abf)$;

\item a \emph{rewriting of} $\Pi$ if it is sound for $\Pi$ and
  complete for $\Pi$.

\end{itemize}
Note that Boolean programs are also covered by the above definitions.
To additionally specify the syntactic shape of $q(\xbf)$, we speak of
a UCQ-rewriting, an MDLog-rewriting, and so on.  An
\emph{FO-rewriting} takes the form of an FO-query that uses only
relation symbols from the relevant EDB schema and possibly equality,
but neither constants nor function symbols.  We say that an MDDLog
program $\Pi$ is \emph{FO-rewritable} if there is an FO-rewriting of
$\Pi$, and likewise for \emph{UCQ-rewritability} and for
\emph{MDLog-rewritability}. Since a generalized CSP defined by a set
of templates $S$ can be viewed as a Boolean query, we can also speak
of a query to be sound and complete for respectively a rewriting of
coCSP$(S)$. The definitions are as expected, paralleling the ones above.

\smallskip

It was shown in \cite{DBLP:journals/tods/BienvenuCLW14} that the
complement of an MMSNP sentence can be translated into an equivalent
Boolean MDDLog program in polynomial time and vice versa; moreover,
the transformations preserve diameter and all other parameters
relevant for this paper. From now on, we will thus not explicitly
distinguish between Boolean MDDLog and (complemented) MMSNP.
 % ; 
% that is, semantically, the MDDLog program defines the same Boolean
% query as the negation of the MMSNP sentence. 

\section{MDDLog, Simple MDDLog and CSP}
\label{sec:shape}

Feder and Vardi show how to translate an MMSNP sentence into (the
complemen of) a generalized CSP that has the same complexity up to
polynomial time reductions \cite{DBLP:journals/siamcomp/FederV98}. The
resulting CSP has a different schema than the original MMSNP sentence
and is thus not equivalent to it. We are going to make use of this
translation to reduce rewritability problems for MDDLog to the
corresponding problems for CSPs.  Consequently, our main interest is
in the precise semantic relationship between the MMSNP sentence and
the constructed CSP, rather than in their complexity. In this section,
we sum up the properties of the results obtained in
\cite{DBLP:journals/siamcomp/FederV98} that are relevant for us. These
properties are all we need in later sections, that is, we do not need
to go further into the details of the translation. For the reader's
convenience and information, we still describe the translations in
full detail in Appendix~\ref{app:translate}; these are based on the
presentation given in \cite{KR-submitted}, which is more detailed than
the original presentation in \cite{DBLP:journals/siamcomp/FederV98}.

Let $\Sbf_E$ be a schema. A schema $\Sbf'_E$ is a
\emph{$k$-aggregation schema} for $\Sbf_E$ if 
%
%\begin{enumerate}
%
%\item 
its relations have the form $R_{q(\xbf)}$ where $q(\xbf)$ is a CQ over
$\Sbf_E$ without quantified variables and the arity of $R_{q(\xbf)}$
is identical to the length of $\xbf$, which is at most $k$.
%
% \item if $R_{q(\xbf)}$ is a relation in $\Sbf'_E$ and $q'(\ybf)$ is
%   obtained from $q(\xbf)$ by consistently identifying variables, then
%   $R_{q'(\ybf)}$ is a relation in $\Sbf'_E$.
% \end{enumerate} 
%
The generalized CSP to be constructed
later makes use of a schema of this form. What is important at this
point is that there are natural translations of instances between the
two schemas.  To make this precise, let $I$ be an
$\Sbf_E$-instance. The \emph{corresponding $\Sbf'_E$-instance} $I'$
consists of all facts $R_{q(\xbf)}(\abf)$ such that $I \models
q(\abf)$. Conversely, let $I'$ be an $\Sbf'_E$-instance. The
\emph{corresponding $\Sbf_E$-instance} $I$ consists of all facts
$S(\bbf)$ such that $R_{q(\xbf)}(\abf) \in I'$ and $S(\bbf)$ is a conjunct of
$q(\abf)$.

\begin{figure}[t]
\begin{center}
\begin{tabular}{cc|cc|cc}

$\Sbf'_E$-instance $I'$
&&
$\Sbf_E$-instance $I$
&&
$\Sbf'_E$-instance $I''$
\\

\begin{tikzpicture}[auto, scale=1.5]
\node at (0, 2.5) {};
\draw [rounded corners,  fill=gray!20,  draw opacity=1] (0, -0.25)--(0.7, 1.15)--(-0.7, 1.15)--cycle;
\draw [rounded corners, fill=gray!20,  draw opacity=1] (-1.35, 2.25)--(0.25, 2.25)--(-0.5, 0.75)--cycle;
\draw [rounded corners,  fill=gray!20,  draw opacity=1] (-0.25, 2.25)--(1.35, 2.25)--(0.5, 0.75)--cycle;

\node at (0, 0) {$b$};
\node at (-0.5, 1.05) {$a'$};
\node at (0.5, 1.05) {$b'$};
\node at (1.1, 2.1) {$c$};
\node at (0, 2.1) {$c'$};
\node at (-1.1, 2.1) {$a$};

\node(hl1) at (-0.6, 1.6) {$R_q$};
\node(hl2) at (0.6, 1.6) {$R_q$};
\node(ml1) at (0, 0.5) {$R_q$};

\end{tikzpicture} 
&&
\begin{tikzpicture}[auto, scale=1.5]
\node at (0, -0.3) {};

\node(la1) at (0, -0.2) {$b$};
\node(ma1) at (-0.5, 1) {$a'$};
\node(ma2) at (0.5, 1) {$b'$};
\node(ha3) at (1, 2) {$c$};
\node(ha2) at (0, 2) {$c'$};
\node(ha1) at (-1, 2) {$a$};

\node(hl1) at (-0.5, 2.15) {$r$};
\node(hl2) at (0.5, 2.15) {$r$};
\node(ml1) at (-0.9, 1.5) {$r$};
\node(ml2) at (-0.4, 1.5) {$r$};
\node(ml3) at (0.15, 1.5) {$r$};
\node(ml4) at (0.6, 1.5) {$r$};
\node(ll1) at (-0.4, 0.5) {$r$};
\node(ll2) at (0.4, 0.5) {$r$};
\node(ll2) at (0, 1.15) {$r$};

\draw [->] (la1)  to (ma2);
\draw [->] (ma2)  to (ma1);
\draw [->] (ma1)  to (la1);
\draw [->] (ma1)  to (ha2);
\draw [->] (ha2)  to (ha1);
\draw [->] (ha1)  to (ma1);
\draw [->] (ma2)  to (ha3);
\draw [->] (ha3)  to (ha2);
\draw [->] (ha2)  to (ma2);
\end{tikzpicture}
&&

\begin{tikzpicture}[auto, scale=1.5]
\node at (0, 2.5) {};
\draw [rounded corners,  fill=gray!20,  draw opacity=1] (0, -0.25)--(0.82, 1.25)--(-0.82, 1.25)--cycle;
\draw [rounded corners, fill=gray!20,  draw opacity=1] (-1.35, 2.25)--(0.25, 2.25)--(-0.6, 0.8)--cycle;
\draw [rounded corners,  fill=gray!20,  draw opacity=1] (0, 2.25)--(1.35, 2.25)--(0.6, 0.8)--cycle;
\draw [rounded corners, fill=gray!20,  draw opacity=1] (0, 2.3)--(-0.67, 0.85)--(0.67, 0.85)--cycle;

\node at (0, 0) {$b$};
\node at (-0.5, 1) {$a'$};
\node at (0.5, 1) {$b'$};
\node at (1.1, 2.1) {$c$};
\node at (-0.02, 2.05) {$c'$};
\node at (-1.1, 2.1) {$a$};

\node(hl1) at (-0.6, 1.6) {$R_q$};
\node(hl2) at (0.6, 1.6) {$R_q$};
\node(ml1) at (0, 0.5) {$R_q$};
\node(bla) at (0, 1.3) {$R_q$};

\end{tikzpicture}

\end{tabular}
\caption{Translating an $\Sbf'_E$-instance into an $\Sbf_E$-instance and back.}
\label{fig:transl} 
\end{center}
\end{figure}

\begin{exa}
\label{ex:veryfirst}
 Let $\Sbf_E = \{ r\}$ with $r$ a binary relation symbol, 
\[
\begin{array}{r@{\;}c@{\;}l}
  q(x_1,x_2,x_3) &=& r(x_1,x_2) \wedge r(x_2,x_3) \wedge r(x_3,x_1), %\\[1mm]
%  q_2(\xbf) &=&   A(x_1) \wedge r(x_1,x_2) \wedge r(x_2,x_3) \wedge r(x_3,x_1) 
\end{array}
\]
and
% where $\xbf=(x_1,x_2,x_3)$ and
let $\Sbf'_E$ consist of
$R_{q(\xbf)}$ where $\xbf=(x_1,x_2,x_3)$ for brevity. Take the $\Sbf'_E$-instance $I'$ defined
by%\marginpar{\color{blue}rather use pictures? anyone?}
\[
  R_{q(\xbf)}(a,a',c'), R_{q(\xbf)}(b,b',a'), R_{q(\xbf)}(c,c',b').
\]
The corresponding $\Sbf_E$-instance $I$ is
\[
   r(a,a'),r(a',c'),r(c',a),r(b,b'),r(b',a'),r(a',b),r(c,c'),r(c',b),r(b',c).
\]
Note that when we transition from $\Sbf_E$ back to $\Sbf'_E$ and take
the $\Sbf'_E$-instance $I''$ corresponding to $I$, we do \emph{not}
obtain $I'$, but rather a strict superset that contains additional
facts such as $R_{q(\xbf)}(c',b',a')$. This is illustrated in Figure
\ref{fig:transl} which shows the instances $I$, $I'$, and a subset
of $I''$ that contains all facts from $I$ plus one additional fact.
\end{exa}
The translation in \cite{DBLP:journals/siamcomp/FederV98} consists of
two steps. We describe them here using Boolean MDDLog instead of
(complemented) MMSNP. The first step is to transform the given Boolean
MDDLog program $\Pi$ into a Boolean MDDLog program $\Pi_S$ over a
suitable aggregation schema $\Sbf'_E$ such that $\Pi_S$ is of a
restricted syntactic form called \emph{simple}. In the second step,
one transforms $\Pi_S$ into a generalized CSP whose complement is
equivalent to~$\Pi_S$.

We start with summing up the important aspects of the first step. A
Boolean MDDLog program $\Pi_S$ is \emph{simple} if it satisfies the
following conditions:
\begin{enumerate}

\item every rule in $\Pi_S$ contains at most one EDB atom and this atom contains all variables of the rule body, each variable exactly once; 

\item rules without an EDB atom contain at most a single variable. 

\end{enumerate} 
Now, the first step achieves the following.
\begin{thmC}[\cite{DBLP:journals/siamcomp/FederV98}]
\label{th:mmsnpToSimple}   
Given a Boolean MDDLog program $\Pi$ over EDB schema $\Sbf_E$ of
diameter $k$, one can construct a simple Boolean MDDLog program
$\Pi_S$ over a $k$-aggregation EDB schema $\Sbf'_E$ for $\Sbf_E$ and
IDB schema $\Sbf'_I$ such that
\begin{enumerate}

  \item If $I$ is an $\Sbf_E$-instance and $I'$ the corresponding
    $\Sbf'_E$-instance, then $I \models \Pi$ iff $I' \models \Pi_S$;

  \item If $I'$ is an $\Sbf'_E$-instance and $I$ the corresponding
    $\Sbf_E$-instance, then 
    \begin{enumerate}

    \item $I' \models \Pi_S$ implies $I \models \Pi$;

    \item $I \models \Pi$ implies $I' \models \Pi_S$ if
      the girth of $I'$ exceeds $k$.

    \end{enumerate}

\end{enumerate}
If $\Pi$ is of size $n$, then the size of $\Pi_S$ and the cardinality
of $\Sbf'_E \cup \Sbf'_I$ are bounded by $2^{p(k \cdot \mn{log}n)}$,
where $p$ is a polynomial.
%
% $r$ is the number of rules in $\Pi$ and $s$ the maximum
% size of rule bodies,
% then % the size of $\Pi^S$, the cardinality of $\Sbf'_E$, and the number
% % of IDB relations in $\Pi^S$ are bounded by $r2^{d^2}$.
%   %
%   \begin{enumerate}
%
%   \setcounter{enumi}{2}
%
%   \item the size of $\Pi^S$ and the cardinality of $\Sbf'_E$ are
%     bounded by $p(r2^{s^2}$)
%     and the arity of relations in $\Sbf'_E$ is bounded by $s$;
%
%   \item $\Pi^S$ contains at most $p(r2^{s^2})$ more IDB relations than $\Pi$
%
%    % \item the diameter of %and atom width of 
%    %   $\Pi^S$ is bounded by that of $\Pi$;
%
% %  \item $|\Sbf'_E| \leq p(r \cdot 2^d)$.
%
%   \end{enumerate}
%   %
%  where $p$ is a polynomial. 
  The construction takes time polynomial
  in the size of~$\Pi_S$.
\end{thmC}
The translation underlying Theorem~\ref{th:mmsnpToSimple} consists of
three steps itself: first saturate $\Pi$ by adding all rules that can
be obtained as a \emph{contraction} of a rule in $\Pi$, that is, by
identifying variables in the rule body and head in a consistent way.
Then rewrite $\Pi$ in an equivalence-preserving way so that all rule
bodies are biconnected, introducing fresh unary and nullary IDB
relations as needed. And finally replace the conjunction $q(\xbf)$ of
all EDB atoms in each rule body with a single EDB atom
$R_{q(\xbf)}(\xbf)$, additionally taking care of interactions between
the new EDB relations that arise e.g.\ when we have two relations
$R_{q(\xbf)}$ and $R_{p(\xbf)}$ such that $q(\xbf)$ is contained in
$p(\xbf)$ in the sense of query containment. Details are in
Appendix~\ref{app:step1tosimple}.

The following theorem summarizes the second step of the translation
of Boolean MDDLog into a generalized CSP. 
\begin{thmC}[\cite{DBLP:journals/siamcomp/FederV98}]
\label{th:simpleToCsp}
Let $\Pi$ be a simple Boolean MDDLog program over EDB schema $\Sbf_E$
and with IDB schema $\Sbf_I$, $m$ the maximum arity of relations in
$\Sbf_E$. Then there exists a set of templates $S_\Pi$ over $\Sbf_E$
such that
\begin{enumerate}
\item  $\Pi$ is equivalent to coCSP$(S_\Pi)$;
\item $|S_\Pi| \leq 2^{|\Sbf_I|}$ and  $|T| \leq |\Sbf_E| \cdot 2^{m|\Sbf_I|}$ for
  each $T \in S_\Pi$;
\end{enumerate}
The construction takes time polynomial in $\sum_{T \in S_\Pi} |T|$. 
\end{thmC}
We again sketch the idea underlying the proof of the theorem. The
desired set of templates $S_\Pi$ contains one template for every
0-type, that is, for every set of nullary IDB relations in $\Pi$ that
does not contain $\mn{goal}()$ and that satisfies all rules in $\Pi$
which use only nullary IDBs. Each template contains one constant $c_M$
for every 1-type $M$, that is, for every set $M$ of unary IDBs that
agrees on nullary IDBs with the 0-type for which the template was
constructed and that satisfy all rules in $\Pi$ which use only IDB
relations that are at most unary. One then interprets all EDB
relations in a maximal way so that all rules in $\Pi$ are satisfied.
The fact that $\Pi$ is simple implies that no choices arise, that
is, there is only one maximal interpretation of each EDB relation and
the interpretations of different such relations do not interact.
Details are given in Appendix~\ref{app:step2tosimple}.

\section{FO- and MDLog-Rewritability of Boolean MDDLog Programs}
\label{sec:dec}

We exploit the translations described in the previous section and the
known results that FO-rewritability of CSPs and MDLog-rewritability of
coCSPs are decidable to obtain analogous results for Boolean MDDLog,
and thus also for MMSNP. In the case of FO-rewritability, we obtain
tight 2{\sc NExpTime} complexity bounds. For MDLog-rewritability, the
exact complexity remains open (as in the CSP case), between 2{\sc
  NExpTime} and 3{\sc ExpTime}.

We start with observing that FO-rewritability and
MDLog-rewrita\-bility are more closely related than one might think at
first glance. Recall that, by Rossman's homomorphism preservation
theorem \cite{RossmanHomPresT08}, a first-order formula is preserved
under homomorphisms on finite structures if and only if it is
equivalent to a UCQ. While every MDLog-rewriting can be viewed as an
infinitary UCQ-rewriting, Rossman's result implies that
FO-rewritability of a Boolean MDDLog program coincides with (finitary)
UCQ-rewritability. The latter is true also in the non-Boolean case.
\begin{prop}
\label{prop:ross}
  Let $\Pi$ be an MDDLog program. Then $\Pi$ is FO-rewritable iff
  $\Pi$ is UCQ-rewritable.
\end{prop}
\begin{proof}\
  It is well known and easy to show that truth of disjunctive Datalog
  programs is preserved under homomorphisms. Thus, the proposition
  immediately follows from Rossman's theorem in the Boolean case.  For
  the non-Boolean case, we observe that Rossman establishes his result
  also in the presence of constants. Let $\Pi$ be an MDDLog program
  and $\varphi(\xbf)$ a rewriting of $\Pi$. We can apply Rossman's
  result to $\varphi(\abf)$, where \abf is a tuple of constants of
  the same length as \xbf, obtaining a UCQ $q(\abf)$ equivalent to
  $\varphi(\abf)$. Let $q(\xbf)$ be obtained from $q(\abf)$ by
  replacing the constants in \abf with the variables from $\xbf$. 
  It can be verified that $q(\xbf)$ is a rewriting of $\Pi$.
\end{proof}
%
% Every CQ can be viewed as an instance by using the variables as
% constants and the atoms as facts. It thus makes sense to speak
% about tree decompositions of CQs and about their tree width.
% Note that the bags of a tree decomposition of a CQ is a set of
% variables. We will thus denote the bag associated with a node $v$
% with $\xbf_v$ rather than $B_v$. This imposes a fixed (but otherwise
% unimportant) order on the elements of the bag.
%
For utilizing the translation of Boolean MDDLog programs to generalized
CSPs in the intended way, the interesting aspect is to deal with the
translation of a Boolean MDDLog program $\Pi$ into a simple program
$\Pi_S$ stated in Theorem~\ref{th:mmsnpToSimple}, since it is not
equivalence preserving. The following proposition relates rewritings
of $\Pi$ to rewritings of $\Pi_S$. It also applies to
Datalog-rewritings, which we will make use of in Section~\ref{sect:dlog-rewr}.
\begin{lem}
\label{prop:rewrCorresp} 
Let $\Pi$ be a Boolean MDDLog program of diameter $k$, $\Pi_S$ as in
Theorem~\ref{th:mmsnpToSimple}, and $\Qmc \in \{
\text{UCQ},\text{MDLog},\text{DLog}\}$. Then
 \begin{enumerate} 
 \item every \Qmc-rewriting of $\Pi_S$ can effectively be converted into a \Qmc-rewriting of~$\Pi$;
 \item every \Qmc-rewriting of $\Pi$ can effectively be converted into
   a \Qmc-rewriting of $\Pi_S$ that is (i)~sound on instances of girth
   exceeding $k$ and (ii)~complete.
\end{enumerate}
\end{lem}
\begin{proof}\ Let $\Sbf_E$ and $\Sbf'_E$ be the EDB schema of $\Pi$
  and of $\Pi_S$, respectively. We start with the case
  $\Qmc = \text{UCQ}$.
  
%  Since FO-definability of CSPs is decidable in {\sc NP} by
%  \cite{DBLP:journals/lmcs/LaroseLT07}, by
%  Lemma~\ref{lem:FOgirthdontcare} it suffices to show the following:
% 

  For Point~1, let $q_{\Pi_S}$ be a UCQ-rewriting of
  $\Pi_S$.  Let $q_\Pi$ be the UCQ obtained from $q_{\Pi_S}$ by
  replacing every atom $R_{q(\xbf)}(\ybf)$ with $q[\ybf/\xbf]$, that
   is, with the result of replacing the variables \xbf in $q(\xbf)$
   with the variables~$\ybf$ (which may lead to identifications).
  We
  show that $q_\Pi$ is as required. Let $I$ be an $\Sbf_E$-instance
  and $I'$ %$I'=c(I)$
  the corresponding $\Sbf'_E$-instance. Then we have $I \models \Pi$
  iff $I' \models \Pi_S$ (by Point~1 of
  Theorem~\ref{th:mmsnpToSimple}) iff $I' \models q_{\Pi_S}$ (by
  choice of $q_{\Pi_S}$) iff $I \models q_\Pi$ (by construction of
  $I'$ and of $q_\Pi$).
  %
  % This part could be dropped in short version.
  %
  Let us expand on the latter. 

  First assume that $I' \models q_{\Pi_S}$.  Then there is a CQ $q$ in
  $q_{\Pi_S}$ and a homomorphism $h$ from $q$ to $I'$. By
  construction, $q_\Pi$ contains a CQ $q'$ that is obtained from $q$
  by replacing every atom $R_{q(\xbf)}(\ybf) \in q$ with
  $q[\ybf/\xbf]$. Clearly, for every atom $R_{q(\xbf)}(\ybf) \in q$,
  we must have $R_{q(x)}(h(\ybf)) \in I'$. The construction of $I'$
  yields $q(h(\ybf)) \subseteq I$. Consequently, $h$ is also a homomorphism
  from $q'$ to $I$.  Conversely, assume that there is a CQ $q'$ in
  $q_{\Pi}$ and a homomorphism $h$ from $q'$ to $I$. Then there is a
  CQ $q$ in $q_{\Pi_S}$ from which $q'$ was obtained by the described
  replacement operation. For every atom $R_{q(\xbf)}(\ybf) \in q$, we
  must have $q(h(\ybf)) \subseteq I$. We obtain $R_{q(\xbf)}(h(\ybf))
  \in q$ and thus $h$ is a homomorphism from $q$ to $I'$.

 % \smallskip

  For Point~2, let $q_\Pi$ be a UCQ-rewriting of $\Pi$. % By
  % Theorem~\ref{th:rewrshape}, we can assume that $q_\Pi$ contains only
  % CQs that have tree width $(1,k)$.
  %
  %  Moreover,
  % we can assume that every $k$-tree CQ $q$ in $q_\Pi$ is associated
  % with a $k$-decomposition $(V,E,(B_v)_{v \in V})$ with $(V,E)$ a tree
  % and such that each $B_v$ is biconnected. In fact, if some $B_v$ is 
  % not biconnected, we can split it at the articulation point,
  % replacing $v \in V$ with two nodes $v_1,v_2$. 
  %
  The UCQ $q_{\Pi_S}$ consists of all CQs that can be obtained as
  follows:
  \begin{enumerate}

  \item choose a CQ $\exists \xbf\, q(\xbf)$ from $q_\Pi$, a
    contraction $\exists \xbf' \, q'(\xbf')$ of $\exists \xbf\,
    q(\xbf)$, and a partition $q_1(\xbf_1),\dots,q_n(\xbf_n)$
    of $q'(\xbf')$; % such that
    % %
    % \begin{enumerate}

    % \item each $q_i$ is biconnected\footnote{\color{blue}nullary relations?} and has at most $k$ variables and

    % \item $q_i$ and $q_j$ share only a single variable, for $1 \leq i <
    %   j \leq n$.

    % \end{enumerate}

  \item for each $i \in \{1,\dots,n\}$, choose a relation
    $R_{p(\zbf)}$ from $\Sbf'_E$ and a tuple $\ybf$ of $|\zbf|$
    variables (repeated occurrences allowed) that are either from
    $\xbf_i$ or do not occur in $\xbf'$ such that
    $q_i(\xbf_i) \subseteq p[\ybf/\zbf]$; then replace $q_i(\xbf_i)$
    in $\exists \xbf' \, q'(\xbf')$ with the single atom $R_{p(\zbf)}(\ybf)$.

  \end{enumerate}
  To establish that $q_{\Pi_S}$ is as desired, we show that for every
  $\Sbf'_E$-instance $I'$
  \begin{enumerate}

  \item[(I)] $I' \models q_{\Pi_S}$ implies $I' \models \Pi_S$ if $I'$
    is of girth exceeding $k$ (soundness) and

  \item[(II)] $I' \models \Pi_S$ implies $I' \models q_{\Pi_S}$ (completeness).

  \end{enumerate}
  Let $I$ be the $\Sbf_E$-instance corresponding to $I'$.

% We then translate every CQ $q$ in $q_\Pi$ with tree width $(1,k)$ in multiple ways into schema~$\Sbf'$ and take as $q_S$ the disjunction of all resulting queries.
%
% Thus let $q$ be a CQ in $q_\Pi$ and let $(V, (I_v)_{v \in V}, R_V)$ be a $(1,k)$-decomposition of $q$ with $R_V$ a tree and $I_v$ biconnected. We translate $q$ into schema $\Sbf'$ by selecting, for each $v \in V$, an atom $R(\xbf')$ such that $I_v$ is a subset of the atoms represented by $R(\xbf')$.  Note that there might be more than one atom $R(\xbf')$ satisfying this property (which is why $q$ might translate into multiple disjuncts of $q_S$) and that $\xbf'$ might contain additional (fresh) variables to the ones which occur in $I_v$. % because the arity of $R$ might be higher than the number of variables in $B_v$.  
% Observe also that all CQs in $q_S$ are trees. %(because the domains of the bags of the $(1,k)$-decomposition overlap in at most a single variable).  
% We show that $q_S$ is as required.

  For Point~(I), we observe that $I' \models q_{\Pi_S}$ implies $I
  \models q_\Pi$ (by construction of $q_{\Pi_S}$ and of $I'$) implies
  $I \models \Pi$ (by choice of $q_\Pi$) implies $I' \models \Pi_S$
  (by Point~2b of Theorem~\ref{th:mmsnpToSimple} and if $I'$ is
  of girth exceeding $k$).  
  %
  % clu: could be dropped in short version
  %
  Let us zoom into the first implication. Assume that $I' \models
  q_{\Pi_S}$. Then there is a CQ $\exists \ubf \, p_0(\ubf)$ in
  $q_{\Pi_S}$ and a homomorphism $h$ from $p_0(\ubf)$
  to $I'$. There must be some CQ $\exists \xbf \, q(\xbf)$ in $q_\Pi$
  from which $\exists \ubf \, p_0(\ubf)$ has been constructed in
  Steps~1 and~2 above. Let $q_1(\xbf_1),\dots,q_n(\xbf_n)$ be as in
  this construction. It suffices to show that $h$ is a homomorphism
  from $q_i(\xbf_i)$ to $I$, for each~$i$. Thus fix a
  $q_i(\xbf_i)$. Then there is a relation $R_{p(\zbf)} \in \Sbf'_E$
  and a tuple \ybf of variables that are either from $\xbf_i$ or do
  not occur in $\xbf'$ such that $q_i(\xbf_i) \subseteq p[\ybf/\zbf]$
  and $R_{p(\zbf)}(\ybf) \in p_0(\ubf)$. Thus $R_{p(\zbf)}(h(\ybf))
  \in I'$. By construction of $I'$, this yields $q_i(h(\xbf_i))
  \subseteq I$ and thus we are done.

 % \smallskip

  For Point~(II), we have that $I' \models \Pi_S$ implies $I \models
  \Pi$ (by Point~2a of Theorem~\ref{th:mmsnpToSimple}) implies $I
  \models q_\Pi$ (by choice of $q_\Pi$). It thus remains to show that
  $I \models q_\Pi$ implies $I' \models q_{\Pi_S}$. Thus assume that
  there is a CQ $\exists \xbf \, q(\xbf)$ in $q_\Pi$ and a
  homomorphism $h$ from $q(\xbf)$ to $I$. We use $\exists \xbf \,
  q(\xbf)$ and $h$ to guide the choices in Step~1 and Step~2 of the
  construction of CQs in $q_{\Pi_S}$ to exhibit a CQ $p_0$ in
  $q_{\Pi_S}$ such that $p_0 \rightarrow I'$.

  We start with Step~1. As $\exists \xbf' \, q'(\xbf')$, we use the
  contraction of $\exists \xbf \, q(\xbf)$ obtained by identifying
  variables $x$ and $y$ whenever $h(x)=h(y)$. Thus, $h$ is an
  injective homomorphism from $q'(\xbf')$ to $I$. We next need to
  choose a partition of $q'(\xbf')$.  For every fact $R(\abf) \in I$,
  choose a fact $R_{p(\xbf_0)}(\bbf) \in I'$ that $R(\abf)$ was
  obtained from during the construction of $I$ and denote this fact
  with
  $\mu(R(\abf))$. % \footnote{In fact, the choice is deterministic for
    % facts involving a relation of arity at least two since $I'$ has
    % girth at least three.}
  Now let $q_1(\xbf_1),\dots,q_n(\xbf_n)$ be the
  partition of $q'(\xbf')$ obtained by grouping together two atoms
  $R_1(\ybf_1)$ and $R_2(\ybf_2)$ if and only if
  $\mu(R_1(h(\ybf_1)))=\mu(R_2(h(\ybf_2)))$. % Each partition contains
  % at most $k$ variables since $h$ is injective and $k$ is the maximum
  % arity of atoms in $I'$; moreover, distinct $q_i,q_j$ overlap in at
  % most one variable since $h$ is injective and $I'$ has girth at least
  % three. With each partition $q_i$, 
  Let $\mu(q_i)$ denote the (unique) value of $\mu$ for all the atoms
  in~$q_i(\xbf_i)$.

  Step~2 deals with each query $q_i(\xbf_i)$ separately. We choose the
  relation $R_{p(\zbf)}$ from $\mu(q_i)=R_{p(\zbf)}(\bbf)$, which
  clearly is in $\Sbf'_E$. We choose the tuple $\ybf$ of variables
  based on the tuple of individuals \bbf. Let $\bbf=b_1,\dots,b_n$.
  Then the $\ell$-th variable in $\ybf$ is $y$ if $h(y)=b_\ell$ (which is
  well-defined since $h$ is injective) and a fresh variable if there
  is no such $y$. This finishes the guiding process and thus gives
  rise to a query $p_0(\ubf)$ in $q_{\Pi_S}$.

  It remains to argue that $h$ can be extended to a homomorphism $h'$
  from $p(\ubf)$ to $I'$. Take a~$q_i(\xbf_i)$ and consider the
  corresponding atom $R_{p(\zbf)}(\ybf)$ in $p_0$. Then all the facts
  in $q_i(h(\xbf)) \subseteq I$ were obtained from the fact
  $\mu(q_i)=R_{p(\zbf)}(\bbf) \in I'$ during the construction of $I$.
  By construction of $\ybf$ from \bbf, we can extend $h$ to the fresh
  variables in \ybf so that $h(\ybf)=\bbf$ and thus
  $R_{p(\zbf)}(h(\ybf)) \in I'$. Doing this for all $q_i$ yields the
  desired $h'$. 

  %\medskip 
  Now for the cases $\Qmc \in \{\text{MDLog}, \text{DLog}
  \}$. We treat these cases in one since our construction preserves
  the width of Datalog-rewritings. In fact, this construction is
  very similar to the case $\Qmc = \text{UCQ}$, so we only give a
  sketch.

  For Point~1, let $\Gamma_{\Pi_S}$ be a Datalog-rewriting of $\Pi_S$.
  We construct a Datalog program $\Gamma_\Pi$ of the same width over
  EDB schema $\Sbf_E$ by modifying the EDB part of each rule body in
  the same way in which we had modified the UCQ-rewriting $q_{\Pi_S}$
  in the case $\Qmc = \text{UCQ}$: replace every EDB-atom
  $R_{q(\xbf)}(\ybf)$ with $q[\ybf/\xbf]$. We then have $I \models
  \Pi$ iff $I' \models \Pi_S$ (by Point~1 of
  Theorem~\ref{th:mmsnpToSimple}) iff $I' \models \Gamma_{\Pi_S}$ (by
  choice of $\Gamma_{\Pi_S}$) iff $I \models \Gamma_\Pi$. The latter
  is by construction of $I'$ and of $\Gamma_\Pi$. To prove it in more
  detail, it suffices to show that for every extension $J$ of $I$ to
  the IDB relations in $\Gamma_{\Pi_S}$ with corresponding extension
  $J'$ of $I'$, and every rule body $q$ in $\Gamma_{\Pi_S}$ which was
  translated into a rule body $q'$ in $\Gamma_\Pi$, we have $q
  \rightarrow J$ iff $q' \rightarrow J'$.  The arguments needed are as
  in the case $\Qmc = \text{UCQ}$.

  The proof of Point~2 can be adapted from UCQs to Datalog in an
  analogous way. Let $\Gamma_\Pi$ be a Datalog-rewriting of $\Pi$.  We
  construct a Datalog program $\Gamma_{\Pi_S}$ of the same width over
  EDB schema $\Sbf_{E'}$. The rules in $\Gamma_{\Pi_S}$ are obtained
  by taking a rule
  \[P_0(\xbf_0) \leftarrow P_1(\xbf_1) \wedge \cdots \wedge
  P_n(\xbf_n) \wedge q(\ybf)
  \]
  from $\Gamma_\Pi$, where the $P_i$ are IDB and $q(\ybf)$ is a CQ
  over schema $\Sbf_E$, converting $q(\ybf)$ into a CQ $q'(\ybf')$
  over schema $\Sbf'_E$ in two steps, in the same way in which a CQ
  over $\Sbf_E$ was converted into a CQ over $\Sbf'_E$ in the case
  $\Qmc = \text{UCQ}$, and then including in $\Gamma_{\Pi_S}$ the rule
  \[P_0(\xbf_0) \leftarrow
  P_1(\xbf_1) \wedge \cdots \wedge P_n(\xbf_n) \wedge q'(\ybf').
  \]
The crucial step in the correctness proof is to show that $I \models
\Gamma_\Pi$ implies $I' \models \Gamma_{\Pi_S}$ for any
$\Sbf'_E$-instance $I'$ and corresponding $\Sbf_E$-instance $I$.  The
arguments are again the same as in the case $\Qmc=\text{UCQ}$, the
main difference being that we need to consider extensions of $I$ and
$I'$ to IDB relations from $\Gamma_\Pi$ instead of working with $I$
and $I'$ themselves.
\end{proof}
Point~2 of Lemma~\ref{prop:rewrCorresp} only yields a rewriting of
$\Pi_S$ on $\Sbf'_E$-instances of high girth. We next show that, for
$\Qmc \in \{ \text{UCQ},\text{MDLog} \}$, the existence of a
\Qmc-rewriting on instances of high girth implies the existence of a \Qmc-rewriting
that works on instances of unrestricted girth. Whether the same is
true for $\Qmc = \text{Datalog}$ remains as an open problem.  We need
the following well-known lemma that goes back to Erd\"os and was
adapted to CSPs by Feder and Vardi. Informally, it says that every
instance can be `exploded' into an instance of high girth that behaves
similarly regarding homomorphisms.
\begin{lem} %[Erd\H{o}s, FV]
\label{lem:erdoes} 
For every instance $I$ and $g,s \geq 0$, there is an instance $I'$
(over the same schema) such that $I' \rightarrow I$, $I'$ has girth
exceeding $g$, and for every instance $T$ of size at
most $s$, we have $I \rightarrow T$ iff $I' \rightarrow T$.
\end{lem}
Feder and Vardi additionally show that $I'$ can be constructed by a
randomized polynomial time reduction that was later derandomized by Kun
\cite{DBLP:journals/combinatorica/Kun13}, but here we do not rely on
such computational properties.  Every CQ $q$ can be viewed as an
instance $I_q$ by using the variables as constants and the atoms as
facts. It thus makes sense to speak about tree decompositions of CQs
and about their treewidth, and it is clear what we mean by saying
that a CQ is a tree (that is, has girth $\omega$).
\begin{lem}
  \label{lem:FOgirthdontcare}
Let $S$ be a set of templates over schema $\Sbf_E$, $g \geq 0$, and
$\Qmc \in \{ \text{UCQ}, \text{MDLog} \}$. If coCSP($S$) is
\Qmc-rewritable on instances of girth exceeding~$g$, then it is
\Qmc-rewritable.
\end{lem}
\begin{proof}\ We start with $\Qmc = \text{UCQ}$.  Let $q_g$ be a UCQ
  that defines coCSP$(S)$ on instances of girth exceeding~$g$, and let
  $q$ be the UCQ that consists of all contractions of a CQ in $q_g$
  that are a tree CQ. We show that $q$ defines coCSP$(S)$ on
  unrestricted $\Sbf_E$-instances.

  Let $I$ be an $\Sbf_E$-instance. First assume that $I
  \not\rightarrow S$.  By Lemma~\ref{lem:erdoes}, there is an
  $\Sbf_E$-instance $I'$ of girth exceeding $g$ and also exceeding the
  number of variables in each CQ in $q_g$ and satisfying $I'
  \rightarrow I$ and $I' \not\rightarrow S$. Thus $I' \models q_g$,
  that is, there is a CQ $q'$ in $q_g$ and a homomorphism $h$ from
  $q'$ to $I'$. Let $q''$ be the contraction of $q'$ obtained by
  identifying variables $x$ and $y$ if $h(x)=h(y)$.  Thus, $h$ is an
  injective homomorphism from $q''$ to $I'$. Since the girth of $I'$
  exceeds the number of variables in $q''$, $q''$ must be a
  tree. Consequently, $q''$ is a CQ in $q$ and we have $I' \models
  q$. From $I' \rightarrow I$, we obtain $I \models q$.

  Now assume that $I \models q$. Then, there is a tree CQ $q'$ in
  $ q$ such that $q' \rightarrow I$. When we view $q'$ as an
  $\Sbf_E$-instance $I_{q'}$, then clearly $I_{q'} \models q_g$ and
  $I_{q'}$ has girth exceeding $k$. Thus, $q' \not \rightarrow S$, and
  from $q' \rightarrow I$ we obtain $I \not \rightarrow S$.

  %\smallskip

  Now for the case $\Qmc = \text{MDLog}$.  Let $\Gamma_g$ be an MDLog
  program that defines coCSP$(S)$ on instances of girth exceeding $g$.
  Let $\Gamma$ be the program obtained from $\Gamma_g$ by replacing
  every rule $P(x) \leftarrow q(\xbf)$ with all rules $P(x) \leftarrow
  q'(\xbf')$ such that $q'(\xbf')$ is a tree CQ that is a contraction
  of $q(\xbf)$. We show that $\Gamma$ is an MDLog-definition of
  coCSP$(S)$ on instances of unrestricted girth.

  Let $I$ be an $\Sbf_E$-instance. First assume that $I
  \not\rightarrow S$.  By Lemma~\ref{lem:erdoes}, there is an
  $\Sbf_E$-instance $I'$ whose girth exceeds $g$ and also exceeds the
  diameter of $\Gamma_g$ and that satisfies $I' \rightarrow I$ and $I'
  \not\rightarrow S$. The latter yields $I' \models \Gamma_g$. It
  remains to show that this implies $I' \models \Gamma$ since with $I'
  \rightarrow I$, this yields $I \models \Gamma$ as required.

  To show that $I' \models \Gamma$ follows from $I' \models \Gamma_g$,
  it suffices to show that all IDB facts derived by $\Gamma_g$
  starting from $I'$ are also derived by $\Gamma$. Thus let $J'$ be an
  extension of $I'$ to the IDBs in $\Gamma_g$. It is enough to show
  that when a single application of a rule from $\Gamma_g$ in $J'$
  yields an IDB atom $P(a)$, then $\Gamma$ can derive the same atom.
  The former is the case only if $\Gamma_g$ contains a rule $P(x)
  \leftarrow q(\xbf)$ such that there is a homomorphism $h$ from
  $q(\xbf)$ to $J'$ with $h(x)=a$. Let $q'(\xbf')$ be the contraction
  of $q(\xbf)$ obtained by identifying variables $x$ and $y$ when
  $h(x)=h(y)$.  Since the girth of $I'$ exceeds the diameter of
  $\Gamma_g$, $q'(\xbf')$ is a tree. Thus, $\Gamma$ contains the rule
  $P(x) \leftarrow q'(\xbf')$ and the application of this rule in $J'$
  enabled by $h$ yields $P(a)$. We have thus shown $I' \models \Gamma$
  and are done.

  Now assume that $I \models \Gamma_g$. Then there is a proof tree for
  $\mn{goal}()$ from $I$ and~$\Gamma_g$, see
  \cite{Abiteboul:1995:FDL:551350} for details. From that tree, we can
  read off an $\Sbf_E$-instance $O$ such that $O \rightarrow I$, $O
  \models \Gamma_g$, and, since $\Gamma_g$ is monadic and only
  comprises rules with tree-shaped bodies, $O$ is a tree. Thus, $O$
  has girth exceeding $g$ and from $O \models \Gamma_g$ we get $O
  \not\rightarrow S$. But with $O \rightarrow I$, this yields $I \not
  \rightarrow S$ as required.
\end{proof}
Putting together Theorems~\ref{th:mmsnpToSimple}
and~\ref{th:simpleToCsp}, Proposition~\ref{prop:ross},
and Lemmas~\ref{prop:rewrCorresp} and~\ref{lem:FOgirthdontcare},
we obtain the following reductions of rewritability of Boolean MDDLog
programs to CSP rewritability.
\begin{prop}
\label{prop:rewrCorrespARGH}  
Every Boolean MDDLog program $\Pi$ can be converted into a set of
templates $S_\Pi$ such that 
\begin{enumerate}
\item $\Pi$ is \Qmc-rewritable iff coCSP$(S_\Pi)$ is \Qmc-rewritable for
  every $\Qmc \in \{ \text{FO}, \text{UCQ}, \text{MDLog} \}$;
\item every \Qmc-rewriting of $\Pi$ can be effectively translated into
  a \Qmc-rewriting of coCSP$(S_\Pi)$ and vice versa, for
  every $\Qmc \in \{ \text{UCQ}, \text{MDLog} \}$.
\item $|S_\Pi| \leq 2^{2^{p(n)}}$ and  $|T| \leq 2^{2^{p(n)}}$ for
  each $T \in S_\Pi$, $n$ the size of $\Pi$ and $p$ a polynomial.
\end{enumerate}
The construction takes time polynomial in $\sum_{T \in S_\Pi} |T|$. 
\end{prop}
FO-rewritability of CSPs (and their complements) is {\sc NP}-complete
\cite{DBLP:journals/lmcs/LaroseLT07} and it was observed in
\cite{DBLP:journals/tods/BienvenuCLW14} that the upper bound lifts to
generalized CSPs. % An analysis of the algorithm yields a
% slightly improved bound, which is need below: there is a
% nondeterministic algorithm that decides FO-rewritability of the
% generalized CSP induced by a given set of templates $S$ and runs in
% time polynomial in $\sum_{T \in S} |\adom(T)|$.
MDLog-rewritability of coCSPs is {\sc NP}-hard and in 
\ExpTime \cite{Chen2017}. We show in
Appendix~\ref{app:mdlogexp} that also this upper bound lifts to
generalized coCSPs.
% ; again, we establish the slightly stronger result
% that there is an algorithm with runtime exponential only in
% $\sum_{T \in S} |\adom(T)|$.
Together with
Proposition~\ref{prop:rewrCorrespARGH},
this yields the upper bounds in the following theorem. The
lower bounds are from 
\cite{KR-submitted}. 
\begin{thm}
\label{th:decidBoolean}
  For Boolean MDDLog programs and the complement of MMSNP sentences,
  \begin{enumerate}
  \item FO-rewritability (equivalently: UCQ-rewritability) is
    2\NExpTime-complete;
  \item MDLog-rewritability is in 3\ExpTime (and 2\NExpTime-hard).
  \end{enumerate}
\end{thm}

\section{Shape of Rewritings, Obstructions, Explosion}
\label{sect:shapeObstrExplosion}

In the FO case, it is possible to extract from the approach in the
previous section an algorithm that computes actual rewritings, if they
exist. However, that algorithm is hardly practical.  An important
first step towards the design of more practical algorithms that
compute rewritings (in an exact or in an approximative way) is to
analyze the shape of rewritings. In fact, both FO- and
MDLog-rewritings of coCSPs are known to be of a rather restricted
shape, far from exploiting the full expressive power of the target
languages. In this section, we establish corresponding results for
Boolean MDDLog. This topic is closely related to the theory of
obstructions, so we also establish connections between the
rewritability of MMSNP sentences and natural obstruction
sets. Finally, we observe an MMSNP counterpart of
Lemma~\ref{lem:erdoes}, the fundamental `explosion' lemma for CSPs.

The following summarizes our results regarding the shape of
rewritings.
\begin{thm}
\label{th:rewrshape}   
Let $\Pi$ be a Boolean MDDLog program of diameter $k$. Then
\begin{enumerate}

\item if $\Pi$ is FO-rewritable, then it has a UCQ-rewriting in which 
  each CQ has treewidth $(1,k)$;

\item if $\Pi$ is MDLog-rewritable, then it has an MDLog-rewriting
  of diameter $k$.

\end{enumerate}
\end{thm}
\begin{proof}\ We analyze the proof of Lemma~\ref{prop:rewrCorresp}
  and use known results from CSP. In fact, any FO-rewritable coCSP has
  a UCQ-rewriting that consists of tree CQs
  \cite{DBLP:journals/jct/NesetrilT00}, and thus the same holds for
  simple Boolean MDDLog programs. If we convert such a rewriting of
  $\Pi_S$ into a rewriting of $\Pi$ as in the proof of
  Lemma~\ref{prop:rewrCorresp}, we obtain a UCQ-rewriting in which
  each CQ has treewidth~$(1,k)$. For Point~2 of
  Theorem~\ref{th:rewrshape}, one uses the proof of
  Lemma~\ref{prop:rewrCorresp} and the known fact that every
  MDLog-rewritable CSP has an MDLog-rewriting in which every rule body
  comprises at most one EDB atom, see e.g.\ the proof of Theorem~19 in
  \cite{DBLP:journals/siamcomp/FederV98}.
\end{proof}
In a sense, the concrete bound $k$ in Points~1 and~2 of
Theorem~\ref{th:rewrshape} is quite remarkable. Point~2 says,
for example, that when eliminating disjunctions from a Boolean
MDDLog program, it is never necessary to increase the diameter!

%\smallskip

We now consider obstructions. An \emph{obstruction set} $\Omc$ for a
CSP template $T$ over schema $\Sbf_E$ is a set of instances over the
same schema such that for any $\Sbf_E$-instance $I$, we have
$I \not\rightarrow T$ iff $O \rightarrow I$ for some $O \in \Omc$.
The elements of \Omc are called \emph{obstructions}.  A lot is known
about CSP obstructions. For example,
% %
% \begin{description}
%
% \item[(o1)]
  $T$ is FO-rewritable if and only if it has a finite 
  obstruction set \cite{Atserias08} if and only if it has a finite 
  obstruction set that consists of finite trees 
  \cite{DBLP:journals/jct/NesetrilT00}, and 
%
% \item[(o2)]
  $T$ is MDLog-rewritable if and only if it has a (potentially
  infinite) obstruction set that consists of finite trees
  \cite{DBLP:journals/siamcomp/FederV98}.
%
%\end{description}
%
  Here we consider obstruction sets for MMSNP, defined in the obvious
  way: an \emph{obstruction set} $\Omc$ for an MMSNP sentence $\theta$
  over schema $\Sbf_E$ is a set of instances over the same schema such
  that for any $\Sbf_E$-instance $I$, we have $I \not\models \theta$
  iff $O \rightarrow I$ for some $O \in \Omc$. This should not be
  confused with colored forbidden patterns used to characterize MMSNP
  in \cite{DBLP:journals/siamcomp/MadelaineS07}. The following
  characterizes FO-rewritability of MMSNP sentences in terms of
  obstruction sets.
\begin{cor}
\label{cor:FO}
For every MMSNP sentence $\theta$, the following are equivalent:
  \begin{enumerate}
  \item $\theta$ is FO-rewritable;
  \item $\theta$ has a finite obstruction set;
  \item $\theta$ has a finite set of finite obstructions of treewidth
    $(1,k)$.
  \end{enumerate}
\end{cor}
\noindent Corollary~\ref{cor:FO} follows from Point~1 of
Theorem~\ref{th:rewrshape} and the straightfoward observations that an
MMSNP sentence $\theta$ is FO-rewritable iff $\neg \theta$ is (which
is equivalent to a Boolean MDDLog program) and that every finite
obstruction set \Omc for $\theta$ gives rise to a UCQ-rewriting
$\bigvee \Omc$ of $\neg \theta$ and vice versa. We now turn to
MDLog-rewritability.
\begin{prop}
\label{prop:mdlogob}
  Let $\theta$ be an MMSNP sentence of diameter $k$. Then
  $\neg \theta$ is MDLog-rewritable iff $\theta$ has a set of 
  obstructions (equivalently: finite obstructions) that are of treewidth $(1,k)$.
\end{prop}
\begin{proof}\
  The ``only if'' direction is a consequence of Point~2 of
  Theorem~\ref{th:rewrshape} and the fact that, for any Boolean
  monadic Datalog program $\Pi \equiv \neg \theta$ of diameter $k$
  over EDB schema $\Sbf_E$, a proof tree for $\mn{goal}()$ from an
  $\Sbf_E$-instance $I$ and~$\Pi$ gives rise to a finite
  $\Sbf_E$-instance $J$ of treewidth $(1,k)$ with $J \rightarrow I$.  The desired
  obstruction set for $\neg \theta$ is then the set of all these $J$.
  The ``if'' direction is a consequence of Theorem~5 in
  \cite{DBLP:journals/jcss/BodirskyD13}. 
\end{proof}
We remark that the results in \cite{DBLP:journals/jcss/BodirskyD13} almost
give Proposition~\ref{prop:mdlogob}, but do not seem to deliver any
concrete bound on the parameter $k$ of the treewidth of obstruction
sets.

%\smallskip

We close with observing an MMSNP counterpart of the `explosion'
Lemma~\ref{lem:erdoes}, first giving a preliminary. Let $I$ be an
instance over some schema $\Sbf_E$. A \emph{$(1,k)$-decomposition} of
$I$ is a pair $(V, (I_v)_{v \in V})$ where $V$ is a set of indices and
$(I_v)_{v \in V}$ is a partition of $I$ such that for all distinct
\mbox{$v,v' \in V$}, $|\adom(I_v) \cap \adom(I_v')| \leq 1$ and
$|\adom(I_v)| \leq k$.  Thus, a $(1,k)$-decomposition
$D=(V, (I_v)_{v \in V})$ decomposes $I$ into parts of size at most $k$
and with little overlap.  These parts can be viewed as the facts of an
instance $I_D$ over an aggregation schema $\Sbf'_E$ defined by the
relations $R_{q_v(\xbf)}$ where $q_v(\xbf)$ is $I_v$ viewed as a CQ,
that is,
\[
I_D=\{R_{q_v(x)}(\adom(I_v))\mid v\in V \}
\]
where we assume some fixed (but otherwise irrelevant) order on the
elements of each $\adom(I_v)$. Now, we say that $I$ has
\emph{$(1,k)$-decomposition girth} $g$ if $g$ is the supremum of the
girths of $I_D$, for all $(1,k)$-decompositions $D$ of $I$. It can be
shown that $I$ has $(1,k)$-decomposition girth $\omega$ if and only if
it has treewidth $(1,k)$. 

Here comes the announced MMSNP counterpart
of Lemma~\ref{lem:erdoes}.
\begin{lem}
\label{lem:erdoesdlog} 
For every instance $I$ and $g\geq s>0$, and every MDDLog program $\Pi$ of diameter at most $s$, there is an instance $J$
(over the same schema) such that $J \rightarrow I$, $J$ has $(1,s)$-decomposition girth exceeding $g$, and $I \models \Pi$ iff $J \models \Pi$.
\end{lem}
\begin{proof}\ Let $\Pi$ be a Boolean MDDLog program of diameter $k \leq s$ over EDB schema $\Sbf_E$. 
By Theorems~\ref{th:mmsnpToSimple}  and~\ref{th:simpleToCsp}, there is a $k$-aggregation schema
$\Sbf'_E$ and a set of templates $S_\Pi$ over $\Sbf'_E$ such that:
 \begin{enumerate}
  \item for any $\Sbf_E$-instance $I$ with corresponding
  $\Sbf'_E$-instance $I'$, $I \models \Pi$ iff $I'  \not\rightarrow S_\Pi$;
  \item for any $\Sbf'_E$-instance $I'$ whose girth exceeds $k$ with corresponding
    $\Sbf_E$-instance $I$, $I'
    \not\rightarrow S_\Pi$ iff $I \models \Pi$.
  \end{enumerate}
  Let $I$ and $I'$ be an $\Sbf_E$-instance and its corresponding
  $\Sbf'_E$-instance. Furthermore, let $J'$ be the $\Sbf'_E$-instance
  obtained from $I'$ by applying Lemma~\ref{lem:erdoes} with
  $s=\max\{ |T| \mid T \in S_\Pi \}$ and $g$ as given. Then $J'
  \rightarrow I'$, $J'$ has girth exceeding $g$, and $J'
  \rightarrow S_\Pi$ iff $I' \rightarrow S_\Pi$ iff $I
  \not\models \Pi$.
%  From
%  Point~1 above, we obtain $I' \not\rightarrow S_\Pi$. 
%  Applying
%  Lemma~\ref{lem:erdoes} with $s=\max\{ |T| \mid T \in S_\Pi \}$, we
%  obtain an $\Sbf'_E$-structure $J'$ such that $J' \rightarrow I'$,
%  $J'$ has girth exceeding $g$, and $J' \in \text{coCSP}(S_\Pi)$. Thus
%  $J' \not\rightarrow S$. 
  Let $J$ be the $\Sbf_E$-instance corresponding to $J'$. As $J'$ has
  girth exceeding $k$, Point~(2) above yields $J \models \Pi$ iff
  $J' \not \rightarrow S_\Pi$. In summary, we thus obtain
  $I \models \Pi$ iff $J \models \Pi$.

%Point~2 above yields $J \models \Pi$. 

  It thus remains to show that $J$ has $(1,s)$-decomposition girth
  exceeding $g$ and that $J \rightarrow I$. The former is witnessed by
  the $(1,k)$-decomposition $D=(V,(I_v)_{v \in V})$ of $J$
  obtained by using as $V$ the facts of $J'$ and as $I_v$ the set of
  facts obtained from fact $v$ during the construction of $J$. % It can be
  % verified that $I_D$ can be obtained from $J'$ by renaming relations.

  As the last step, we argue that $J \rightarrow I$ follows from $J'
  \rightarrow I'$, and that in fact any homomorphism $h$ from
  $J'$ to $I'$ is also a homomorphism from $J$ to $I$. Thus let $h$ be
  such a homomorphism. For any fact $R(a_{i_1},\dots,a_{i_\ell})$ in
  $J$, there is a fact $R_{q(x_1,\dots,x_n)}(a_1,\dots,a_m)$ in $J'$
  such that $R(x_{i_1},\dots,x_{i_\ell}) \in q_i(x_1,\dots,x_n)$.  We
  have $R_{q(x_1,\dots,x_n)}(h(a_1),\dots,h(a_m)) \in I'$.  By
  definition of $I'$, this means $R(h(a_{i_1}),\dots,h(a_{i_\ell}))
  \in I$ and we are done.
\end{proof}
We believe that Lemma~\ref{lem:erdoesdlog} can be useful in many
contexts, saving a detour via CSPs. For example, it enables an
alternative proof of Theorem~\ref{th:rewrshape}. We illustrate this
for Point~1. We can start with a UCQ-rewriting $q$ of an MDDLog
program $\Pi$ of diameter $k$ and show that the UCQ $q_t$ that
consists of all CQs of treewidth $(1,k)$ that are a contraction of a
CQ in $q$ must also be a rewriting of $\Pi$: take an instance $I$ that
makes $\Pi$ true, use Lemma~\ref{lem:erdoesdlog} to transform $I$ to
an $I'$ of girth exceeding $k$ and also exceeding the size of any CQ
in $q$ such that $I' \models \Pi$ and $I' \rightarrow I$, observe that
$I' \models q$ and that a homomorphism from any CQ $p$ in $q$ to $I'$
gives rise to a homomorphism from a CQ $p'$ in $q_t$ to $I'$, and
derive $p' \rightarrow I$ from $I' \rightarrow I$.

\section{Datalog-Rewritability of Boolean MDDLog Programs
and Canonical Datalog Programs}
\label{sect:dlog-rewr}

We consider rewriting Boolean MDDLog programs into Datalog programs,
making two contributions. First, we show that Datalog-rewritability is
decidable for programs that have equality, a condition that is defined
in detail below. For programs that do not have equality, the same
construction yields a procedure that is sound, but whose completeness
remains an open problem. And second, we give a new and direct
construction of canonical Datalog-rewritings of Boolean MDDLog
programs (equivalently: the complements of MMSNP sentences), bypassing
the construction of infinite templates
\cite{DBLP:journals/jcss/BodirskyD13} which involves the application
of a non-trivial construction due to Cherlin, Shelah, and Shi
\cite{CherlinEtAl99}. This construction is potentially useful even
though it is yet unknown whether Datalog-rewritability of MDDLog
programs $\Pi$ is decidable (for programs that do not have equality):
when $\Pi$ is not rewritable, then the canonical Datalog-rewriting is
the best possible approximation of $\Pi$ in terms of a Datalog program
(of given width and diameter).

% MDDLog
% programs obtained from translating OMQs, though, will typically not
% have equality.

\subsection{Datalog-Rewritability of Boolean MDDLog Programs}

A CSP template $T$ \emph{has equality} if its EDB schema includes the
distinguished binary relation \mn{eq} and $T$ interprets \mn{eq} as
the relation $\{ (a,a) \mid a \in \adom(T) \}$. Thus, $\mn{eq}$ is an
extremely natural kind of constraint: a fact $\mn{eq}(a,b)$ in the
input instance means that $a$ and $b$ must be mapped to the same
template element; spoken from the perspective of constraint
satisfaction, they are variables that must receive the same value.

In accordance with the above, we say that an MDDLog program $\Pi$
\emph{has equality} if its EDB schema includes the distinguished
binary relation \mn{eq}, $\Pi$ contains the rules
\[
  P(x) \wedge \mn{eq}(x,y) \rightarrow P(y)  \quad \text{ and } \quad
  P(y) \wedge \mn{eq}(x,y) \rightarrow P(x)
\]
for each IDB relation $P$, and these are the only rules that mention
\mn{eq}. Thus, a fact $\mn{eq}(a,b)$ in the input instance says that
the same IDB relations can be derived by $\Pi$ for $a$ and for~$b$.
It can be verified that when an MDDLog program that has equality is
converted into a generalized CSP based on a set of templates $S_\Pi$
according to Theorems~\ref{th:mmsnpToSimple} and~\ref{th:simpleToCsp}
(using the concrete constructions in the appendix), then all templates
in $S_\Pi$ have equality.

%\smallskip

We aim to show decidability of the Datalog-rewritability of MDDLog
programs that have equality following the strategy that we have used
for rewritability into FO and into MDLog in Section~\ref{sec:dec}. We
thus need a counterpart of Lemma~\ref{lem:FOgirthdontcare}, that is, we
have to show that for all templates $T$ that have equality,
Datalog-rewritability of coCSP$(T)$ on instances of high girth implies
unrestricted Datalog-rewritability.  It is here that having equality
is an advantage. In particular, every input instance for coCSP$(T)$
can be made high girth preserving (non-)homomorphisms to $T$ by
introducing additional $\mn{eq}$-facts. This is similar in spirit to
the explosion Lemma~\ref{lem:erdoes}, but the construction is much
simpler than in the proof of that lemma. We next make it explicit.

Let $I$ be an $\Sbf_E$-instance and let $g \geq 0$. We use
$\mn{pos}(I)$ to denote the set of pairs $(R(\abf),i)$ such that
$R(\abf) \in I$ and $i \in \mn{pos}(R)$. In what follows, for any
tuple of constants $\abf$, we use $a_i$ to denote its $i$-th
component.  Reserve fresh constants as follows:
  \begin{itemize}
    
    \item a constant $b_p$, for all $p=(R(\abf),i) \in \mn{pos}(I)$;

    \item $g$ constants 
      $b_{p,p',1},\dots,b_{p,p',g}$, for all $p,p'=(R(\abf),i),(R'(\abf'),i') \in \mn{pos}(I)$ with
      \mbox{$a_i = a'_{i'}$}.
    
  \end{itemize}
 Define an instance $I^g$ that consists of the following facts:
  \begin{enumerate}

  \item for every $R(\abf) \in I$ with $R$ of arity $n$,
    the fact $R(b_{p_1},\dots,b_{p_n})$ where $p_i=(R(\abf),i)$;

  \item for all distinct $p,p'=(R(\abf),i),(R'(\abf'),i') \in \mn{pos}(I)$ with
    $a_i = a'_{i'}$, the facts 
    \[
    \mn{eq}(b_p,b_{p,p',1}),
    \mn{eq}(b_{p,p',1},b_{p,p',2}),\dots,
    \mn{eq}(b_{p,p',{g-1}},b_{p,p',g}),
    \mn{eq}(b_{p,p',g},b_{p'}).
    \]

  \end{enumerate}
  Observe that $I^g$ has girth exceeding
  $g$. Moreover, it satisfies the following crucial property.
\begin{lem}
  For every CSP template $T$ over $\Sbf_E$ that has equality, $I^g
  \rightarrow T$ iff $I \rightarrow T$.
\end{lem}
\begin{proof}\ Let $T$ be a template over $\Sbf_E$ that has equality.
  We have to show that there is a homomorphism $h$ from $I$ to $T$ iff
  there is a homomorphism $h_g$ from $I^g$ to $T$.  In fact, $h_g$ can
  be obtained from $h$ by setting $h_g(b_p)=h_g(b_{p,p',j})=h(a_i)$
  when $p=(R(\abf),i)$; conversely, $h$ can be obtained from $h_g$ by
  setting $h(a_i)=h_g(b_p)$ when $p=(R(\abf),i)$---the latter is
  well-defined by construction of $I^g$ and since $\mn{eq}$ is
  interpreted as the reflexive relation in~$T$.
\end{proof}
We are now ready to establish the announced counterpart of
Lemma~\ref{lem:FOgirthdontcare}.
%
% A set $S$ of CSP templates \emph{has equality}
% if every $T \in S$ has. 
%
\begin{lem}
\label{lem:girthLemForDLog}
  Let $S$ be a set of templates over schema $\Sbf_E$ that have
  equality, and let $g \geq 0$.  If coCSP($S$) is DLog-rewritable on
  instances of girth exceeding $g$, then it is DLog-rewritable.
\end{lem}
\begin{proof}\ Assume that coCSP$(S)$ is DLog-rewritable on
  instances of girth exceeding~$g$ and let $\Gamma$ be a concrete
  rewriting. We construct a Datalog program $\Gamma'$ such that for
  any $\Sbf_E$-instance $I$, $I \models \Gamma'$ iff $I^g \models
  \Gamma$. Clearly, $\Gamma'$ is then a rewriting of coCSP$(S)$ 
  on instances of unrestricted girth. 
  
  We aim to construct $\Gamma'$ such that it mimics the execution of
  $\Gamma$ on $I^g$, despite being executed on $I$. One challenge is
  that the domains of $I$ and $I^g$ are not identical. %  Note that every
  % constant in $\adom(I^g)$ that is of the form $b_p$ is identified by
  % a fact $R(\abf)$ in $I$ and a number $i$ that is bounded by the
  % arity of $R$; likewise, every constant of the form $b_{p,p',j}$ is
  % identified by two facts $R(\abf),R'(\abf') \in I$, numbers $i,i'$
  % that identify positions in which $R(\abf)$ and $R'(\abf')$ share a
  % constant and that are bounded by the maximum arity of $R$ and $R'$,
  % and th enumber $j$ which is bounded by $g$.
  In $\Gamma'$, the IDB relations of $\Gamma$ need to be adapted to
  reflect this change of domain, and so do the rules. Let $m$ be the
  maximum arity of any relation in $\Sbf_E$.  Every IDB relation $P$
  of $\Gamma$ gives rise to a \emph{set} of IDB relations in
  $\Gamma'$. In fact, every position of $P$ can be replaced either
  with
  \begin{enumerate}

  \item $\ell$ positions, for some $\ell \leq m$, reflecting the case
    that the position is filled with a constant $b_p$ where
    $p=(R(\abf),i)$ with $R$ $\ell$-ary; or with 

  \item $\ell + \ell'$ positions, for some $\ell,\ell' \leq m$,
    reflecting the case that the position is filled with a constant
    $b_{p,p',j}$ where $p=(R(\abf),i)$ and $p'=(R'(\abf'),i')$, with
    $R$ $\ell$-ary and $R'$ $\ell'$-ary.

  \end{enumerate}
  In Case~1, the $\ell$ positions store the constants in \abf. The
  symbol $R$ and the number $i$ from $p$ also need to be stored, which
  is done as an annotation to the IDB relation. In Case~2, the first
  $\ell$ positions store the constants in \abf while the latter
  $\ell'$ positions store the constants in~$\abf'$; we additionally
  need to store the symbols $R$ and $R'$, the numbers $i$ and $i'$
  from $p$ and~$p'$, and the number $j$, which is again done by
  annotation of the IDB relation.

  Let us make this formal. The IDB relations of
  $\Gamma'$ take the form $P^{\mu}$ where $P$ is an IDB relation of
  $\Gamma$ and $\mu$ is a function from $\mn{pos}(P)$ to
  \[
  \Omega := (\Sbf_E \times \{1,\dots,m\}) \cup
  (\Sbf_E \times \{1,\dots,m\} \times \Sbf_E \times 
  \{1,\dots,m\} \times \{ 1,\dots,g\})
  \]
  such that if $\mu(\ell)=(R,i)$, then $i \in \mn{pos}(R)$ and if
  $\mu(\ell)=(R,i,R',i',j)$, then $i \in \mn{pos}(R)$ and $i' \in
  \mn{pos}(R')$.  The arity of $P^\mu$ is $\sum_{\ell=1..\mn{pos}(P)}
  q_\ell$ where $q_\ell$ is the arity of $R$ if $\mu(\ell)=(R,i)$ and
  $q_\ell$ is the sum of the arities of $R$ and $R'$ if
  $\mu(\ell)=(R,i,R',i',j)$.  In the construction of~$\Gamma'$, we
  manipulate the rules of $\Gamma$ to account for this change in the
  IDB schema.  We can assume w.l.o.g.\ that $\Gamma$ is closed under
  contractions of rules. Let
  \[ 
  \begin{array}{rcl}
     P_0(\xbf_0) &\leftarrow& P_1(\xbf_1) \wedge \cdots \wedge
     P_{\ell_1}(\xbf_{\ell_1}) \,
     \wedge \\[1mm]
     && R_1(\ybf_1) \wedge \cdots \wedge 
     R_{\ell_2}(\ybf_{\ell_2})  \,
     \wedge \\[1mm]
     &&
     \mn{eq}(z_{1,1},z_{1,2}) \wedge \cdots \wedge 
     \mn{eq}(z_{\ell_3,1},z_{\ell_3,2})
  \end{array}
  \]
  be a rule in $\Gamma$ where $P_0,\dots,P_{\ell_1}$ are IDB and
  $R_1,\dots,R_{\ell_2}$ are EDB (possibly the distinguished \mn{eq}
  relation), such that % the following condition is satisfied:
  \begin{itemize}[align=left]

   \item[($*$)] every variable occurs at most once in 
     $R_1(\ybf_1) \wedge \cdots \wedge 
     R_{\ell_2}(\ybf_{\ell_2})$.

  %\item[($*$)] $i \neq j$ implies that $\ybf_i$ and $\ybf_j$ are
    %disjoint.% ,\footnote{If the latter condition is not satisfied, then
      % all rule appications in $I^g$ will map $R_i(\ybf_i)$ and
      % $R_j(\ybf_j)$ to the same atom (which requires $R_i=R_j$,
      % otherwise no rule application is possible at all); consequently,
      % we can identify variables so that $R_i(\ybf_i)$ and
      % $R_j(\ybf_j)$ become the same atom, which we assume to be
      % included only once in the rule achieving Condition ($*$).}

  \end{itemize}
  Note that it might be possible to write a single rule from $\Gamma$
  in the above form in more than one way because $\mn{eq}$-atoms can
  be placed in the second line or in the third line; we then consider
  all possible ways. Informally, this choice corresponds to the
  decision whether the $\mn{eq}$-atom is mapped to an $\mn{eq}$-fact
  in $I^g$ that comes from an $\mn{eq}$-fact in $I$ (Point~1 of the
  definition of $I^g$) or to a freh $\mn{eq}$-fact (Point~2 of the
  definition of $I^g$). Also note that rules that do not satisfy ($*$)
  can be ignored since they never apply in $I^g$.

%\smallskip
  
  %
  % where $P_0,\dots,P_{\ell_1}$ are IDB and $R_1,\dots,R_{\ell_2}$ are EDB
  % and distinct from \mn{eq}.
  %
  % ohne das hier kriegen wir ein problem: Gamma' ist zu stark. Z.B.
  % wuerde die Regel goal() <- R(x1,x2) and S(x1,x2') ohne Aenderung
  % uebernommen und waehrend sie in I^g NIE feuern koennte, kann
  % sie in I problemlos feuern.
  %
  % Die Variablenidentifikation wird gebraucht, weil R(x1,x2) and R(x1,x2')
  % auftauchen koennte, aber im Prinzip harmlos ist; zumindest dann,
  % wenn x2 uns x2' auf dieselbe Konstante gemappt werden.
  %
  Let $\xbf$ be the variables in the rule, and let $\delta:\xbf
  \rightarrow \Omega$ be such that the following conditions are
  satisfied:
  \begin{enumerate}

  \item for each $R_i(\ybf_i)$ with $\ybf_i=y_1 \cdots y_k$, we have 
    $\delta(y_j) = (R_i,j)$ for all $j$;
   \item for each $\mn{eq}(z_{i,1},z_{i,2})$, one of the following is
     true for some $R,i,R',i',j$:
    \begin{enumerate}

%    \item $\delta(z_{i,1}),\delta(z_{i,2}) \in \Sbf_E \times \{1,\dots,m\}$;

    \item $\delta(z_{i,1})=(R,i)$ and $\delta(z_{i,2})=(R,i,R',i',1)$;

    \item $\delta(z_{i,1})=(R,i,R',i',g)$ and
      $\delta(z_{i,2})=(R',i')$;

    \item $\delta(z_{i,1})=(R,i,R',i',j)$ and
      $\delta(z_{i,2})=(R,i,R',i',j \pm 1)$.

    % \item $\delta(z_{i,1})=(R,i,R',i',j)$ and
    %   $\delta(z_{i,2})=(R,i,R',i',j-1)$.

    \end{enumerate}
  \end{enumerate}
  With each variable $x$ in \xbf, we associate a tuple $\ubf_x$ of
  distinct variables. If $\delta(x)$ is of the form $(R,i)$, then the
  length of $\ubf_x$ matches the arity of $R$ and $\ubf_x$ is called a
  \emph{variable block}. If $\delta(x)$ is of the form
  $(R,i,R',i',j)$, then the length of $\ubf_x$ is the sum of the
  arities $n$ and $n'$ of $R$ and $R'$; the first $n$ variables in
  $\ubf_x$ are then also called a \emph{variable block}, and so are
  the last $n'$ variables. Variable blocks will either be disjoint or
  identical. Identities are minimized such that the following
  conditions are satisfied:
  \begin{enumerate}[align=left]

  \item[(I1)] if $x$ occurs in some $\ybf_i$, then $\ubf_x=\ybf_i$;

  \item[(I2)] if Case~2a applies to $\mn{eq}(z_{i,1},z_{i,2})$, then
    $\ubf_{z_{i,1}}$ is identical to the first variable block in~$\ubf_{z_{i,2}}$;

  \item[(I3)] if Case~2b applies to $\mn{eq}(z_{i,1},z_{i,2})$, then
    $\ubf_{z_{i,2}}$ is identical to the second variable block in~$\ubf_{z_{i,1}}$;

  \item[(I4)] if Case~2c %or~2d 
    applies to $\mn{eq}(z_{i,1},z_{i,2})$,
    then the first variable blocks of $\ubf_{z_{i,1}}$ and $\ubf_{z_{i,2}}$
    are identical, and so are the second variable blocks.

  \end{enumerate}
  Regarding (I1), note that $x$ cannot occur in more than one $\ybf_i$
  because of ($*$), thus the condition can always be satisfied
  `without conflicts'.
  %
  % To see why this is needed, consider ... <- S(x1,x) and R(x1,x2). 
  % assume \delta(x)=(R',R'',i',i'',j) with R' and R'' of arity three.
  % Should then translate to 
  % ... <- S^mu(x1,x2,u1,u2,u3,u4,u5,u6) and R(x1,x2)
  %
  % With each variable $x$ in \xbf such that $\delta(x)=(R,R',i,i',j)$,
  % associate a tuple $\ubf_x$ of distinct variables whose length is
  % the sum of the arities of $R$ and $R'$. 
  Then include in $\Gamma'$
  the rule
  \[ 
  \begin{array}{rcl}
     P^{\mu_0}_0(\xbf'_0) &\leftarrow& P^{\mu_1}_1(\xbf'_1) \wedge \cdots \wedge
     P^{\mu_{\ell_1}}_{\ell_1}(\xbf'_{\ell_1}) \,
     \wedge \\[1mm]
     && R_1(\ybf_1) \wedge \cdots \wedge 
     R_{\ell_2}(\ybf_{\ell_2}) \,
     \wedge \\[1mm]
     &&
     W
  \end{array}
  \]
  such that
  \begin{enumerate}[align=left]

  \item[(R1)] if the $k$-th component in $\xbf_0$ is $x$, then
    $\mu_i(k)=\delta(x)$;

  \item[(R2)] $\xbf'_i$ is obtained from $\xbf_i$ by replacing each variable
    $x$ with $\ubf_x$;
            
  \item[(R3)] $W$ contains the following atoms:
    \begin{itemize}

    \item for each variable $x \in \xbf$ with $\delta(x)$ of the form
      $(R,i)$, an atom $R(\wbf)$ where the $i$-th component of $\wbf$
      is $x$ and all other variables are distinct and fresh;

    \item for each variable $x \in \xbf$ with $\delta(x)$ of the form
      $(R,i,R',i',j)$, atoms $R(\wbf), R'(\wbf')$ where the $i$-th
      component of $\wbf$ and the $i'$-th component of $\wbf'$ is $x$
      and all other variables are distinct and fresh.

    \end{itemize}
  \end{enumerate}
  As an example, consider the following rule in $\Gamma$:
  \[
      \mn{goal}() \leftarrow P(x_1,x_2) \wedge R(x_1,x_2) \wedge \mn{eq}(x_2,x_3)
  \]
  where $R$ is EDB and $P$ IDB, and let $\delta(x_1)=(R,1)$,
  $\delta(x_2)=(R,2)$, and $\delta(x_3)=(R,2,R,1,1)$. Note that
  Case~2a applies to $\mn{eq}(x_2,x_3)$. We have
  $\ubf_{x_1}=\ubf_{x_2}=x_1x_2$ and $\ubf_{x_3}=x_1x_2u_1u_2$ and
  thus obtain the following rule in $\Gamma'$:
  \[
  \begin{array}{rcl}
    \mn{goal}() &\leftarrow& 
    P^\mu(x_1,x_2,x_1,x_2,u_1,u_2) \wedge R(x_1,x_2) \, \wedge \\[1mm]
    && R(x_1,w_1) \wedge R(w_2,x_2) \wedge 
    R(w_3,x_3) \wedge R(x_3,w_4)
  \end{array}
  \]
  where the last line corresponds to $W$ above, and where
  $\mu(1)=(R,1)$ and $\mu(2)=(R,2,R,1,1)$.

  %\smallskip

  We have have to show that $I \models \Gamma'$ iff $I^g \models
  \Gamma$ for any $\Sbf_E$-instance $I$. There is a correspondence
  between extensions of $I^g$ to the IDB relations in $\Gamma$ and
  extensions of $I$ to the IDB relations in $\Gamma'$. More precisely,
  a fact $P^\mu(\abf)$ in an extension of $I$ represents a fact
  $P(\bbf)$ in an extension of $I^g$ as follows (and vice versa): for
  each $i \in \mn{pos}(P)$, let $\abf_i$ be the subtuple of $\abf$
  that starts at position $\sum_{\ell=1..i-1} q_\ell$ and is of length
  $q_i$ (where, as before, $q_\ell$ is the arity of $R$ if $\mu(\ell)=(R,i)$ and
  $q_\ell$ is the sum of the arities of $R$ and $R'$ if
  $\mu(\ell)=(R,i,R',i',j)$); the $i$-th constant
  in \bbf is $b_{R(\abf_i),j}$ if $\mu(i)=(R,j)$ and
  $b_{R(\cbf),j,R'(\cbf'),j',\ell}$ if $\mu(i)=(R,j,R',j',\ell)$ and
  $\abf_i=\cbf\cbf'$.

  One essentially has to show that every application of a rule from
  $\Gamma'$ in an extension of $I$ can be reproduced by an application
  of a rule from $\Gamma$ in the corresponding extension of $I^g$, and
  vice versa. We only sketch the details. First let $J^g$ be an
  extension of $I^g$ to the IDB relations in $\Gamma$ and let $P(\ybf)
  \leftarrow q(\xbf)$ be a rule in $\Gamma$ applicable in $J^g$, and
  $h$ a homomorphism from $q(\xbf)$ to $J^g$ such that $P(h(\ybf))
  \notin J^g$. Since $\Gamma$ is closed under contractions of rules,
  we can assume that $h$ is injective.  Let
  \[ 
  \begin{array}{rcl}
     q(\xbf) & = & P_1(\xbf_1) \wedge \cdots \wedge
     P_{\ell_1}(\xbf_{\ell_1}) \,
     \wedge \\[1mm]
     && R_1(\ybf_1) \wedge \cdots \wedge 
     R_{\ell_2}(\ybf_{\ell_2})  \,
     \wedge \\[1mm]
     &&
     \mn{eq}(z_{1,1},z_{1,2}) \wedge \cdots \wedge 
     \mn{eq}(z_{\ell_3,1},z_{\ell_3,2})
  \end{array}
  \]
  such that all $P_i$ are IDB, all $R_i$ EDB, and an equality atom
  $\mn{eq}(x,y)$ is included in the third line if and only if at least
  one of $h(x)$ and $h(y)$ is not of the form $b_p$. Consequently, for
  all variables $x$ that occur in the second line, $h(x)$ is of the
  form $b_p$. One can now verify that Condition~($*$) is
  satisfied. Assume that this is not the case. The first case is that
  that there are distinct atoms $R_i(\ybf_i)$ and $R_j(\ybf_j)$ that
  share a variable $x$. In $I^g$, every constant of the form $b_p$
  occurs in exactly one fact that only contains constants of the form
  $b_p$.  Thus, $h$ must take $R_i(\ybf_i)$ and $R_j(\ybf_j)$ to the
  same fact in $J^g$. Since $h$ is injective, $R_i(\ybf_i)$ and
  $R_j(\ybf_j)$ must be identical which is a contradiction. The second
  case is that there is an atom $R_i(\ybf_i)$ in which a variable
  occurs more than once. This is in contradiction to $h$ being a 
  homomorphism to $J^g$.

  Now define a map $\delta:\xbf \rightarrow \Omega$ by putting
  $\delta(x)=p$ if $h(x)=b_p$ and $\delta(x)=(p,p',i)$ if
  $h(x)=b_{p,p',i}$. It can be verified that the two conditions
  required of $\delta$ are satisfied. We thus obtain a corresponding
  rule in $\Gamma'$. It can be verified that applying this rule in the
  extension $J$ of $I$ corresponding to $J^g$ adds the fact that
  corresponds to $P(h(\ybf))$. 

  %\smallskip

  Conversely, let $J$ be an extension of $I$ to the IDB relations in
  $\Gamma'$ and let 
  \[ 
  \begin{array}{rcl}
     P^{\mu_0}_0(\xbf'_0) &\leftarrow& P^{\mu_1}_1(\xbf'_1) \wedge \cdots \wedge
     P^{\mu_{\ell_1}}_{\ell_1}(\xbf'_{\ell_1}) \,
     \wedge \\[1mm]
     && R_1(\ybf_1) \wedge \cdots \wedge 
     R_{\ell_2}(\ybf_{\ell_2}) \,
     \wedge \\[1mm]
     &&
     W
  \end{array}
  \]
  be a rule in $\Gamma'$ and $h$ a homomorphism from the rule body to
  $J$ such that $P^{\mu_0}(h(\xbf'_0))$ is not in $J$. This rule was
  derived 
  from a rule
  \[ 
  \begin{array}{rcl}
     P_0(\xbf_0) &\leftarrow& P_1(\xbf_1) \wedge \cdots \wedge
     P_{\ell_1}(\xbf_{\ell_1}) \,
     \wedge \\[1mm]
     && R_1(\ybf_1) \wedge \cdots \wedge 
     R_{\ell_2}(\ybf_{\ell_2})  \,
     \wedge \\[1mm]
     &&
     \mn{eq}(z_{1,1},z_{1,2}) \wedge \cdots \wedge 
     \mn{eq}(z_{\ell_3,1},z_{\ell_3,2})
  \end{array}
  \]
  in $\Gamma$ and a map $\delta: \xbf \rightarrow \Omega$, $\xbf$ the
  variables in the latter rule. We define a map $h'$ from $\xbf$ to
  $\adom(J^g)$, where $J^g$ is the extension of $I^g$ that corresponds
  to $J$. Let $x \in \xbf$. If $\delta(x)=(R,i)$ and $h(\ubf_x)=\abf$,
  then set $h'(x)=b_{R(\abf),i}$. If $\delta(x)=(R,i,R',i',j)$ and
  $h(\ubf_x)=\abf\abf'$, then set $h'(x)=b_{R(\abf),i,R'(\abf'),i',j}$.
  We argue that $h'$ is a homomorphism from the body of the
  latter rule to $J^g$. There are three cases:
  \begin{itemize}

  \item Consider an atom $P_i(\xbf_i)$. Let $\xbf_i=x_1 \cdots
    x_n$. Then there is a corresponding atom $P_i^{\mu_i}(\xbf_i)$ in
    the former rule and thus $P_i^{\mu_i}(h(\xbf_i)) \in J$. For each
    $j \in \mn{pos}(P_i)$, let $\abf_j$ be the subtuple of $h(\xbf_i)$
    that starts at position $\sum_{\ell=1..j-1} q_\ell$ and is of
    length $q_j$. Define the tuple \bbf by letting the $j$-th constant
    be $b_{R(\abf_j),\ell}$ if $\mu(j)=(R,\ell)$ and
    $b_{R(\cbf),\ell,R'(\cbf'),\ell',k}$ if
    $\mu(j)=(R,\ell,R',\ell',k)$ and $\abf_j=\cbf\cbf'$. By (R3), all
    constants in \bbf occur in the domain of $J^g$. Moreover,
    $P_i(\bbf) \in J^g$. It thus remains to observe that
    $h'(\xbf_i)=\bbf$, which follows from (R1) and (R2) and the
    definition of $h'$. %  implies
    % $R(b_{R_i(h(\ybf_i)),1},\dots,b_{R_i(h(\ybf_i)),n}) \in J^g$. By
    % Condition~1 imposed on $\delta$, we have $\delta(y_j)=(R_i,j)$ for
    % each $j$. Moreover, by (I1) we must have $\ubf_{y_j}=\ybf_i$ for
    % each $j$. Thus, the definition of $h'$ yields
    % $h'(\ybf_i)=b_{R_i(h(\ybf_i)),1}\cdots b_{R_i(h(\ybf_i)),n}$ and
    % we are done
%    .

  \item Consider an atom $R_i(\ybf_i)$. Let
    $\ybf_i=y_1 \cdots y_n$. Then the atom $R_i(\ybf_i)$ must also be in the
    former rule and thus $R_i(h(\ybf_i)) \in J$, yielding
    $R(b_{R_i(h(\ybf_i)),1},\dots,b_{R_i(h(\ybf_i)),n}) \in J^g$. By Condition~1 imposed on $\delta$, we
    have $\delta(y_j)=(R_i,j)$ for each $j$. Moreover, by (I1) we must
    have $\ubf_{y_j}=\ybf_i$ for each $j$. Thus, the definition of
    $h'$ yields $h'(\ybf_i)=b_{R_i(h(\ybf_i)),1}\cdots
    b_{R_i(h(\ybf_i)),n}$ and we are done.

  \item Consider an atom $\mn{eq}(z_{i,1},z_{i,2})$. We know that one
    of the Cases~2a to 2d apply to $\mn{eq}(z_{i,1},z_{i,2})$. We only
    treat the first case explicitly. Thus assume that
    $\delta(z_{i,1})=(R,j)$ and $\delta(z_{i,2})=(R,j,R',j',1)$. By
    definition, $h'(z_{i,1})=b_{R(h(\ubf_{z_{i,1}})),j}$ and
    $h'(z_{i,2})=b_{R(\cbf),j,R'(\cbf'),j',1}$ where
    $h(\ubf_{z_{i,1}})=\cbf\cbf'$.  By (I2), $\ubf_{z_{i,1}}$ is
    identical to the first variable block in~$\ubf_{z_{i,2}}$
    and thus $h(\ubf_{z_{i,1}})=\cbf$. By definition if $I^g$, $J^g$
    contains
    $\mn{eq}(b_{R(\cbf),j},b_{R(\cbf),j,R'(\cbf'),j',1})$
    and we are done.

  \end{itemize}
  It can now be verified that the application of the latter rule adds
  to $J^g$ the fact that corresponds to $P^{\mu_0}(h(\xbf'_0))$.
\end{proof}
DLog-rewritability of CSPs is {\sc NP}-complete
\cite{DBLP:journals/logcom/Barto16,Chen2017} and it
was observed in \cite{DBLP:journals/tods/BienvenuCLW14} that this
result lifts to generalized CSPs.  It thus follows from
Theorems~\ref{th:mmsnpToSimple} and~\ref{th:simpleToCsp} and
Lemma~\ref{lem:girthLemForDLog} that DLog-rewritability of Boolean
MDDLog programs that have equality is decidable in 2\NExpTime.  It is
straightforward to verify that the 2\NExpTime lower bound for
Datalog-rewritability of MDDLog programs from \cite{KR-submitted}
applies also to programs that have equality.
\begin{thm}
\label{thm:DLogBoolWithEq}
  For Boolean MDDLog programs that have equality, Datalog-rewri\-tability
  is 2\NExpTime-complete. 
\end{thm}
%
%As was already mentioned, 
%MDDLog programs obtained from OMQs typically do not have equality.
%
Regarding MDDLog programs that do not have equality, the above yields
a sound but possibly incomplete algorithm for deciding
DLog-rewritability. Let us make this more precise.  For an MDDLog
program $\Pi$ that does not have equality, we use $\Pi^=$ to denote
the extension of $\Pi$ with the fresh EDB relation \mn{eq} and the
above rules. If $\Pi$ has equality, then $\Pi^=$ simply denotes~$\Pi$.
\begin{lem}
\label{lem:trivial}
For MDDLog programs $\Pi$, DLog-rewitability of $\Pi^=$ implies
DLog-rewritability of $\Pi$.
\end{lem}
Lemma~\ref{lem:trivial} follows from the trivial observation that any
DLog-rewriting of $\Pi^=$ can be converted into a DLog-rewriting of
$\Pi$ by dropping all rules that use the relation $\mn{eq}$. It is an
interesting open question whether the converse of
Lemma~\ref{lem:trivial} holds. Due to Lemma~\ref{lem:trivial}, a sound
but possibly incomplete algorithm for unrestricted MDDLog programs
$\Pi$ can thus be formulated as follows: first replace $\Pi$ with
$\Pi^=$ and then decide DLog-rewritability as per
Theorem~\ref{thm:DLogBoolWithEq}. We speculate that this algorithm is
actually complete. In particular, for CSPs it is known that adding
equality does preserve Datalog-rewritability \cite{LaroseZadori07},
and completeness of our algorithm is equivalent to an analogous result
holding for MDDLog.

\subsection{Canonical Datalog-Rewritings}

For constructing actual DLog-rewritings instead of only deciding their
existence, \emph{canonical Datalog programs} play an important role.
Feder and Vardi show that for every CSP template $T$ and all
$\ell,k>0$, one can construct an $(\ell,k)$-Datalog
program that is canonical for $T$ in the sense that if there is any
$(\ell,k)$-Datalog program which is equivalent to the complement of
$T$, then the canonical one is \cite{DBLP:journals/siamcomp/FederV98}.
% It uses the IDB relations
% $P_M$ where $M$ is a set of $k$-tuples of elements of $T$ and contains
% all rules $q(\xbf) \rightarrow P_M(x_1,\dots,x_\ell)$ that satisfy the
% following conditions:
% %
% \begin{enumerate}
%
% \item $q(\xbf)$ contains exactly $\ell$ variables;
%
% \item if $h$ is a homomorphism from $q|_{\Sbf_E}$ to $T$ such
%   that for every $P_N(y_1,\dots,y_\ell) \in q$, we have
%   $(h(y_1),\dots,h(y_\ell)) \in N$, then $(h(x_1),\dots,h(x_\ell)) \in M$
%
% \end{enumerate}
% %
% There is an additional rule $P_\emptyset(x_1,\dots,x_\ell) \rightarrow
% \mn{goal}()$. 
% %
% % $q(\xbf)$ contains a single EDB atom that contains each variable from 
% % $\xbf$ exactly once, potentially plus IDB atoms;
% %
In this section, we show that there are similarly simple canonical
Datalog programs for Boolean MDDLog.  Note that the existence of
canonical Datalog programs for MMSNP (and thus for Boolean MDDLog) is
already known from \cite{DBLP:journals/jcss/BodirskyD13}. However, the
construction given there is more general and rather complex,
proceeding via an infinite template and exploiting that it is
$\omega$-categorial. This makes it hard to analyze the exact structure
and size of the resulting canonical programs.  Here, we define
canonical Datalog programs for Boolean MDDLog programs in a more
elementary way. In contrast to the previous subsection, we do not
assume that equality is available.
% These can also be utilized for OMQs, see
% Section~\ref{sect:omqs}.

Let $\Pi$ be a Boolean MDDLog program over EDB schema $\Sbf_E$ and
with IDB relations from $\Sbf_I$. Further let $0 \leq \ell < k$.  We
aim to construct a canonical $(\ell,k)$-DLog program for~$\Pi$.
The most important properties of this program is that it is sound
for $\Pi$ and complete for $\Pi$ on $\Sbf_E$-instances of treewidth
$(\ell,k)$.
 % \footnote{\color{blue}We should do it also for non-Boolean
  % programs. Not clear whether this works. } 
%of diameter $k$ 
We first convert $\Pi$ into a DDLog program $\Pi'$ that is equivalent
to $\Pi$ on instances of treewidth $(\ell,k)$ and then construct the
canonical program for $\Pi'$ rather than for $\Pi$. Unlike~$\Pi$, the
new program $\Pi'$ is not monadic. Informally, the canonical program
simulates $\Pi$ on $\Sbf_E$-instances of treewidth $(\ell,k)$
proceeding in a bag-by-bag fashion. This is enabled by the additional
non-monadic IDB relations introduced in $\Pi'$ which represent
information that needs to be passed from bag to bag.
% It consists of MDDLog rules and rules
% that are neither disjunctive nor monadic.
We remark that the construction of $\Pi'$ is vaguely similar in spirit
to the first step of converting an MDDLog program into simple form,
c.f.\ Appendix~\ref{app:step1tosimple}. To describe it, we need a
preliminary.

With every MDDLog rule $p(\ybf) \leftarrow q(\xbf)$ where $q(\xbf)$ is of
treewidth $(\ell,k)$ and every $(\ell,k)$-tree decomposition
$(T,(B_v)_{v \in V})$ of $q(\xbf)$, 
% such that $q(\xbf)|_{B_v}$ is $\ell+1$-connected for all $v \in V$, 
we associate a set of \emph{rewritten rules} constructed as
follows. Choose a root $v_0$ of the undirected tree $T$, thus inducing
a direction. We write $v \prec v'$ if $v'$ is a successor of $v$ in
$T$ and use $\xbf_{v'}$ to denote $B_v \cap B_{v'}$. For all
$v \in V \setminus \{ v_0 \}$ such that $|\xbf_{v}|=m$, introduce a
fresh $m$-ary IDB relation $Q_{v}$; note that $m \leq \ell$. Now, the
set of rewritten rules contains one rule for each $v \in V$. For
$v \neq v_0$, the rule is
\[
p_v(\ybf_v) \vee Q_{v}(\xbf_{v})
\leftarrow q(\xbf)|_{B_v} \wedge \bigwedge_{v \prec v'} Q_{v'}(\xbf_{v'})
\]
where $p_v(\ybf_v)$ is the sub-disjunction of $p(\ybf)$ that contains
all disjuncts $P(\zbf)$ with $\zbf \subseteq B_v$ and $q(\xbf)|_{B_v}$
is the restriction of $q$ to the atoms that contain only variables
from $B_v$. For $v_0$, we include the same rule, but use only
$p_v(\ybf_v)$ as the head.  The set of rewritten rules associated with
$p(\ybf) \leftarrow q(\xbf)$ is obtained by taking the union of the
rewritten rules associated with $p(\ybf) \leftarrow q(\xbf)$ and any
$(T,(B_v)_{v \in V})$. % that satisfies the required conditions.

%\smallskip

The DDLog program $\Pi'$ is constructed from $\Pi$ as follows:
\begin{enumerate}

\item first extend $\Pi$ with all contractions of rules in $\Pi$;

\item then delete all rules with $q(\xbf)$ not of treewidth $(\ell,k)$
  and replace every rule $p(\ybf) \leftarrow q(\xbf)$ with $q(\xbf)$
  of treewidth $(\ell,k)$ with the rewritten rules associated with it.

\end{enumerate}
%
% \begin{itemize}
%
% \item[($*$)] replace every rule
%   $p(\vect{y}) \leftarrow q_1(\vect{x}_1) \wedge q_2(\vect{x}_2)$
%   where $\vect{x}_1$ and $\vect{x}_2$ share $m \in \{0,\dots,\ell\}$
%   variables $\xbf$ but both contain also other variables 
%   %
%   % needed because we cannot futher normalize non-monadic
%   % head relations
%   %
%   and $q_1(\vect{x}_1)$ is $\ell+1$-connected 
%   %
%   with the rules
%   $p_1(\vect{y}_1) \vee Q(\xbf) \leftarrow q_1(\vect{x}_1)$ and
%   $p_2(\vect{y}_2) \leftarrow Q(\xbf) \wedge q_2(\vect{x}_2)$, where
%   $Q$ is a fresh $m$-ary IDB relation and $p_i(\vect{y}_i)$ is the
%   restriction of $p(\vect{y})$ to atoms that are nullary or contain a
%   variable from $q_i$, $i \in \{1,2\}$.
%  
% \end{itemize}
%
%We then drop all rules whose body is not $\ell+1$-connected. 
To clarify the relation between $\Pi$ and $\Pi'$, we remark that
it is possible to verify the following conditions; a detailed proof
is omitted since these conditions are not going to be used in 
what follows:
%
%It can be verified that $\Pi'$ satisfies the following conditions:
% It is easy to see that the resulting program $\Pi'$ is equivalent to
% $\Pi$.
%
\begin{enumerate}[align=left]

\item[(I)] $\Pi'$ is sound for $\Pi$, that is, for all
  $\Sbf_E$-instances $I$, $I \models \Pi'$ implies $I \models
  \Pi$;

\item[(II)] $\Pi'$ is complete for $\Pi$ on $\Sbf_E$-instances of treewidth
  $(\ell,k)$,  that is, for all such instances $I$, $I \models \Pi$
  implies
  $I \models \Pi'$.

\end{enumerate}
Note that $\Pi'$ is not complete for $\Pi$ on instances of unrestricted
treewidth. For example, if $\Pi$ consists of only a goal rule whose
rule body is a $k+1$-clique (without reflexive loops), then $\Pi'$
returns false on the instance that consists of the same clique.
\begin{exa}
  Assume that $\Pi$ contains the rule 
  \[
    P_1(x) \vee P_2(z) 
    \leftarrow R(x,y_1) \wedge S(x,y_2) \wedge R(y_1,z) \wedge R(y_2,z)
  \]
  and consider the $(2,3)$-tree decomposition of the rule body that
  consists of two nodes $v,v'$, $v'$ successor of $v$, with
  $B_v=\{x,y_1,y_2\}$ and $B_{v'}=\{y_1,y_2,z\}$. In $\Pi'$, the rule
  is split into two rules
  \[
  \begin{array}{rcl}
    P_2(z) \vee Q_{v'}(y_1,y_2) 
    \leftarrow R(y_1,z) \wedge R(y_2,z) \\[1mm]
    P_1(z) \leftarrow R(x,y_1) \wedge S(x,y_2) \wedge Q_{v'}(y_1,y_2).
  \end{array}
  \]
  Informally, these rules are supposed to cover homomorphisms from the
  body of the original rule to an $\Sbf'_E$-instance of treewidth
  $(\ell,k)$ such that the variables in $B_{v'}$ are mapped to
  constants from some bag and variables from $B_v$ to constants from a
  neighboring bag. The IDB relation $Q_{v'}$ memorizes that we have
  already seen part of the rule body.
\end{exa}

%\smallskip

Let $\Sbf'_I$ denote the additional IDB relations in $\Pi'$. We now
construct the canonical $(\ell,k)$-DLog program $\Gamma^c$ for $\Pi$.
Fix constants $a_1,\dots,a_\ell$. For $\ell' \leq \ell$, we use
$\Imf_{\ell'}$ to denote the set of all $\Sbf_I \cup
\Sbf'_{I}$-instances with domain $\abf_{\ell'}:=a_1,\dots,a_{\ell'}$.
The program uses $\ell'$-ary IDB relations $P_M$, for all $\ell' \leq
\ell$ and all $M \subseteq \Imf_{\ell'}$. It contains all rules
$q(\xbf) \rightarrow P_{M}(\ybf)$, $M \subseteq \Imf_{\ell'}$, that
satisfy the following conditions:
\begin{enumerate}

\item $q(\xbf)$ is over schema $\Sbf_E \cup \{ P_M \mid M
  \subseteq
  \Imf_{\ell'},~\ell' \leq \ell \}$ and contains at most $k$ variables;

%\item if the IDB relation $P_{S',t'}$ occurs in $q(\xbf)$, then $t'=t$;

\item for every extension $J$ of the $\Sbf_E$-instance $I_q|_{\Sbf_E}$ with
  $\Sbf_I \cup \Sbf'_I$-facts such that
  \begin{enumerate}

  \item $J$ satisfies all rules of $\Pi'$ and does not contain
    $\mn{goal}()$ and

%  \item for each $P_{N}(\zbf) \in q$, we have $J|_{\Sbf_I \cup \Sbf'_I}[\abf/\zbf] \in N$,
%DUPLICATION!!

  \item for each $P_{N}(\zbf) \in q$, $N \subseteq
    \Imf_{\ell''}$, there is an $L \in N$ such that
%    $L[\zbf/\abf] \subseteq J$
    $L[\zbf/\abf_{\ell''}]=J|_{\Sbf_I \cup
      \Sbf'_I,\zbf}$

  \end{enumerate}
  there is an $L \in M$ such that $L[\ybf/\abf_{\ell'}]=J|_{\Sbf_I \cup
    \Sbf'_I,\ybf}$

\end{enumerate}
where $I_q$ is $q$ viewed as an instance, $L[\xbf/\abf]$ denotes the
result of replacing the constants in \abf with the variables in~\xbf
(possibly resulting in identifications), and $J|_{\Sbf_I \cup
  \Sbf'_I,\xbf}$ denotes the simultaneous restriction of $J$ to schema
$\Sbf_I \cup \Sbf'_I$ and constants $\xbf$.\footnote{We could
  additionally demand that $M$ is minimal so that Condition~2 is
  satisfied, but this is not strictly required.}  We also include in
$\Gamma^c$ all rules of the form $P_\emptyset(\xbf)\rightarrow
\mn{goal}()$, $P_\emptyset$ of any arity from $0$ to $\ell$.

The intuition behind the construction of $\Gamma^c$ is as follows.
When starting with an input $\Sbf_E$-instance $I$ of treewidth
$(\ell,k)$ and then chasing with $\Gamma^c$, that is, exhaustively
applying these rules in an unspecified order, then the resulting
instance $I'$ represents all extensions $J$ of $I$ to the relations in
$\Sbf_I \cup \Sbf_I'$ that satisfy all rules in $\Pi'$ and do not
contain $\mn{goal}()$. A fact $P_M(\abf) \in I'$, $M \subseteq
\Imf_{\ell''}$, means that for every such $J$ there is an $L \in M$
such that $J$ contains the facts in $L[\abf/\abf_{\ell''}]$. Thus, the
set $M$ in the index of $P_M$ should be read disjunctively. Note that
$P_\emptyset(\abf)\in I'$ then indicates that every extension of $I$
that satisfies all rules in $\Pi'$ must contain $\mn{goal}()$. The
bodies of rules in $\Gamma^c$ are large enough to cover the
restriction of $I$ to the constants from any single bag. This suffices
only because we have transitioned from $\Pi$ to $\Pi'$ before
constructing $\Gamma^c$.

% It can be shown that the number of IBD relations is bounded by $2^{2^{|\Sbf_I|
%     \cdot \ell}}$ and number of rules by $2^{|\Sbf_E| \cdot
%   k^{a+\ell} \cdot 2^{|\Sbf_I|\cdot\ell}}+1$ where $a$ is the maximum
% arity of relations in $\Sbf_E$.\footnote{To obtain this bound, one has
%   to assume that the rule body contains at most one IDB atom
%   $P_M(x_1,\dots,x_\ell)$ for any $\ell$-tuple of variables
%   $x_1,\dots,x_\ell$; obviously, this can be done without loss of
%   generality.} The canonical DLog
% It can be verified that 
%  the constructed program is of size double
% exponential in the size of the original program.
%
The following are central properties of canonical DLog programs.
\begin{lem}
\label{lem:cancompllk}\hfill
\begin{enumerate}
\item $\Gamma^c$ is sound for $\Pi$;
\item $\Gamma^c$ is complete for $\Pi$ on instances of treewidth $(\ell,k)$.
\end{enumerate}
%
  % Let $\Pi$ be a Boolean MDDLog program over EDB schema $\Sbf_E$, $0 <
  % \ell \leq k$, $\Gamma^c$ the canonical $(\ell,k)$-DLog program for
  % $\Pi$, and $I$ an $\Sbf_E$ instance of tree-width $(\ell,k)$ with 
  % $I \models \Pi$. Then $I \models \Gamma^c$.
\end{lem}
\begin{proof}\ For Point~1, let $I$ be an $\Sbf_E$-instance with $I
  \models \Gamma^c$.  It suffices to show that $I \models \Pi'$. Let
  $I=I_1,I_2,\dots$ be the sequence of $\Sbf_E \cup \Sbf_I \cup
  \Sbf'_I$-instances obtained by chasing $I$ with $\Gamma^c$.
% , that is,
%   by exhaustively applying the rules in $\Gamma^c$ in an unspecified
%   order.
  We first note that the following can be proved by induction
  on $i$ (and using the definition of~$\Gamma_c$):
  \\[2mm]
  {\bf Claim.} If $P_M(\bbf) \in I_i$, $M \subseteq \Imf_{\ell'}$,
  then for every extension $J$ of $I$ to the relations in
  $\Sbf_I\cup\Sbf'_I$ that satisfies all rules of $\Pi'$ and does not
  contain $\mn{goal}()$, there is an $L \in M$ such that $L[\bbf/\abf_{\ell'}]
 =  J|_{\Sbf_I\cup\Sbf'_I,\bbf}$.
  \\[2mm]
  Since $I \models \Gamma^c$, there are $i>0$ and $\bbf \subseteq
  \adom(I)$ such that $P_\emptyset(\bbf) \in I_i$. By the claim,
  there is thus no extension $J$ of $I$ to the relations in $\Sbf_I
  \cup \Sbf'_I$ that satisfies all rules of $\Pi$ and does not contain
  $\mn{goal}()$. Consequently, $I \models \Pi'$.

  %\medskip
 
  For Point~2, assume that $I \not\models \Gamma^c$ and let $(T,(B_v)_{v
    \in V})$ be an $(\ell,k)$-tree decomposition of~$I$, $T=(V,E)$.
  Then there is an extension $J$ of $I$ to the IDB relations in
  $\Gamma^c$ such that all rules in $\Gamma^c$ are satisfied and $J$
  contains no atom of the form $P_\emptyset(\bbf)$.

  We use $J$ to construct an extension $J'$ of $I$ to the relations in
  $\Sbf_I \cup \Sbf'_I$. Choose a root $v_0$ of $T$, thus inducing a
  direction on the undirected tree $T$. For all $v \in V$ and
  successors~$v'$ of $v$, choose an ordering $\cbf_{v,v'}$ of the
  constants in $B_v \cap B_{v'}$ and let ${\ell_{v,v'}}$ denote the
  number of these constants. Let
  $P_{M_1}(\cbf_{v,v'}),\dots,P_{M_r}(\cbf_{v,v'})$ be all facts of
  this form in $J$. By construction of $\Gamma^c$, there must be at
  least one such fact, and the fact $P_{M_1 \cap \cdots \cap
    M_r}(\cbf_{v,v'})$ must also be in $J$. Thus, we can associate
  with $v,v'$ a unique minimal set $M_{v,v'}$ so that
  $P_{M_{v,v'}}(\cbf_{v,v'}) \in J$.

  The construction of $J'$ proceeds top down over $T$. At all points,
  we maintain the invariant that
  \begin{itemize}[align=left]

  \item[($*$)] for all nodes $v \in V$ and successors $v'$ of $v$,
    there is an $L \in M_{v,v'}$ such that
    $L[\cbf_{v,v'}/\abf_{\ell_{v,v'}}] = J'|_{\Sbf_I \cup
      \Sbf'_I,\cbf_{v,v'}}$.

  \end{itemize}
  The construction of $J'$ starts at the root $v_0$ of $T$. There must
  be an
  extension $J_{v_0}$ of $I|_{B_{v_0}}$ with $S_I \cup S'_I$-facts such that
  \begin{enumerate}[label=(\roman*)]

  \item $J_{v_0}$ satisfies all rules of $\Pi$ and does not contain
    $\mn{goal}()$ 
    
  \item for each $P_M(\bbf) \in J|_{B_{v_0}}$, $M \subseteq
    \Imf_{\ell'}$, there is an
    $L \in M$
    such that $L[\bbf/\abf_{\ell'}] = J_{v_0}|_{S_i \cup S'_I,\bbf}$

  % \item for each successor $v$ of $v_0$, there is an $L_{v_0,v} \in M_{v_0,v}$
  %   such that $L_{v_0,v}[\zbf/\abf] = J_{v_0}|_{S_i \cup S'_I,\zbf}$

  \end{enumerate}
  as, otherwise, a rule of $\Gamma^c$ would create an atom of the form
  $P_\emptyset(\cbf)$ in $J$. Start with putting $J'=I \cup J_{v_0}$.
  Note that for each successor $v$ of $v_0$, ($*$) is satisfied
  because of Point~(ii)
  and since $P_{M_{v_0,v}}(a_{v_0,v}) \in J|_{B_{v_0}}$.

  We proceed top-down over $T$. Assume that $v'$ is a successor of $v$
  and $B_v$ has already been treated. There must be an
  extension $J_{v'}$ of $I|_{B_{v'}}$ with $S_I \cup S'_I$-facts
  such that
  \begin{enumerate}[label=(\roman*)]

  \item $J_{v'}$ satisfies all rules of $\Pi$ and does not contain
    $\mn{goal}()$,

  \item $J_{v'}|_{\Sbf_I \cup \Sbf'_I,\cbf_{v,v'}}=J'|_{\Sbf_I \cup
      \Sbf'_I,\cbf_{v,v'}}$, and

  \item for each $P_M(\bbf) \in J|_{B_{v'}}$, $M \subseteq
    \Imf_{\ell'}$, there is an
    $L \in M$
    such that $L[\bbf/\abf_{\ell'}] = J_{v'}|_{S_i \cup S'_I,\bbf}$

  \end{enumerate}
  as, otherwise, because of ($*$) a rule of $\Gamma^c$ would create an
  atom of the form
  $P_M(\cbf_{v,v'})$ in $J$ with $M \subsetneq M_{v,v'}$, in contradiction
  to
  $M_{v,v'}$ being minimal with $P_{M_{v,v'}}(\cbf_{v,v'}) \in J$. 
  Put $J'=J' \cup J_{v'}$. It can again be verified that ($*$)
  is satisfied.

%\smallskip 

  By construction, the instance $J'$ does not contain $\mn{goal}()$
  and $(T,(B_v)_{v \in V})$ is also a tree decomposition of $J'$, that
  is, each EDB atom \emph{and each IDB atom} of $J'$ falls within some
  bag $B_v$.  We aim to show that $J'$ satisfies all rules of $\Pi$,
  thus $I \not \models \Pi$ as required.

%\smallskip 

  Let $\Pi_0$ be the result of closing $\Pi$ under contractions of
  rules and recall that $\Pi'$ is obtained from $\Pi_0$ by dropping
  and rewriting rules. Let $\rho$ be a rule in $\Pi$ and let $h$ be a
  homomorphism from its body to $J'$. We have to show that one of the
  disjuncts in the head of $\rho$ is satisfied under $h$.  $\Pi_0$
  contains the rule $\rho_0$ obtained from $\rho$ by identifying all
  variables $x,y$ such that $h(x)=h(y)$.
% Let $\rho_0=q(\xbf) \rightarrow
% R_{1}(\xbf_{1}) \vee \cdots \vee R_{n}(\xbf_{n})$. 
It clearly
suffices to show that  one of the
disjuncts in the head of $\rho_0$ is satisfied under $h$.
%$R_i(h(\xbf_i)) \in J'$ for some $i$.  
Note that $h$ is an injective homomorphism from the body $q(\xbf)$ of
$\rho_0$ to $J'$ which implies that $q(\xbf)$ is of treewidth
$(\ell,k)$. % since $(T,(B_v)_{v \in V})$ is a tree decomposition of
            % $J'$. 
Moreover, we can read off an $(\ell,k)$-tree decomposition
$(T',(B'_v)_{v \in V'})$ of $q(\xbf)$ from $h$ and $(T,(B_v)_{v \in
  V})$.

% By further splitting bags in
% $(T',(B'_v)_{v \in V'})$, we can make sure that $B'_v$ is
% $\ell+1$-connected for every $v \in V'$. 
In~$\Pi'$, $\rho_0$ and $(T',(B'_v)_{v \in V'})$ are rewritten into
rules $\rho_1,\dots,\rho_m$ such that no $\rho_i$ uses a fresh IDB
relation from the head of any $\rho_j$ with $j \geq i$ (that is, an
IDB relation that does not occur in $\Pi_0$, of arity at most $\ell$).
Let $\rho_i$ be $q_i(\xbf_i) \rightarrow R_{i,1}(\xbf_{i,1}) \vee \cdots
\vee R_{i,n_i}(\xbf_{i,n_i}) \vee Q_i(\zbf_i)$ where
$R_{i,1}(\xbf_{i,1}),\dots,R_{i,n_i}(\xbf_{i,n_i})$ are disjuncts that also
occur in the head of $\rho_0$ and $Q_i$ is a fresh IDB relation
introduced by the rewriting in the case that $i<m$ and $\mn{false}$ if
$i=m$ (by which we mean: there is no $Q_i(\zbf_i)$ disjunct in the latter
case).  One can show by induction on $i$ that for $1 \leq i\leq m$,
\begin{enumerate}

\item $R_{j,m}(h(\xbf_{j,t})) \in J'$ for some $j \leq i$ and $t \in
  \{1,\dots,n_j\}$ or

\item %$h$ is an injective homomorphism from $q_i(\xbf_i)$ to $J'$
    %and 
$Q_i(h(\zbf_i)) \in J'$.

\end{enumerate}
To see this, assume that Point~1 is not satisfied for some $i$. Then
Point~2 holds for all $j < i$.  By choice of $\rho_i$, there is
a $v \in V$ such that $h(\xbf_i) \subseteq B_v$. 
%
% note in particular that for each $i$, there is a $v \in
% V$ such that $h(\xbf_i) \subseteq B_v$ (by construction of
% $\rho_1,\dots,\rho_m$).
% ; in fact, this follows from the
% $\ell+1$-connectedness of $q_i(\xbf_i)$ and the fact that all atoms
% of~$J'$, including IDB atoms, fall within some $B_v$.
Thus $h$ is a homomorphism from $q_i(\xbf_i)$ to $J_v$, and
consequently there is a disjunct $R(\zbf)$ in the head of $\rho_i$
such that $R(h(\zbf)) \in J_v \subseteq J'$. This implies that one of
Points~1 or~2 is satisfied for $i$.

%\smallskip

Note that Point~2 cannot hold for $i=m$ because the $Q_m$
disjunct is not present in $\rho_m$. Thus there is an $i \leq m$
such that $R_{i,j}(h(\xbf_{i,j})) \in J'$ for some $j$. Since
$R_{i,j}(\xbf_{i,j})$ occurs in the head of $\rho_0$, we are done.
\end{proof}
We are now ready to show that the canonical program is indeed
canonical, as detailed by the following theorem. For two Boolean DLog
programs $\Pi_1,\Pi_2$ over the same EDB schema $\Sbf_E$, we write
$\Pi_1 \subseteq \Pi_2$ if for every $\Sbf_E$-instance $I$, $I \models
\Pi_1$ implies $I \models \Pi_2$.
\begin{thm}
\label{thm:canonical}
%A Boolean MDDLog program $\Pi$ is DLog-rewritable iff 
%the canonical DLog program for $\Pi$ is a rewriting of $\Pi$. 
Let $\Pi$ be a Boolean MDDLog program, $0 \leq \ell \leq k$, and
$\Gamma^c$ the canonical $(\ell,k)$-DLog program for $\Pi$. Then
\begin{enumerate}

%\item $\Gamma^c$ is a sound for $\Pi$;

\item $\Gamma \subseteq \Gamma^c$ for every $(\ell,k)$-DLog
  program $\Gamma$ that is sound for $\Pi$;

\item $\Pi$ is $(\ell,k)$-DLog-rewritable iff $\Gamma^c$ is a
  DLog-rewriting of~$\Pi$.

\end{enumerate}
\end{thm}
\begin{proof}\ 
  Let $\Sbf_E$ be the EDB schema of $\Pi$.

 % \smallskip 

  For Point~1, let $\Gamma$ be an $(\ell,k)$-DLog program that is
  sound for $\Pi$ and let $I$ be an $\Sbf_E$-instance with $I \models
  \Gamma$. From the proof tree for $\mn{goal}()$ from $I$ and $\Gamma$,
  we can construct an $\Sbf_E$-instance $J$ of treewidth $(\ell,k)$
  such that $J \models \Gamma$ and $J \rightarrow I$. It suffices to
  show that $J \models \Gamma_c$, which is easy: from $J \models
  \Gamma$, we obtain $J \models \Pi$ and Point~2 of
  Lemma~\ref{lem:cancompllk} yields $J \models \Gamma^c$.
  
  %\smallskip 

  The ``if'' direction of Point~2 is trivial. For the ``only if''
  direction, assume that $\Pi$ is $(\ell,k)$-DLog-rewritable and let
  $\Gamma$ be a concrete rewriting.  We have to show that $\Gamma^c$
  is sound and complete for $\Pi$. The former is Point~1 of
  Lemma~\ref{lem:cancompllk}. For the latter, we get $\Pi \subseteq
  \Gamma$ since $\Gamma$ is a rewriting of $\Pi$ and $\Gamma \subseteq
  \Gamma^c$ from Point~1, thus $\Pi \subseteq \Gamma^c$ as required.
\end{proof}
Note that by Point~1 of Theorem~\ref{thm:canonical}, the canonical
$(\ell,k)$-DLog program for an MDDLog program $\Pi$ is interesting
even if $\Pi$ is not rewritable into an $(\ell,k)$-DLog program as it
is the strongest sound $(\ell,k)$-DLog approximation of $\Pi$.

% \smallskip
%
% NOT TRUE
%
% It is straightforward to generalize the above construction of
% canonical Datalog programs to non-Boolean MDDLog programs that are of
% arity one or in which all goal rules are EDB-guarded in the sense that
% the variables in the rule head cooccur in an EDB atom in the rule
% body. It it not clear to us how to extend the construction to
% non-Boolean MDDLog programs that do not satisfy this condition.

\section{Non-Boolean MDDLog Programs}
\label{sec:answer}

We lift the results about the complexity of rewritability, about
canonical DLog programs, and about the shape of rewritings and
obstructions from the case of Boolean MDDLog programs to the
non-Boolean case. For all of this, a certain extension of
$(\ell,k)$-Datalog programs with parameters plays a central role.
We thus begin by introducing these extended programs. 

\subsection{Deciding Rewritability}

% An \emph{$(\ell,k)$-Datalog program with $n$ parameters} is a an
% $n$-ary $(\ell+n,k+n)$-Datalog program with distinguished
% \emph{parameter variables} $x_1,\dots,x_n$ in which all IDBs $P$
% (including \mn{goal}) have arity at least $n$ and only occur in the
% form $P(\cdots,x_1,\dots,x_{n})$, both in rule bodies and
% heads. Monadic Datalog programs with $n$ parameters are then defined
% in the expected way. It is easy to see that $(\ell,k)$-DLog programs
% with $0$ parameters are the same as Boolean $(\ell,k)$-DLog programs.

An \emph{$(\ell,k)$-Datalog program with $n$ parameters} is an $n$-ary
$(\ell+n,k+n)$-Datalog program in which all IDBs have arity at least
$n$ and where in every rule, all IDB atoms agree on the variables used
in the last $n$ positions (both in rule bodies and heads and including
the \mn{goal} IDB). The last $n$ positions of IDBs are called
\emph{parameter positions}. To visually
separate the parameter positions from the non-distinguished positions,
we use ``$|$'' as a delimiter to replace the usual comma, writing
e.g.\
\[
P(x_1,x_2\,|\,y_1,y_1,y_2) \leftarrow Q(y_1\,|\,y_1,y_1,y_2) \wedge
R(x_1,y_1,y_2,x_2)
\]
where $P,Q$ are IDB, $R$ is EDB, and there are three parameter
positions.  Note that, by definition, all variable positions in
\mn{goal} atoms are parameter positions.
\begin{exa}
\label{ex:par}
  The following is an MDLog program with one parameter that returns
  all constants which are on an $R$-cycle, $R$ a binary EDB relation:
\[
\begin{array}{rcl}
     P(y \,|\, x) &\leftarrow& R(y \,|\, x)  
     \\[\myeqnsep]
     P(z \,|\, x) &\leftarrow& P(y \,|\, x) \wedge R(z,y)
     \\[\myeqnsep]
     \mn{goal}(x) &\leftarrow& P(x \,|\, x)
\end{array}
\]
\end{exa}
Parameters in Datalog programs play a similar role as parameters to
least fixed-point operators in FO(LFP), see for example
\cite{BenediktLICS16} and references therein. The program in
Example~\ref{ex:par} is not definable in MDLog without parameters,
which shows that adding parameters increases expressive
power. Although $(\ell,k)$-DLog programs with $n$ parameters are
$(\ell+n,k+n)$-DLog programs, one should think of them as a mild
generalization of
$(\ell,k)$-programs.% In spirit, an $(\ell,k)$-program with $n$
% parameters is very close to an $(\ell,k)$-program with $n$ constants
% (whose answer consists of the constants or is empty). This is made
% more precise in the proof of Lemma~\ref{lem:nonBoolRed1} below.

%\smallskip

A DLog program is an \emph{$\ell$-DLog program} if it is an
$(\ell,k)$-DLog program for some $k$.  To lift decidability and
complexity results from the Boolean to the non-Boolean case, we show
that rewritability of an $n$-ary MDDLog program into $\ell$-DLog with
$n$ parameters can be reduced to rewritability of a Boolean MDDLog
program into $\ell$-DLog (without parameters). We believe that Datalog
with parameters is a natural rewriting target for non-Boolean MDDLog
programs since, in a sense, the $n$ parameters reflect the special
role of the constants from the input instance that are returned as an
answer. Note that the case $\ell=0$ is about UCQ-rewritability (and
thus FO-rewritability) because $0$-DLog programs (with and without
parameters) are an alternative presentation of UCQs. The reduction
proceeds in two steps, described by subsequent
Lemmas~\ref{lem:nonBoolRed1} and~\ref{lem:nonBoolRed2}.  % We remark
% that, as expected, rewritability into $(\ell,k)$-DLog-programs with
% parameters
%
\begin{exa}
\label{ex:avar2}
  The following MDDLog program is rewritable into the
  MDLog program with parameters from Example~\ref{ex:par},
  but not into an MDLog program without parameters:
\[
  \begin{array}{r@{\;}c@{\;}l}
P_0(x) \vee P_1(y) &\leftarrow& R(x,y) \\[\myeqnsep]
\mn{goal}(x) &\leftarrow& P_0(x) \\[\myeqnsep]
P_{1}(y) &\leftarrow& P_1(x) \wedge R(x,y) 
\\[\myeqnsep]
\mn{goal}(x) &\leftarrow& P_1(x).
\end{array}
\]
\end{exa}
%
% When we speak of rewritability of an MDDLog program $\Pi$ into
% $(\ell,k)$-DLog with parameters, $\Pi$ $n$-ary, we generally mean
% rewriting into an $(\ell,k)$-DLog program with $n$ parameters. When we
% speak of rewritability of an MDDLog program $\Pi$ with constants into
% an $(\ell,k)$-DLog program with constants, we generally mean a
% rewriting into an $(\ell,k)$-DLog program that uses no other constants
% that those in $\Pi$. 
The following lemma shows that, by introducing constants, we can
reduce the rewritability of non-Boolean MDDLog programs into
Datalog with parameters to the rewritability of Boolean MDDLog programs
with constants into Datalog with constants. Note that the presence of
constants in an $(\ell,k)$-DLog program is not reflected in the values
of $\ell$ and~$k$. We will show in a second step that the
rewritability of Boolean MDDLog programs with constants into Datalog
with constants can be reduced to the rewritability of Boolean MDDLog
programs without constants into Datalog without constants. 

The \emph{diameter} of an $(\ell,k)$-DLog program with $n$ parameters
is $k$ and the \emph{diameter} of a DLog program with constants is
defined as for DLog programs without constants, that is, only
variables contribute to the diameter, but constants do not. The
\emph{rule size} of an MDDLog program is the maximum number of
variable \emph{occurrences} in a rule body.
\begin{lem}
\label{lem:nonBoolRed1} 
Given an $n$-ary MDDLog program $\Pi$, one can construct Boolean
MDDLog programs with constants $\Pi_1,\dots,\Pi_m$ over the
same EDB schema such that for all $\ell,k$,
\begin{enumerate}

\item $\Pi$ is rewritable into an $(\ell,k)$-DLog program with $n$
  parameters iff each of $\Pi_1,\dots,\Pi_m$ is rewritable into an
  $(\ell,k)$-DLog program with constants;

\item $m \leq n^n$ and the size (resp.\ diameter, rule size) of
  each program $\Pi_i$ is bounded by the size (resp.\ diameter, rule size) of $\Pi$.

\end{enumerate}
The construction takes time polynomial in the size of $|\Pi_1 \cup
\cdots \cup \Pi_m|$. 
\end{lem}
\begin{proof}\ Let $\Pi$ be an $n$-ary MDDLog program over EDB schema
  $\Sbf_E$. Fix a set $C$ of $n$ constants. For each $\vect{c} \in
  C^n$, we construct from $\Pi$ a Boolean MDDLog program $\Pi_\cbf$
  such that for any $\ell < k$, $\Pi$ is $(\ell,k)$-DLog rewritable
  iff all programs $\Pi_\cbf$ are.

  Let $\cbf \in C^n$.  Given two $n$-tuples of terms (constants or
  variables) $\vect{s}$ and $\vect{t}$, we write $\sbf \preceq \tbf$
  if $t_i=t_j$ implies $s_i=s_j$ for $1 \leq i <j \leq n$. We write
  $\sbf \approx \tbf$ when $\sbf \preceq \tbf \preceq \sbf$. The
  program $\Pi_\cbf$ is obtained from $\Pi$ as follows:
  \begin{itemize}

  \item replace every rule $\mn{goal}(\xbf) \leftarrow q(\xbf,\ybf)$
    with $\cbf \preceq \xbf$ by $\mn{goal}() \leftarrow q(\cbf,\ybf)$;

  \item drop every rule  $\mn{goal}(\xbf) \leftarrow q(\xbf,\ybf)$
    with $\cbf  \not\preceq \xbf$.

  \end{itemize}
  Note that the non-goal rules in $\Pi_\cbf$ are identical to those in
  $\Pi$. By converting proof trees for $\Pi$ into
  proof trees for $\Pi_\cbf$ and vice versa, one can show
  the following.
  \\[2mm]
  {\bf Claim.} For all $\Sbf_E$-instances $I$ and $\vect{a} \subseteq
  \adom(I)^n$ with $\abf \approx \cbf$, $I \models \Pi(\vect{a})
  \mbox{ iff } I[\cbf/\abf] \models \Pi_{\vect{c}}$.
  \\[2mm]
  We show that $\Pi$ is rewritable into an $(\ell,k)$-DLog program
  with $n$ parameters iff all of the constructed programs
  $\Pi_{\vect{c}}$ are rewritable into an $(\ell,k)$-DLog program with
  constants.

  %\medskip 
  Let $\Gamma$ be an $(\ell,k)$-DLog program with $n$
  parameters that is a rewriting of~$\Pi$. For each $\vect{c} \in
  C^n$, let $\Gamma_{\vect{c}}$ be the Boolean $(\ell,k)$-DLog program with
  constants obtained from $\Gamma$ as follows:
  \begin{itemize}

  \item replace every rule $P(\vect{x}\,|\,\vect{y}) \leftarrow
    q(\vect{z} \,|\, \vect{y})$ with $\cbf \preceq \ybf$ (and where
    $P$ might be $\mn{goal}$) by $P(\vect{x}_{\cbf}) \leftarrow
    q(\vect{z}_{\cbf})$, where $\vect{v}_{\cbf}$ is the
    result of replacing in \vbf each variable $y_i$ with $c_i$;
% ; thus $v_i$ is replaced with $c_j$
%     whenever $v_i=y_j$.

  \item drop every rule $P(\vect{x}\,|\,\vect{y}) \leftarrow
    q(\vect{z} \,|\, \vect{y})$ with $\cbf \not\preceq \ybf$.

  \end{itemize}
  By translating proof trees, it can be shown that ($*$) $I \models
  \Gamma(\vect{c})$ iff $I \models \Gamma_{\vect{c}}$. It is now easy
  to show that $\Gamma_{\vect{c}}$ is a rewriting of $\Pi_{\vect{c}}$:
  for every $\Sbf_E$-instance $I$, $I \models \Pi_{\vect{c}}$ iff $I
  \models \Pi(\vect{c})$ (by the claim) iff $I \models
  \Gamma(\vect{c})$ (since $\Gamma$ is a rewriting of $\Pi$) iff $I
  \models \Gamma_{\vect{c}}$ (by ($*$)).

  %\medskip 
  Conversely, for all $\vect{c} \in C^n$ let $\Gamma_\cbf$ be a
  Boolean $(\ell,k)$-DLog program with constants that is a rewriting
  of~$\Pi_\cbf$. We construct an $(\ell,k)$-DLog program with $n$
  parameters $\Gamma$ as follows. For each $\vect{c} \in C^n$, fix a
  tuple $\vect{v}$ of fresh variables such that $\vbf \approx
  \cbf$. Let $\Gamma^v_{\vect{c}}$ be the $(\ell,k)$-DLog program with
  $n$ parameters obtained from $\Gamma_{\vect{c}}$ as follows:
  \begin{enumerate}

  \item[(i)] replace each $c_i$ with $v_i$;

  \item[(ii)] replace each non-goal IDB atom $P(\vect{x})$ with the
    atom $P^{\vect{c}}(\vect{x} \,|\, \vect{v})$ (both in rule bodies and
    heads), $P^\cbf$ a fresh IDB relation;

  \item[(iii)] replace $\mn{goal}()$ with $\mn{goal}(\vbf)$.

  \end{enumerate}
  Then $\Gamma$ is defined as the union of all programs
  $\Gamma^v_{\vect{c}}$.  We first argue that for every $\cbf \in C^n$,
  $\Sbf_E$-instance $I$, and $\vect{a} \subseteq \adom(I)^n$ with $\abf
  \approx \cbf$,
\begin{enumerate}

\item $I \models \Gamma_{\vect{c}}$ implies $I[\abf/\cbf] \models
  \Gamma^v_{\vect{c}} (\vect{a})$ and

\item $I \models \Gamma^v_{\vect{c'}}(\vect{a})$, with $\abf \preceq
  \cbf'$, implies $I[\cbf/\abf] \models \Gamma_{\vect{c}}$.

\end{enumerate}
Point~1 can be proved by showing that, from a proof tree of
$\mn{goal}()$ from $I$ and $\Gamma_\cbf$, one can construct a proof
tree of $\mn{goal}(\abf)$ from $I[\abf/\cbf]$ and $\Gamma^v_\cbf$. For
Point~2, assume $I \models \Gamma^{v}_{\vect{c'}}(\vect{a})$ with
$\abf \preceq \cbf'$. Then $I[\cbf/\abf] \models
(\Gamma_{\vect{c'}})[\vect{c}/\vect{c'}]$ can again be shown by
manipulating proof trees. It can be verified that, by construction,
$(\Pi_{\vect{c'}})[\vect{c}/\vect{c'}] \subseteq
\Pi_{\vect{c}}$. Consequently and since $\Gamma_{\cbf'}$ is a
rewriting of $\Pi_{\cbf'}$, $J \models
(\Gamma_{\vect{c'}})[\vect{c}/\vect{c'}]$
implies $J \models \Gamma_{\vect{c}}$ for all $J$, that is, 
$(\Gamma_{\vect{c'}})[\vect{c}/\vect{c'}]$ is contained
in $\Gamma_{\vect{c}}$ in the sense of query containment.
Thus in particular $I[\cbf/\abf] \models \Gamma_\cbf$, as required.

%\smallskip

It remains to show that $\Gamma$ is a rewriting for $\Pi$.  First
assume that $I \models \Pi(\vect{a})$. Choose some $\cbf \in C^n$ with
$\abf \approx \cbf$. Then $I[\cbf/\abf] \models \Pi_{\vect{c}}$ by the
claim and thus $I[\cbf/\abf] \models \Gamma_{\vect{c}}$ since
$\Gamma_\cbf$ is a rewriting of $\Pi_\cbf$. Point~1 above yields
$I[\cbf/\abf][\abf/\cbf]=I \models \Gamma(\abf)$.

Now assume that $I \models \Gamma(\vect{a})$. Then by construction of
$\Gamma$, there is a $\cbf' \in C^n$ such that $\abf \preceq \cbf'$
and $I \models \Gamma^v_{\vect{c'}}(\vect{a})$. To see this, note in
particular that the different programs $\Gamma^v_\cbf$ do not share
any IDBs and thus do not interact in $\Gamma$.  Choose a $\cbf \in
C^n$ with $\abf \approx \cbf$.  From Point~2 above, we obtain
$I[\cbf/\abf] \models \Gamma_{\vect{c}}$ which yields $I[\cbf/\abf]
\models \Pi_{\vect{c}}$. This implies $I \models \Pi(\vect{a})$ by the
claim.
\end{proof}
We next show that constants can be eliminated from Boolean programs. 
\begin{lem}
\label{lem:nonBoolRed2}
  % Rewritability of Boolean MDDLog-programs with constants into
  % $\ell$-DLog with constants can be reduced to rewritability
  % of Boolean MDDLog-programs into $\ell$-DLog (both without
  % constants).
Given a Boolean MDDLog program $\Pi_c$ with constants 
%of diameter $k$
over EDB schema $\Sbf_E$, one can construct a Boolean MDDLog program
$\Pi$ over an EDB schema $\Sbf_E'$ such that
\begin{enumerate}

\item $\Pi_c$ is rewritable into $\ell$-DLog with constants iff 
  $\Pi$ is rewritable into $\ell$-DLog, for any~$\ell$;

\item If $\Pi_c$ is of size $n$ and diameter $k$, then the size of
  $\Pi$ is $2^{p(k \cdot \mn{log}n)}$; moreover, the diameter of
  $\Pi$ is bounded by the rule size of $\Pi_c$.

\end{enumerate}
The construction takes time polynomial in the size of $|\Pi|$. 
\end{lem}
\begin{proof}\ Let $\Pi_c$ be a Boolean MDDLog program over EDB schema
  $\Sbf_E$ that contains constants $c_1, \ldots, c_n$. The program
  $\Pi$ will be over EDB schema $\Sbf'_E = \Sbf_E\cup\{R_1, \ldots,
  R_n\}$ where $R_1, \dots, R_n$ are fresh monadic relation
  symbols. $\Pi$ contains all rules that can be obtained from a rule
  $\rho$ in~$\Pi$ by choosing a partial function $\delta$ that maps
  terms (variables or constants) in $\rho$ to  elements of
  $\{1,\dots,n\}$ such that $\delta(c_i)=i$ for each constant $c_i$
  and then, for each term $t$ with $\delta(t)=i$,
\begin{enumerate}

\item replacing each occurrence of $t$ in the body of $\rho$ with a
  fresh variable $x$ and adding $R_i(x)$, and 
  %; let $\mn{var}(t)$ be the set of fresh variables introduced for $t$ in this way;

\item replacing each occurrence of $t$ in the head of $\rho$ with one
  of the fresh variables introduced for $t$ in Step~1.
% \footnote{Note that
%   multiple variables might be introduced in case that there are
%   multiple occurrences of $t$.}

\end{enumerate}
Additionally, $\Pi$ contains the rule $\mn{goal}() \leftarrow R_i(x),
R_j(x)$, for $1 \leq i <j \leq n$. 

%\smallskip

Note that the rewriting presented above, which we call
\emph{dejoining} since it introduces different variables for each
occurence of a term $t$ in a rule body, can be applied not only to
MDDLog programs, but also to MDLog programs. Before we proceed, we
make a basic observation about dejoining and its connection to a
certain quotient construction.  Let $\Pi$ be an MDDLog program or an
MDLog program, with constants $c_1,\dots,c_n$, and let $\Pi_d$ be the
result of dejoining $\Pi$.  Let $I$ be an $\Sbf'_E$-instance such that
$R_i, R_j$ are disjoint whenever $i \neq j$ and which does not contain
the constants $c_1,\dots,c_n$. The \emph{quotient} of $I$ is the
$\Sbf_E$-instance $I'$ obtained from $I$ by replacing every $d \in
\adom(I)$ with $R_i(d) \in I$ by the constant $c_i$ (which also
results in the identification of elements in the active domain) and
removing all atoms involving one of the $R_i$ relations. By converting
proof trees of $\mn{goal}()$ from $\Pi$ into proof trees of
$\mn{goal}()$ from $\Pi_c$ and vice versa, one can show the following.
\\[2mm]
{\bf Claim}.  $I \models \Pi$ iff $I' \models \Pi_d$. 
\\[2mm]
We now show that $\Pi_c$ is rewritable into $\ell$-DLog iff $\Pi$ is.

%\medskip 
%
First let $\Gamma_c$ be an $\ell$-DLog rewriting of $\Pi_c$. Let
$\Gamma$ be obtained from $\Gamma_c$ by dejoining all rules and adding
the rule $\mn{goal}() \leftarrow R_i(x), R_j(x)$ for $1 \leq i <j \leq
n$. Clearly, $\Gamma$ is an $\ell$-DLog program. We argue that
$\Gamma$ is a rewriting of $\Pi$. Let $I$ be an
$\Sbf'_E$-instance. W.l.o.g., we can assume that $I$ does not contain
$c_1,\dots,c_n$. If $R_i, R_j$ are not disjoint for some $i \neq j$,
then $I \models \Pi$ and $I \models \Gamma$. Otherwise, let $I'$ be
the quotient of~$I$.  We have $I \models \Pi$ iff $I' \models \Pi_c$
(by the claim) iff $I' \models \Gamma_c$ ($\Gamma_c$ is rewriting of
$\Pi_c$) iff $I \models \Gamma$ (again by the claim).

%\medskip

Let $\Gamma$ be an $\ell$-DLog rewriting of $\Pi$. Let $\Gamma_c$ be
the program constructed from $\Gamma$ by removing all rules that
contain atoms of the form $R_i(x)$ and $R_j(x)$ with $i \neq j$ and
replacing all variables $x$ that occur in a rule body in atoms of the
form $R_i(x)$ with $c_i$ and removing all $R_i$-atoms from such
rules. Clearly, $\Gamma_c$ is an $\ell$-DLog program (with constants
$c_1,\dots,c_n$). We argue that $\Gamma_c$ is a rewriting of $\Pi_c$.
Let $I$ be an $\Sbf_E$-instance that w.l.o.g. does not contain
$c_1,\dots,c_n$ and let $I'=I \cup \{R_1(c_1), \ldots,
R_n(c_n)\}$. Note that $I$ is the quotient of $I'$. Then $I \models
\Pi_c$ iff $I' \models \Pi$ (by the claim) iff $I' \models \Gamma$
($\Gamma$ is rewriting of $\Pi$) iff $I \models \Gamma_c$ (by
construction of $\Gamma_c$).
\end{proof}
We are now ready to lift the complexity results from
Theorems~\ref{th:decidBoolean} and~\ref{thm:DLogBoolWithEq} to the
non-Boolean case, by putting them together with
Lemmas~\ref{lem:nonBoolRed1}
and~\ref{lem:nonBoolRed2}.
\begin{thm}
\label{th:general}
For $n$-ary MDDLog programs,
  \begin{enumerate}
  \item FO-rewritability (equivalently: UCQ-rewritability) is
    2\NExpTime-complete;
  \item rewritability into MDLog with $n$ parameters is in 3\ExpTime (and
    2\NExpTime-hard);
  \item DLog-rewritability is 2\NExpTime-complete for programs
    that have equality.
  \end{enumerate}
\end{thm}
\begin{proof}\ We remind that the upper bounds for rewritability of
  Boolean MDDLog programs stated in Theorems~\ref{th:decidBoolean}
  and~\ref{thm:DLogBoolWithEq} are obtained by a (generalized) CSP and
  then deciding the rewritability of (the complement of) that CSP.
  One can trace the blowups stated in Lemmas~\ref{lem:nonBoolRed1}
  and~\ref{lem:nonBoolRed2} as well as in
  Theorems~\ref{th:mmsnpToSimple} and~\ref{th:simpleToCsp} to verify
  that the constructed CSP does not become significantly larger in the
  non-Boolean case, that is, it still satisfy the bounds stated in
  Point~3 of Proposition~\ref{th:mmsnpToSimple}. Thus, we obtain the
  same upper bounds as in the Boolean case.  Regarding Point~1, we
  additionally recall that Proposition~\ref{prop:ross} also covers non-Boolean
  MDDLog programs and thus it suffices to consider
  UCQ-rewritability. Regarding Point~3, we note that it can be verified that the
  constructions in the proofs of Lemmas~\ref{lem:nonBoolRed1}
  and~\ref{lem:nonBoolRed2} preserve the property of having equality.
\end{proof}
In view of Point~2, we remark (once more) that for non-Boolean MDDLog
programs $\Pi$, MDLog with parameters is in a sense a more natural
target for rewriting than MDLog without parameters.  The intuitive
reason is that positions in the answer to $\Pi$ can be thought of as
constants, and constants correspond to parameters. To make this a bit
more precise, consider the grounding $\Pi'$ of $\Pi$ obtained by
replacing, in every goal rule, each variable that occurs in the head
by a constant. In contrast to the standard database setup (and in
contrast to the proof of Lemma~\ref{lem:nonBoolRed1}), we mean here
constants that are interpreted according to the standard FO semantics,
that is, different constants can denote the same element of an
instance. When looking for an MDLog-rewriting of $\Pi'$, it is clearly
very natural to admit the constants from $\Pi'$ also in the
rewriting. Now, one can verify that any such rewriting can be
translated in a straightforward way into a rewriting of $\Pi$ into
MDLog with parameters, and vice versa.

We further note that MDLog with parameters enjoys similarly nice
properties as standard MDLog. For example, containment is
decidable. This follows from
\cite{DBLP:conf/pods/RudolphK13,DBLP:conf/ijcai/BourhisKR15} where
generalizations of MDLog with parameters are studied, the actual
parameters being represented by constants.

We also remark that Theorem~\ref{th:general} remains true when we
admit constants in MDDLog programs. In fact, the proof of
Lemma~\ref{lem:nonBoolRed1} goes through also when the original
MDDLog program contains constants, and both the original and the newly
introduced constants can then be removed by
Lemma~\ref{lem:nonBoolRed2}. 

\subsection{Canonical Datalog-Rewritings}

We now turn our attention to canonical DLog-rewritings for non-Boolean
MDDLog programs. Let $\Pi$ be an $n$-ary MDDLog program. We associate with
$\Pi$ a \emph{canonical $(\ell,k)$-DLog program with $n$ parameters}, for
any $\ell <k$.  The construction is a refinement of the one from the
Boolean case.

We start with some preliminaries. An \emph{$n$-marked instance} is an
instance $I$ endowed with $n$ (not necessarily distinct) distinguished
elements $\cbf=c_1,\dots,c_n$.  An \emph{$(\ell,k)$-tree decomposition
  with $n$ parameters} of an $n$-marked instance $(I,\cbf)$ is an
$(\ell+m,k+m)$-tree decomposition of $I$, $m$ the number of distinct
constants in \cbf, in which every bag $B_v$ contains all constants
from~\cbf. An $n$-marked instance \emph{has treewidth
$(\ell,k)$ with $n$ parameters} if it admits an $(\ell,k)$-tree
decomposition with $n$ parameters. % Note
% that the last parameter is actually not free, it must be $n$ if we
% speak about $n$-marked instances. We still list it explicitly to
% visually distinguish the standard notion of treewidth from the
% non-standard one just introduced.

We first convert $\Pi$ into a DDLog program $\Pi'$ that is equivalent
to $\Pi$ on instances of bounded treewidth. The construction is
identical to the Boolean case (first variable identification, then
rewriting) except that
\begin{enumerate}

\item we use treewidth $(\ell+n,k+n)$ in place of treewidth
  $(\ell,k)$; consequently, the arity of the freshly introduced IDB relations
  may also be up to $\ell+n$;

\item for $\mn{goal}$ rules, all head variables must occur in the root
  bag of the tree decomposition (they can then be treated in the same
  way as a Boolean \mn{goal} rule despite the $n$-ary head relation).
  
\end{enumerate}
It can be verified that $\Pi'$ is sound for $\Pi$ and that it is
complete for $\Pi$ on $n$-marked instances of treewidth $(\ell,k)$
with $n$ parameters in the sense that, for all such instances
$(I,\cbf)$, $I \models \Pi[\cbf]$ implies $I \models \Pi'[\cbf]$.
$\Pi'$ is not guaranteed to be complete for answers other than \cbf
because of the way we treat goal rules in Point~2 above, for example
when $\Pi$ contains a rule of the form $\mn{goal}(x,y) \leftarrow A(x)
\wedge B(y)$.

Let $\Sbf'_I$ denote the additional IDB relations in the resulting
program $\Pi'$. We now construct the canonical $(\ell,k)$-DLog program
with $n$ parameters $\Gamma^c$.  Fix constants $a_1,\dots,a_\ell,$
$b_1,\dots,b_n$ and let $\Imf_{\ell'+n}$ denote the set of all $\Sbf_I
\cup \Sbf_{I'}$-instances with domain
$\abf_{\ell',n}:=a_1,\dots,a_{\ell'},$ $b_1,\dots,b_n$.  The program
uses $\ell'+n$-ary IDB relations $P_M$, for all $\ell' \leq \ell$ and
all $M \subseteq \Imf_{\ell',n}$. It contains all rules $q(\xbf)
\rightarrow P_{M}(\ybf\,|\,\xbf_p)$, $M \subseteq \Imf_{\ell',n}$, that satisfy the
following conditions:
\begin{enumerate}

\item $q(\xbf)$ contains at most $k+n$ variables;

%\item if the IDB relation $P_{S',t'}$ occurs in $q(\xbf)$, then $t'=t$;

\item in every extension $J$ of the $\Sbf_E$-instance $I_q|_{\Sbf_E}$ with
  $\Sbf_I \cup \Sbf'_I$-facts such that
  \begin{enumerate}

  \item $J$ satisfies all rules of $\Pi'$ and does not contain
    $\mn{goal}(\xbf_p)$ and

%  \item for each $P_{N}(\zbf) \in q$, we have $J|_{\Sbf_I \cup \Sbf'_I}[\abf/\zbf] \in N$,
%DUPLICATION!!

  \item for each $P_{N}(\zbf\,|\,\xbf_p) \in q$, $N \subseteq \Imf_{\ell'',n}$,
 there is an $L \in N$ such that
%    $L[\zbf/\abf] \subseteq J$
    $L[\zbf\xbf_p/\abf_{\ell'',n}]=J|_{\Sbf_I \cup  \Sbf'_I},\zbf$

  \end{enumerate}
  there is an $L \in M$ such that $L[\ybf\xbf_p/\abf_{\ell',n}]=J|_{\Sbf_I \cup
    \Sbf'_I},\ybf$

\end{enumerate}
We also include all rules of the form
$P_\emptyset(\ybf\,|\,\xbf_p)\rightarrow \mn{goal}(\xbf_p)$.  
This finishes the construction of~$\Gamma^c$.
%
% It can be shown that the number of IBD relations is bounded by $2^{2^{|\Sbf_I|
%     \cdot \ell}}$ and the number of rules by $2^{|\Sbf_E| \cdot
%   k^{a+\ell} \cdot 2^{|\Sbf_I|\cdot\ell}}+1$ where $a$ is the maximum
% arity of relations in $\Sbf_E$.\footnote{To obtain this bound, one has
%   to assume that the rule body contains at most one IDB atom
%   $P_M(x_1,\dots,x_\ell)$ for any $\ell$-tuple of variables
%   $x_1,\dots,x_\ell$; obviously, this can be done without loss of
%   generality.} The canonical DLog
% It can be verified that 
%  the constructed program is of size double
% exponential in the size of the original program.
%
It is straightforward to verify that $\Gamma^c$ is sound for $\Pi$.
It is complete in the same sense as $\Pi'$.
\begin{lem}
\label{lem:cancompllk2}
$\Gamma^c$ is sound for $\Pi$. It is complete for $\Pi$ on $n$-marked
instances of treewidth $(\ell,k)$ with $n$ parameters in the sense
that for any such instance $(I,\cbf)$, $I \models \Pi(\cbf)$ implies
$I \models \Gamma^c(\cbf)$.
\end{lem}
%
% clu: SUBTLE POINT:
%
% one might be tempted to think that the elements of \cbf should
% not be required to be distinct because elements of answers need
% not be distinct. However, when we read of a homomorphic pre-
% image from the proof-tree, we can make non-distinct elements
% of the answer distinct!!!
%
The proof of Lemma~\ref{lem:cancompllk2} is similar to that of
Lemma~\ref{lem:cancompllk}, details are omitted. In analogy with
Theorem~\ref{thm:canonical}, we can then obtain the following result about
canonical DLog programs.
\begin{thm}
\label{thm:canonical2}
%A Boolean MDDLog program $\Pi$ is DLog-rewritable iff 
%the canonical DLog program for $\Pi$ is a rewriting of $\Pi$. 
Let $\Pi$ be an $n$-ary MDDLog program, $0 < \ell \leq k$, and
$\Gamma^c$ the canonical $(\ell,k)$-DLog program with $n$ parameters
associated with $\Pi$. Then
\begin{enumerate}

%\item $\Gamma^c$ is a sound for $\Pi$;

\item $\Gamma \subseteq \Gamma^c$ for every $(\ell,k)$-DLog
  program $\Gamma$ that is sound for $\Pi$;

\item $\Pi$ is rewritable into $(\ell,k)$-DLog with $n$ parameters iff
  $\Gamma^c$ is a rewriting of~$\Pi$.

\end{enumerate}
\end{thm}
Note that, as a consequence of Theorem~\ref{thm:canonical2}, an
$n$-ary MDDLog program $\Pi$ is DLog-rewritable (in the standard
sense, without parameters) iff the canonical $(\ell,k)$-DLog program
with $n$ parameters is a rewriting, for some $\ell,k$. % When one is
% interested in DLog-rewritability without caring for the exact
% $\ell,k$-parameterization, canonical programs for the non-Boolean case
% thus behave exactly as canonical programs in the Boolean case.
In a sense, this exactly parallels the behaviour of canonical DLog
programs in the Boolean case. As an important consequence, the
reductions presented in Lemmas~\ref{lem:nonBoolRed1}
and~\ref{lem:nonBoolRed2} show that if DLog-rewritability of Boolean
programs turns out to be decidable (without assuming equality), then
the same is true for DLog-rewritability of non-Boolean programs.
Theorems~\ref{th:decidBoolean} and~\ref{thm:DLogBoolWithEq}

\subsection{Shape of Rewritings and Obstructions}

We now analyze the shape of rewritings of non-Boolean MDDLog programs.
An \emph{$(\ell,k)$-tree decomposition with $n$ parameters} of an
$n$-ary CQ $q$ is an $(\ell+n,k+n)$-tree decomposition of $q$
in which every bag $B_v$ contains all answer variables of~$q$.  The
treewidth with $n$ parameters of an $n$-ary CQ is now defined in the
expected way.
\begin{restatable}{thm}{threwrshapet}
\label{TH:REWRSHAPE2}   
  Let $\Pi$ be an $n$-ary MDDLog program of diameter $k$. Then
\begin{enumerate}

\item if $\Pi$ is FO-rewritable, then it has a UCQ-rewriting in which 
  each CQ has treewidth $(1,k)$ with $n$ parameters;

\item if $\Pi$ is rewritable into MDLog with $n$ parameters, then it has
  an MDLog-rewriting with $n$ parameters of diameter
  $k$.

\end{enumerate}
\end{restatable}
Note that Theorem~\ref{TH:REWRSHAPE2} is immediate from
Theorem~\ref{th:rewrshape} and Lemmas~\ref{lem:nonBoolRed1}
and~\ref{lem:nonBoolRed2} when $k$ denotes the rule size of $\Pi$
instead of its diameter. To get the improved version, one needs to
carefully trace the construction of rewritings, starting with
rewritings for the CSPs ultimately constructed and then through the
proofs of Lemmas~\ref{prop:rewrCorresp},~\ref{lem:nonBoolRed2},
and~\ref{lem:nonBoolRed1}. In particular, the constructions in 
Lemmas~\ref{prop:rewrCorresp} and~\ref{lem:nonBoolRed2} interplay
in a subtle way that can be exploited to improve the bound. Details
are given in Appendix~\ref{appc}.

%\smallskip

As in the Boolean case, rewritings are closely related to
obstructions. We define obstruction sets for MMSNP formulas with free
variables and summarize the results that we obtain for them. % An
% \emph{$n$-marked instance} is an instance $I$ endowed with $n$
% distinct distinguished elements $c_1,\dots,c_n$. 
A \emph{set of marked obstructions} $\Omc$ for an MMSNP formula
$\theta$ with $n$ free variables over schema $\Sbf_E$ is a set of
$n$-marked instances over the same schema such that for any
$\Sbf_E$-instance $I$, we have $I \not\models \theta[\abf]$ iff for
some $(O,\cbf) \in \Omc$, there is a homomorphism $h$ from $O$ to $I$ with
$h(\cbf)=\abf$. % An \emph{$(\ell,k)$-tree
%   decomposition with $n$ parameters} of an $n$-marked instance
% $(I,\cbf)$ is an $(\ell+n,k+n)$-tree decomposition of $I$ in which
% every bag $B_v$ contains all constants from \cbf.
We obtain the
following corollary from Point~1 of Theorem~\ref{TH:REWRSHAPE2} in
exactly the same way in which Corollary~\ref{cor:FO} is obtained from
Point~1 of Theorem~\ref{th:rewrshape}.
\begin{cor}
For every MMSNP formula $\theta$ with $n$ free variables, the
following are equivalent:
  \begin{enumerate}
  \item $\theta$ is FO-rewritable;
  \item $\theta$ has a finite marked obstruction set;
  \item $\theta$ has a finite set of finite marked obstructions of treewidth
    $(1,k)$ with $n$ parameters.
  \end{enumerate}
\end{cor}
It is interesting to note that this result can be viewed as a
generalization of the characterization of obstruction sets for CSP
templates with constants in terms of `c-acyclicity' in
\cite{DBLP:journals/tods/AlexeCKT11}; our parameters correspond to
constants in that paper. 
We now turn to MDLog-rewritability.
\begin{prop}
  Let $\theta$ be an MMSNP formula of diameter $k$ with $n$ free
  variables. Then $\neg \theta$ is rewritable into an MDLog program
  with $n$ parameters iff $\theta$ has a set of marked obstructions
  (equivalently: finite marked obstructions) that are of 
  treewidth~$(1,k)$ with $n$ parameters.
\end{prop}
\begin{proof}\ The ``only if'' direction is a consequence of Point~2
  of Theorem~\ref{TH:REWRSHAPE2} and the fact that, for any MDLog
  program $\Pi \equiv \neg \theta$ with $n$ parameters of diameter
  $k$ over EDB schema $\Sbf_E$, a proof tree for
  $\mn{goal}(\cbf)$ from an $\Sbf_E$-instance $I$ and~$\Pi$ gives rise
  to a finite $n$-marked $\Sbf_E$-instance $(J,\cbf)$ of treewidth
  $(1,k)$ with $n$ parameters that satisfies $J \rightarrow I$.  The
  ``if'' direction is a consequence of the fact that the canonical
  $(1,k)$-DLog program with parameters associated with
  $\neg \theta$ viewed as an MDDLog program is complete on inputs of
  treewidth $(1,k,n)$ with $n$ parameters in the sense of
  Lemma~\ref{lem:cancompllk2}.
\end{proof}
As an illustration, it might be interesting to reconsider
Example~\ref{ex:avar2}. The unary MDDLog program shown there is the
negation of a unary MMSNP formula that has as a set of marked
obstructions the set of all $R$-cycles on which one element is the
marked element. Each of these obstructions has treewidth
$(1,2)$ with one parameter, but not treewidth $(1,2)$ in the
strict sense.

% An \emph{$(\ell,k)$-Datalog program with $n$ parameters} is an
% $(\ell+n,k+n)$-Datalog program (of any arity)  with distinguished
% variables $x_1,\dots,x_n$ such that
% %
% \begin{enumerate}

% \item all IDBs $P$ except \mn{goal} have arity at least $n$ and only
%   occur in the form $P(\cdots,x_1,\dots,x_n)$, both in rule bodies and
%   heads;

% \item each \mn{goal} rule contains no IDBs in the body or only
%   distinguished variables in the head.

% \end{enumerate}
%

% \medskip

% An \emph{$(\ell,k)$-Datalog program with $n$ parameters} is an $n$-ary
% $(\ell+n,k+n)$-Datalog program in which all IDBs have arity at least
% $n$ and such that in every rule, all IDBs agree on the variables used
% in the last $n$ positions (both in rule bodies
% and heads, including \mn{goal}).

\section{Ontology-Mediated Queries}
\label{sect:omqs}

While the results on disjunctive Datalog and on MMSNP obtained in the
previous sections are interesting in their own right, our premier aim
is to study fundamental question of rewritability in the context of
ontology-mediated queries (OMQs). Such questions have received a lot
of interest in the OMQ context, see for example
\cite{DBLP:journals/tods/BienvenuCLW14,BHLW-IJCAI16,LutzSabellekIJCAI17}
and references therein. In particular, we settle an open question from
\cite{DBLP:journals/tods/BienvenuCLW14} by showing that in the OMQ
language $(\mathcal{ALCI},\text{CQ})$, introduced in detail below,
FO-rewritability is decidable and {\sc 2NExpTime}-complete.  In what
follows, we first introduce several prominent description logics to
serve as ontology languages and, based on that, ontology-mediated
queries. We then show how the results from the previous sections can
be used to obtain results about ontology-mediated queries.

\subsection{Preliminaries}

In description logics, ontologies are defined by so-called TBoxes. A
TBox, in turn, is a set of inclusions (that is, logical implications)
between concepts (that is, logical formulas), and possibly also
additional kinds of statements. Each description logic is determined
by the constructors that are available to build up concepts
and by the statements that are allowed in TBoxes. Here, we introduce
the widely known description logics \ALC, \ALCI, and $\mathcal{SHI}$,
listed in the order of increasing expressive power. We refer the
reader to \cite{DL-Textbook} for a more thorough introduction to DLs.

An \emph{$\mathcal{ALCI}$-concept} is formed according to the syntax rule
\[
C,D ::= \top \mid \bot \mid A \mid \neg C \mid C \sqcap D \mid C
\sqcup D \mid \exists r . C \mid \exists r^-.C \mid \forall r . C \mid
\forall r^- . C 
\]
where $A$ ranges over a fixed countably infinite set of \emph{concept
  names} and $r$ over a fixed countably infinite set of \emph{role
  names}.  An \emph{\ALC-concept} is an $\mathcal{ALCI}$-concept in
which the constructors $\exists r^- . C$ and $\forall r^-.C$ are not
used.  An \ALC-TBox (resp.\ \ALCI-TBox) is a finite set of concept
inclusions $C \sqsubseteq D$, $C$ and $D$ $\ALC$-concepts (resp.\
\ALCI-concepts). While \ALCI extends \ALC with additional concept
constructors, $\mathcal{SHI}$ extends \ALCI with additional types of
TBox statements. There is thus no need to define
$\mathcal{SHI}$-concepts as these are simply \ALCI-concepts. A
\emph{role} is either a role name or an expression $r^-$ with $r$ a
role name. A \emph{$\mathcal{SHI}$-TBox} is a finite set of
\begin{itemize}
\item \emph{concept inclusions}
$C \sqsubseteq D$, $C$ and $D$ $\mathcal{ALCI}$-concepts, 

\item \emph{role
  inclusions} $r \sqsubseteq s$, $r$ and $s$ roles, and

\item
\emph{transitivity statements} $\mn{trans}(r)$, $r$ a role
name.

\end{itemize}
DL semantics is given in terms of interpretations. An
\emph{interpretation} takes that form $\Imc=(\Delta^\Imc,\cdot^\Imc)$
where $\Delta^\Imc$ is a non-empty set called the \emph{domain} and
$\cdot^\Imc$ is the \emph{interpretation function} which maps each
concept name $A$ to a subset $A^\Imc \subseteq \Delta^\Imc$ and each
role name $r$ to a binary relation $r^\Imc \subseteq r^\Imc \times
r^\Imc$. Note that an interpretation is simply a notational variant of
a relational FO-structure that interprets only unary and binary
relations. The interpretation function is extended to compound
concepts in the standard way, as given in Figure~\ref{fig:semantics}.
\begin{figure}[t]
  \centering
\[
\begin{array}{r@{\;}c@{\;}l}
(\neg C)^\Imc &=& \Delta^\Imc \setminus C^\Imc \\[\myeqnsep] 
(C \sqcap D)^\Imc &=& C^\Imc \cap D^\Imc \\[\myeqnsep]
(C \sqcup D)^\Imc &=& C^\Imc \cup D^\Imc \\[\myeqnsep]

  (\exists r . C)^\Imc &=& \{ d \in \Delta^\Imc \mid \exists e \in 
                           C^\Imc: (d,e) \in r^\Imc \} \\[\myeqnsep] 
  (\exists r^- . C)^\Imc &=& \{ d \in \Delta^\Imc \mid \exists e \in  C^\Imc: (e,d) \in r^\Imc \}\\[\myeqnsep] 
  (\forall r . C)^\Imc &=& \{ d \in \Delta^\Imc \mid \forall e \in 
                           \Delta^\Imc : (d,e) \in r^\Imc \Rightarrow e\in  C^\Imc\} \\[\myeqnsep] 
  (\forall r^- . C)^\Imc &=& \{ d \in \Delta^\Imc \mid \forall e \in 
                           \Delta^\Imc : (e,d) \in r^\Imc \Rightarrow e\in  C^\Imc\}. 
\end{array}
\]
  \caption{Semantics of \ALCI-concepts}
  \label{fig:semantics}
\end{figure}
An intepretation is a \emph{model} of a TBox \Tmc if it
\emph{satisfies} all statements in \Tmc, that is,
\begin{itemize}

\item $C \sqsubseteq D \in \Tmc$ implies $C^\Imc \subseteq D^\Imc$;

\item $r \sqsubseteq s \in \Tmc$ implies $r^\Imc \subseteq s^\Imc$;

\item $\mn{trans}(r) \in \Tmc$ implies that $r^\Imc$ is transitive.

\end{itemize}
For roles $r,s$, we write $\Tmc \models r \sqsubseteq s$ if every model
\Imc of \Tmc satisfies $r^\Imc \subseteq s^\Imc$.

In description logic, data is typically stored in so-called ABoxes.
For uniformity with MDDLog, we use instances instead, identifying
unary relations with concept names, binary relations with role names,
and disallowing relations of any other arity. An interpretation \Imc
is a \emph{model} of an instance $I$ if $A(a) \in I$ implies
$a \in A^\Imc$ and $r(a,b) \in I$ implies $(a,b) \in r^\Imc$. We say
that an instance $I$ is \emph{consistent} with a TBox \Tmc if $I$ and
\Tmc have a joint model. 

% An \emph{ontology-mediated query (OMQ)} takes the form
% $Q=(\Tmc,\Sbf_E,q)$ with \Tmc a TBox, $\Sbf_E$ a set of concept and
% role names, and $q$ a UCQ.  We use $(\Lmc,\Qmc)$ to refer to the set
% of all OMQs whose TBox is formulated in the language \Lmc and where
% the actual queries are from the language \Qmc. For example,
% $(\ALC,\text{UCQ})$ refers to the set of all OMQs that consist of an
% \ALC-TBox and a UCQ. For OMQs $(\Tmc,q)$ from
% $(\mathcal{SHI},\mathcal{UCQ})$, we adopt the following additional
% restriction: when \Tmc contains a transitivity $\mn{trans}(r)$ and
% $\Tmc \models r \sqsubseteq s$, we disallow the use of $s$ in the
% query $q$. Let $I$ be an $\Sbf_E$-instance and $\vect{a}$ a tuple of
% constants from $I$. We write $I \models Q(\vect{a})$ and call
% $\vect{a}$ a \emph{certain answer to $Q$ on} $I$ if for all models
% \Imc of $I$ and~\Tmc, we have $\Imc \models q(\vect{a})$ (defined in
% the usual way).

An \emph{ontology-mediated query (OMQ)} over a schema $\Sbf_E$ is a
triple $(\Tmc,\Sbf_E,q)$ where \Tmc is a TBox formulated in a
description logic and $q$ is a query. % To use similar terminology as for MDDLog, we
% call $\Sbf$ the \emph{EDB schema}.
The TBox can introduce symbols that are not in $\Sbf_E$, which allows
it to enrich the schema available for formulating the query~$q$. In
fact, $q$ can use symbols from $\Sbf_E$, additional symbols from \Tmc,
and also completely fresh symbols (which is useful only in very rare
cases). As the TBox language, we may use any of the description logics
introduced above. Since all these logics admit only unary and binary
relations, we assume that these are the only allowed arities in
schemas throughout Section~\ref{sect:omqs}. As the actual query
language, we use UCQs and CQs. The OMQ languages that these choices
give rise to are denoted with $(\ALC,\text{CQ})$,
$(\mathcal{SHI},\text{UCQ})$, and so on. In the actual query, we
generally disallow the use of role names $r$ such that for some role
name $s$, $\mn{trans}(s) \in \Tmc$ and $\Tmc \models s \sqsubseteq r$.
In fact, admitting such roles in the query poses serious additional
complications, which are outside the scope of this paper; see e.g.\
\cite{DBLP:conf/dlog/BienvenuELOS10,DBLP:conf/icalp/GottlobPT13}.  To
make the restriction explicit, we add a superscript $\cdot^-$ to OMQ
languages when the DL used permits transitivity statements in the
TBox, such as in $(\mathcal{SHI},\text{UCQ})^-$.

The semantics of an OMQ is given in terms of certain answers. Let $I$
be an $\Sbf_E$-instance and $\vect{a}$ a tuple of constants
from~$I$. We write $I \models Q(\vect{a})$ and call $\vect{a}$ a
\emph{certain answer to $Q$ on} $I$ if for all models \Imc of $I$
and~\Tmc, we have $\Imc \models q(\vect{a})$. The latter denotes
satisfaction of $q(\vect{a})$ in \Imc in the usual sense of
first-order logic.
\begin{exa}
Let $Q=(\Tmc,\Sbf_E,q)$ be the following OMQ, formulated in $(\ALC,\text{CQ})$: 
\[
\begin{array}{r@{\;}c@{\;}l@{\;\;\;}c}
  \Tmc &=& \{ \; \exists \mn{hasAbn}. \mn{CTest} \sqsubseteq \mn{Smoker} \sqcup \exists \mn{hasRisk}.\mn{MTC},& (1) \\[0.5mm]
   && \; \;   \;\exists \mn{hasAbn}. \mn{CTest} \sqcap \exists \mn{hasRisk}.\mn{MEN2} \sqsubseteq \exists \mn{hasRisk}. \mn{MTC},  &(2)
\\ &&\; \;  \;\mn{PCCPatient} \sqsubseteq \exists \mn{hasRisk}.\mn{MEN2}   & (3)\\[0.5mm]
 &&\; \; \; \exists \mn{hasRelative} . \exists \mn{hasRisk}.\mn{MEN2}
 \sqsubseteq \exists \mn{hasRisk}. \mn{MEN2} \ \} & (4)\\[0.5mm]
     \Sbf_E &=& \{ \; \mn{hasAbn}, \mn{CTest}, \mn{Smoker}, \mn{hasRelative}, \mn{PCCPatient} \; \} \\[0.5mm]
  q(x) &=& \; \mn{hasRisk}(x,y) \wedge \mn{MTC}(y).
\end{array}
\]
The TBox \Tmc describes the risk of somebody having medullary thyroid
cancer (MTC) in the presence of an abnormal calcitonin test
(CTest). While abnormal calcitonin levels are a marker for MTC, there
can also be false positives, for example due to smoking (Line~1 of
\Tmc). However, in the presence of a high risk for the genetic
syndrome MEN2, high calcitonin levels immediately raise an MTC
suspicion (Line~$2$). Pheochromocitoma patients (PCCPatient) have a
high MEN2 risk (Line~$3$). As MEN2 is caused by a genetic mutation,
the risk carries within families (Line~$4$). %Finally, the \sf{hasRelative} relation is transitive and symmetric (axioms $5$ and $6$).
On the $\Sbf_E$-instance
\[ \mn{hasAbn}(\mn{john}, \mn{t}),\mn{ CTest}(\mn{t}),
\mn{Smoker}(\mn{john}), \mn{hasRelative}(\mn{john}, \mn{anna}), \mn{PCCPatient}(\mn{anna}),
\]
%
%\[
%\begin{array}{l}
%\{ \textsf{hasAbnormal}(john, t),\textsf{ CTest}(t),
%\textsf{Smoker}(john), \textsf{hasRelative}(john, anna), \\[1mm]
%\;\;\textsf{PCCPatient}(anna)\}
%\end{array}
%\] 
%
the
only certain answer to $Q$ is $\mn{john}$.

\end{exa}
An OMQ $Q=(\Tmc,\Sbf_E,q)$ is \emph{FO-rewritable} if there is an FO
query $\varphi(x)$ over schema~$\Sbf_E$ (and possibly involving
equality), called an \emph{FO-rewriting} of $Q$, such that for all
$\Sbf_E$-instances $I$ and $\abf \subseteq \adom(I)$, we have $I
\models Q(\abf)$ iff $I \models \varphi(\abf)$. Other notions of
rewritability such as UCQ-rewritability and MDLog-rewritability are
defined accordingly. % We
% use $|\Tmc|$ and $|q|$ to denote the \emph{size} of a TBox \Tmc and
% query $q$, meaning the number of symbols it takes to write them
% (concept names, role names, and variable names counting as one).

Note that the TBox \Tmc can be inconsistent with the input instance
$I$, that is, there could be no joint model of \Tmc and~$I$. It can thus
be a sensible alternative to work with \emph{consistent
  FO-rewritability}, considering only $\Sbf_E$-instances $I$ that are
consistent w.r.t.\ \Tmc. This can then be complemented with
rewritability of inconsistency for $\Tmc$, that is, rewritability of
the Boolean OMQ $(\Tmc,\Sbf_E,\exists x \, A(x))$, $A(x)$ a fresh
concept name, which is true on an $\Sbf_E$-instance $I$ iff $I$ is
inconsistent with \Tmc. It is not hard to prove, though, that
consistent \Qmc-rewritability can be reduced to \Qmc-rewritability in
polynomial time for all OMQ languages condidered in this paper and all $\Qmc
\in \{ \text{FO}, \text{MDLog}, \text{DLog} \}$; see the corresponding
proof for query containment in \cite{KR-submitted}. Moreover,
rewritability of consistency was studied in \cite{DBLP:journals/tods/BienvenuCLW14} and shown to
be \NExpTime-complete for all OMQ languages considered in this paper.

\subsection{Rewritability of OMQs}

We now lift the results from earlier sections to OMQs.  There is a
known equivalence-preserving translation from the relevant OMQ
languages to MDDLog, but it involves a double exponential blowup
\cite{DBLP:journals/tods/BienvenuCLW14} that most likely is
unavoidable.\footnote{It was shown in
  \cite{DBLP:journals/tods/BienvenuCLW14} that a single exponential
  blowup is unavoidable. Whether the blowup has to be double
  exponential is an open problem.}  We refine this translation and
carefully trace the parameters in which the blowup occurs to show
that, despite these blowups, the complexity of the relevant problems
does not increase. The following is our main result concerning OMQs.
\begin{thm}
\label{thm:alcicompl}
In all OMQ languages between $(\ALC,\text{UCQ})$ and
$(\mathcal{SHI},\text{UCQ})^-$, as well as between $(\ALCI,
\text{CQ})$ and $(\SHI, \text{CQ})^-$,
  \begin{enumerate}
  \item FO-rewritability (equivalently: UCQ-rewritability) is
    2\NExpTime-complete; in fact, there is an algorithm which, given
    an OMQ $Q=(\Tmc,\Sbf_E,q)$, decides in time
    $2^{2^{p(n_q\cdot\mn{log}n_\Tmc)}}$ whether $Q$ is FO-rewritable;

  \item MDLog-rewritability is in 3\ExpTime (and 2\NExpTime-hard); in
    fact, there is an algorithm which, given an OMQ
    $Q=(\Tmc,\Sbf_E,q)$, decides in time
    $2^{2^{2^{p(n_q\cdot\mn{log}n_\Tmc)}}}$ whether $Q$ is MDLog-rewritable
  \end{enumerate}
  where $n_q$ and $n_\Tmc$ are the size of $q$ and \Tmc and $p$ is a polynomial. 
\end{thm}
Note that the runtime for deciding FO-rewritability stated in 
Theorem~\ref{thm:alcicompl} is double exponential only in the size of
the actual query $q$ (which tends to be very small) while it is only
single exponential in the size of the TBox (which can become large)
and similarly for  MDLog-rewritability, only one exponential higher.

The lower bounds in Theorem~\ref{thm:alcicompl} are from
\cite{KR-submitted}. To prove the upper bounds, we first give a
refined translation from OMQs to MDDLog.  A proof is provided in
Appendix~\ref{app:fromOMQtoMDDLog}.
%
%\begin{thm} %[$\!\!$\cite{KR-submitted}]
\begin{restatable}{thm}{thmalci}
\label{THM:ALCI}
For every OMQ $Q=(\Tmc,\Sbf_E,q)$ from $(\mathcal{SHI},\text{UCQ})^-$, one can
construct an equivalent MDDLog program $\Pi$ such that
  \begin{enumerate}

  \item the size of $\Pi$ is bounded by $2^{2^{p(n_q \cdot \mn{log}n_\Tmc)}}$;

  \item the IDB schema of $\Pi$ is of size at most $2^{p(n_q \cdot \mn{log}n_\Tmc)}$;

  \item the rule size of $\Pi$ is bounded by $n_q$

  \end{enumerate}
  where $n_q$ and $n_\Tmc$ are the size of $q$ and \Tmc and $p$ is a
  polynomial. The construction takes time polynomial in the size of $\Pi$.
%\end{thm}
\end{restatable}

Let $Q$ be an OMQ from $(\mathcal{SHI},\text{UCQ})^-$. Instead of
deciding FO- or MDLog-rewritability of~$Q$, we can decide the same
problem for the MDDLog program delivered by
Theorem~\ref{THM:ALCI}. The bounds stated in Theorem~\ref{THM:ALCI},
Lemmas~\ref{lem:nonBoolRed1} and~\ref{lem:nonBoolRed2}, and
Theorems~\ref{th:mmsnpToSimple} and~\ref{th:simpleToCsp}, though, only
guarantee that we obtain a CSP template with 3-exponentially many
elements, which does not yield {\sc 2NExpTime} upper bounds.
However, it is possible to combine the construction underlying
Theorem~\ref{THM:ALCI} with those underlying
Lemmas~\ref{lem:nonBoolRed1} and~\ref{lem:nonBoolRed2} and
Theorem~\ref{th:mmsnpToSimple} to obtain the following.%
\begin{restatable}{lem}{lemomqtosimple}
\label{LEM:OMQTOSIMPLE}
Given an OMQ $Q=(\Tmc,\Sbf_E,q)$
from $(\mathcal{SHI},\text{UCQ})^-$ with $\Tmc$ of size $n_\Tmc$ and $q$
of size $n_q$, one can construct a simple MDDLog program $\Pi_Q$ over
an aggregation EDB schema $\Sbf'_E$ such that%
  \begin{enumerate}
  \item $Q$ is \Qmc-rewritable iff $\Pi_Q$ is \Qmc-rewritable for
  every $\Qmc \in \{ \text{FO}, \text{UCQ}, \text{MDLog} \}$;

  \item the size of $\Pi_Q$ and the cardinality of $\Sbf'_E$ are
    bounded by $2^{2^{p(n_q \cdot \mn{log}n_\Tmc)}}$ and the arity of
    relations in $\Sbf'_E$ is bounded by $\max\{n_q,2\}$;
  \item the IDB schema of $\Pi_Q$ is of size $2^{p(n_q \cdot \mn{log}n_\Tmc)}$
%  \item the diameter of $\Pi$ is bounded by $|q|$,
  \end{enumerate}
  where $p$ is a polynomial.  The construction takes time polynomial
  in the size of~$\Pi_Q$.
\end{restatable}
%
  % In fact, this was already observed in \cite{KR-submitted} where
  % exactly the same constructions that are also used in
  % Lemmas~\ref{lem:nonBoolRed1} and~\ref{lem:nonBoolRed2} and
  % Theorem~\ref{th:mmsnpToSimple} are considered, see the proof of
  % Theorem~20 therein.
A proof is provided in Appendix~\ref{app:OMQtosimple}.
  Constructing a CSP template from this refined
  simple program and applying the decision procedures for
  rewritability of CSP templates, we obtain the upper bounds stated in
  Theorem~\ref{th:simpleToCsp}.

%\medskip

We remark that it is not possible to extend
Theorem~\ref{thm:alcicompl} to description logics with functional
roles or number restrictions since, in such DLs, FO-rewritability of
OMQs is undecidable \cite{DBLP:journals/tods/BienvenuCLW14}. The proof
can be adapted to MDLog-rewritability.

%\smallskip

The results about the shape of rewritings stated in
Theorem~\ref{TH:REWRSHAPE2} (of course) also apply to the OMQ case.
Note that, in Points~1 and~2 of that theorem, we can then replace $k$
with $\max\{n_q,2\}$.  Moreover,
the canonical DLog programs introduced for MDDLog in
Section~\ref{sec:answer} can also be utilized for OMQs via the
translation underlying the proof of Theorem~\ref{THM:ALCI}.

%\smallskip

Regarding Datalog-rewritability of OMQs, we obtain a potentially
incomplete decision procedure by combining Theorem~\ref{THM:ALCI} with
Lemmas~\ref{lem:nonBoolRed1} and~\ref{lem:nonBoolRed2} and the
algorithm from Section~\ref{sect:dlog-rewr}. It is possible to define
a class of OMQs $(\Tmc,\Sbf_E,q)$ that \emph{have equality} and for
which this procedure is complete. Roughly, $\Sbf_E$ needs to contain a
relation $\mn{eq}$ and $\Tmc$ enforces that for all models $\Imc$ of
\Tmc and all $(d,e) \in \mn{eq}^\Imc$, $d$ and $e$ satisfy exactly the
same subconcepts of \Tmc and exactly the same tree contractions of $q$
and then taking a subquery. We refrain from working out the details.

\section{Dichotomy and Deciding PTime Query Evaluation}
\label{sect:dicho}

There was a recent breakthrough in research on CSPs, independently
achieved by Bulatov and by Zhuk, who have proved the long standing
Feder-Vardi conjecture thus establishing a dichotomy between \PTime
and {\sc NP} for the complexity of CSPs
\cite{Bulatov17,Zhuk17}. Together with results by Chen and
Larose~\cite{Chen2017}, this also implies that it
is decidable and {\sc NP}-complete whether the CSP defined by a given
template has \PTime complexity. We observe that, together with the
translations given in this paper, we obtain several interesting
results on MMSNP, MDDLog, and OMQs.

In particular, we consider the (data) complexity of query evaluation,
which is defined in the expected way. For example, each OMQ
$Q=(\Tmc,\Sbf_E,q)$ gives rise to the following \emph{query evaluation
  problem}: given an $\Sbf_E$-instance $I$ and a tuple
$\abf \subseteq \adom(I)$, decide whether $I \models Q(\abf)$. This
problem is guaranteed to be in {\sc coNP} when $Q$ is from any of the
OMQ languages studied in this paper
\cite{DBLP:journals/tods/BienvenuCLW14}, but of course there are also
OMQs $Q$ for which it is in {\PTime} or even simpler and from a
practical perspective it is very important to understand the exact
complexity of evaluating the concrete queries that are relevant for
the application at hand. The definition of query evaluation is
analogous for MDDLog and for MMSNP; note that MMSNP only gives rise to
Boolean queries and that there is an {\sc NP} upper bound for the
complexity rather than a {\sc coNP} one.

The question of \PTime query evaluation also comes with an associated
`meta problem': given an OMQ $Q$ (or a query from some other relevant
language), decide whether $Q$ admits \PTime query evaluation. We
remark that the data complexity of OMQs as well as the associated meta
problem and dichotomy questions have received significant interest
\cite{DBLP:journals/tods/BienvenuCLW14,LMCS-18,LutzSabellekIJCAI17,DBLP:conf/aaai/HernichLW17,DBLP:conf/ijcai/LutzSW15,DBLP:conf/ijcai/LutzSW13,DBLP:journals/ai/CalvaneseGLLR13}. The
following theorem summarizes our results regarding the complexity of
query evaluation.
\begin{thm}
  In MDDLog and all OMQ languages between $(\ALC,\text{UCQ})$
  and $(\mathcal{SHI},\text{UCQ})^-$, as well as between
  $(\ALCI, \text{CQ})$ and $(\SHI, \text{CQ})^-$, 
  \begin{enumerate}

  \item there is a dichotomy in the complexity of query evaluation
    between \PTime and {\sc coNP};
    
  \item deciding \PTime-query evaluation is 2\NExpTime-complete.

  \end{enumerate}
  The same holds for MMSNP, with {\sc coNP} replaced by {\sc NP}
  in Point~(1).
\end{thm}
\begin{proof}\ For MMSNP, it is well-known that there is a dichotomy
  for query evaluation between \PTime and {\sc NP} iff there is such a
  dichotomy for the complexity of CSPs. This gives the MMSNP version
  of~(1). In fact, it is even known that the constructions from the
  proofs of Theorems~\ref{th:mmsnpToSimple} and~\ref{th:simpleToCsp},
  which transform an MMSNP sentence into a CSP, preserve complexity up
  to polynomial time reductions
  \cite{DBLP:journals/siamcomp/FederV98,DBLP:journals/combinatorica/Kun13}. We
  thus obtain the upper bound in the MMSNP version of (2). The lower
  bound is obtained from the reduction used in \cite{KR-submitted}
  to show that the Datalog-rewritability of (the complement of) MMSNP
  sentences is 2\NExpTime-hard. The proof is by reduction of a tiling
  problem.  Given such a tiling problem, one constructs an MMSNP
  sentence $\vp$ such that $\vp$ is FO-rewritable if there is a tiling
  and equivalent to 3-colorability (thus not Datalog-rewritable and
  {\sc NP}-hard) otherwise.  Clearly, such a reduction also yields
  2\NExpTime-hardness of \PTime query evaluation.

  For the cases of MDDLog and for OMQs, lower bounds are obtained
  along the same lines, that is, by observing that the lower bound
  constructions from \cite{KR-submitted} are directly applicable. For
  the upper bounds and the dichotomies, we first observe that the
  construction in the proofs of Lemma~\ref{lem:nonBoolRed2} preserves
  complexity up to polynomial time reductions (which is implicit in
  the proof) and recall that the translation in
  Lemma~\ref{LEM:OMQTOSIMPLE} is even equivalence-preserving. It thus
  remains to deal with Lemma~\ref{lem:nonBoolRed1}. There, an $n$-ary
  MDDLog program $\Pi$ is translated into a family of Boolean MDDLog
  programs (with constants) $\Pi_\cbf$, $\cbf \in C^n$ where $C$ is
  fixed set of $n$ constants. The claim formulated in the proof of
  Lemma~\ref{lem:nonBoolRed1} provides
  \begin{enumerate}

  \item a polynomial time reduction of evaluating $\Pi$ to evaluating
    the programs $\Pi_\cbf$: given an instance $I$ and an
    $\abf \subseteq \mn{Ind}(I)$, to decide whether $I \models \Pi(\abf)$
    choose $\cbf$ with $\abf \approx \cbf$ and check whether
    $I[\cbf/\abf] \models \Pi_\cbf$;

  \item for each $\cbf \in C^n$, a polynomial time reduction of
    evaluating $\Pi_\cbf$ to evaluating $\Pi$: given an instance $I$,
    to decide whether $I \models \Pi_\cbf$ check whether
    $I \models \Pi(\cbf)$.

  \end{enumerate}
  It follows that $\Pi$ can be evaluated in \PTime iff all of the
  programs $\Pi_\cbf$ can. This is enough to transfer the upper bound
  for deciding \PTime query evaluation. For dichotomy, we additionally
  need that if one of the programs $\Pi_\cbf$ is {\sc coNP}-hard, then
  so is $\Pi$, which follows from~(2).
\end{proof}

\section{Discussion}
\label{sec:concl}

We have clarified the decidability status and computational complexity
of FO- and MDLog-rewritabil\-ity in MMSNP, MDDLog, and various OMQ
languages based on expressive description logics and conjunctive
queries, and we also made several interesting observations regarding
dichotomies and the decidability and complexity of \PTime query
evaluation. For Datalog-rewritability, we were only able to obtain
partial results, namely a sound algorithm that is complete on a
certain class of inputs and potentially incomplete in general.  This
raises several natural questions: is our algorithm actually complete
in general? Does an analogue of Lemma~\ref{lem:FOgirthdontcare} (that
is, rewritability on instances of high girth implies rewritability)
hold for Datalog as a target language? What is the complexity of
deciding Datalog-rewritability in the afore-mentioned languages? From
an OMQ perspective, it would also be important to work towards more
practical approaches for computing (FO-, MDLog-, and DLog-)
rewritings. Given the high computational complexities involved, such
approaches might have to be incomplete to be practically
feasible. However, the degree/nature of incompleteness should then be
characterized, and we expect the results in this paper to be helpful
in such an endeavour.

%\medskip

\bibliographystyle{alpha}
\bibliography{lmcssubmission}

\newpage

\appendix

\section{Translating Boolean MDDLog to generalized CSP}
\label{app:translate}

\subsection{From MDDLog to Simple MDDLog}
\label{app:step1tosimple}

Let $\Pi$ be a Boolean MDDLog program over schema $\Sbf_E$ and of
diameter $k$. We first construct from $\Pi$ an equivalent Boolean
MDDLog program $\Pi_B$ such that the following conditions are
satisfied:
\begin{enumerate}[label=(\roman*)]

\item all rule bodies are biconnected,
that is, when any single variable is removed from the body (by
deleting all atoms that contain it), then the resulting rule body is
still connected; 

\item if $R(x,\dots,x)$ occurs in a rule body with $R$ EDB, then
  the body contains no other EDB atoms.

\end{enumerate}
A good way to think about what is achieved in this first step is that,
when the resulting program is evaluated on an instance of bounded
treewidth, then it suffices to map the rule bodies to individual bags
while it is never necessary to cross `bag boundaries'.

To construct $\Pi_B$, we first extend $\Pi$ with all contractions of
rules in $\Pi$; we will refer to this step as the \emph{collapsing
  step}. We then split up rules that are not biconnected into multiple
rules by exhaustively executing the following rewriting steps:
\begin{itemize}

\item replace every rule
  $p(\vect{y}) \leftarrow q_1(\vect{x}_1) \wedge q_2(\vect{x}_2)$
  where $\vect{x}_1$ and $\vect{x}_2$ share exactly one variable $x$
  but both contain also other variables with the rules
$p_1(\vect{y}_1) \vee Q(x) \leftarrow q_1(\vect{x}_1)$ and
  $p_2(\vect{y}_2) \leftarrow Q(x) \wedge q_2(\vect{x}_2)$, where $Q$
  is a fresh monadic IDB relation and $p_i(\vect{y}_i)$ is the
  restriction of $p(\vect{y})$ to atoms that are nullary or contain a
  variable from $\xbf_i$, $i \in \{1,2\}$;

\item replace every rule
  $p(\vect{y}) \leftarrow q_1(\vect{x}_1) \wedge q_2(\vect{x}_2)$
  where $\vect{x}_1$ and $\vect{x}_2$ share no variables and are both
  non-empty with the rules $p_1(\vect{y}_1) \vee Q() \leftarrow
  q_1(\vect{x}_1)$ and $p_2(\vect{y}_2) \leftarrow
  Q() \wedge q_2(\vect{x}_2)$, where
  $Q()$ is a fresh nullary IDB relation and the $p_i(\vect{y}_i)$ are
  as above;

\item replace every rule
  $p(\vect{y}) \leftarrow R(x,\dots,x) \wedge q(\vect{x})$ where $R$
  is an EDB relation and $q$ contains at least one EDB atom and the
  variable $x$, with the rules $Q(x) \leftarrow R(x,\dots,x)$ and
  $p(\vect{y}) \leftarrow Q(x) \wedge q(\vect{x})$, where $Q$ is a
  fresh monadic IDB relation.
  
\end{itemize}
It is easy to see that the resulting program $\Pi_B$ is equivalent to
the original program~$\Pi$ and that all $\Pi_B$ satisfies
Conditions~(i) and~(ii)
above.  % Note that we need to collapse \emph{before}
% making rules biconnected because a biconnected rule body might be no
% longer biconnected after collapsing.

We next construct from $\Pi_B$ the desired simple program $\Pi_S$ by
replacing, in every rule, the EDB atoms in the rule body with a single
EDB atom that represents the conjunction of all atoms replaced. We
thus introduce fresh EDB relations that represent conjunctions of old
EDB relations. Note that there can be implications between the new EDB
relations that we will have to take care of in the construction of
$\Pi_B$.
 % The construction of $\Pi_S$ thus involves replacing the EDB
% schema $\Sbf_E$ with a different EDB schema $\Sbf'_E$.

Let $Q_\Pi$ denote the set of CQs that can be obtained from a rule
body in $\Pi_B$ by consistently renaming variables, using only
variables that occur in $\Pi_B$. Let $\Sbf_I$ be the IDB schema of
$\Pi_B$. For every $q(\vect{x}) \in \Qmc_\Pi$, we write
$q(\vect{x})|_{\Sbf_E}$ to denote the restriction of $q(\vect{x})$ to
$\Sbf_E$-atoms, and likewise for $q(\vect{x})|_{\Sbf_I}$ and IDB
atoms. The EDB schema $\Sbf'_E$ of $\Pi_S$ consists of the relations
$R_{q(\vect{x})|_{\Sbf_E}}$, $q(\vect{x}) \in \Qmc_\Pi$, whose arity is
the number of variables in $q(\vect{x})$ (which, by construction of
$\Pi_B$, is identical to the number of variables in
$q(\vect{x})|_{\Sbf_E}$). % A \emph{variable
%   substitution} is a map $\sigma$ that takes each variable to a
% variable.
% For any variable substitution $\sigma$ and CQ
% $q(\vect{x})$, we use $\sigma(q(\vect{x}))$ to denote the
% result of consistently replacing variables in $q(\vect{x})$ according to
% $\sigma$. 
% 
%
The program $\Pi_S$ consists
of the following rules:
\begin{quote}

  whenever $p(\vect{y}) \leftarrow q_1(\vect{x}_1)$ is a rule in
  $\Pi_B$, $q_2(\vect{x}_2) \in \Qmc_\Pi$, and
  $q_1(\vect{x}_1) \subseteq q_2(\vect{x}_2)$, then $\Pi_S$ contains
  the rule
  $p(\vect{y}) \leftarrow R_{q_2(\xbf_2)}(\vect{x}_2) \wedge
  q_1(\vect{x}_1)|_{\Sbf_I}$

% \item[($*$)] whenever $q_1(\vect{x}_1) \rightarrow
%   p(\vect{y})$ is a rule in $\Pi_B$, $R_{q_2(\vect{x}_2)}$ an
%   EDB relation in $\Sbf'_E$, and $\sigma: \vect{x}_1 \rightarrow
%   \vect{x}_2$ an
%   injective variable substitution such that $\sigma(q_1(\vect{x}_1)) \subseteq
%   q_2(\vect{x}_2)$, then $\Pi_S$ contains the rule
%   $\sigma^{-1}(R_{q_2}(\vect{x}_2)) \wedge
%   q_1(\vect{x}_1)|_{\Sbf_I} \rightarrow p(\vect{y})$

\end{quote}
The case where $q_1(\vect{x}_1)$ is identical to $q_2(\vect{x}_2)$
corresponds to adapting rules in $\Pi_B$ to the new EDB signature and
the other cases take care of implications between EDB relations.
\begin{exa}
\label{ex:trans}
Assume that $\Pi_B$ contains the following rules, where $A$ and $r$
are EDB relations:
\[
\begin{array}{r@{\;}c@{\;}l}
  P(x_3) 
&\leftarrow&
  A(x_1) \wedge r(x_1,x_2) \wedge r(x_2,x_3) \wedge r(x_3,x_1)
\\[\myeqnsep]
\mn{goal}()
&\leftarrow& 
  r(x_1,x_2) \wedge r(x_2,x_3) \wedge r(x_3,x_1)
  \, \wedge \\[\myeqnsep]
  && P(x_1)   \wedge P(x_2)   \wedge P(x_3)
\end{array}
\]
A new ternary EDB relation $R_{q_2}$ is introduced for the EDB body
atoms of the lower rule, where
$q_2 = r(x_1,x_2) \wedge r(x_2,x_3) \wedge r(x_3,x_1)$, and a new
ternary EDB relation $R_{q_1}$ is introduced for the upper rule,
$q_1 = A(x_1) \wedge q_2$. In $\Pi_S$, the rules are replaced with
\[
\begin{array}{r@{\;}c@{\;}l}
P(x_3) 
&\leftarrow& 
  R_{q_1}(x_1,x_2,x_3) 
\\[\myeqnsep]
\mn{goal}() 
&\leftarrow& 
  R_{q_2}(x_1,x_2,x_3)    \wedge P(x_1)   \wedge P(x_2)   \wedge
  P(x_3)
 \\[\myeqnsep]
\mn{goal}() 
&\leftarrow& 
  R_{q_1}(x_1,x_2,x_3)    \wedge P(x_1)   \wedge P(x_2)   \wedge
  P(x_3)
\end{array}
\]
Note that $q_2 \subseteq q_1$ and thus $q_1$ logically implies $q_2$,
which results in two copies of the goal rule to be generated.
\end{exa} 
Proof details for the following lemma can be found in
\cite{KR-submitted}.  Recall that $k$ is the diameter of the original
MDDLog program $\Pi$.
\begin{lem}
\label{lem:instancesbackforth}
~\\[-5mm]
  \begin{enumerate}

  \item If $I$ is an $\Sbf_E$-instance and $I'$ the corresponding
    $\Sbf'_E$-instance, then $I \models \Pi$ iff $I' \models \Pi_S$;

  \item If $I'$ is an $\Sbf'_E$-instance and $I$ the corresponding
    $\Sbf_E$-instance, then 
    \begin{enumerate}

    \item $I' \models \Pi_S$ implies $I \models \Pi$;

    \item $I \models \Pi$ implies $I' \models \Pi_S$ if
      the girth of $I'$ exceeds $k$.

    \end{enumerate}

  \end{enumerate}
\end{lem}
The following example demonstrates why the restriction to high girth
instances in Point~2b of Lemma~\ref{lem:instancesbackforth} is
necessary, see also Example~\ref{ex:veryfirst}.
\begin{exa}
\label{ex2}
Consider the programs $\Pi_B$ and $\Pi_S$ from
Example~\ref{ex:trans}.  Take the $\Sbf'_E$-instance $I'$ defined by
\[
  R_{q_1}(a,a',c'), R_{q_1}(b,b',a'), R_{q_1}(c,c',b').
\]
It can be verified that $I' \not\models \Pi_S$. But the corresponding
$\Sbf_E$-instance $I$ is such that $\Pi_B$ derives the IDB relation $P$ at
$a'$, $b'$, and $c'$, and additionally $I$ contains the facts
\[
   r(c',b'), r(b',a'),r(a',c')
\]
which are not covered by any fact in $I'$. Thus clearly $I \models \Pi_B$.
\end{exa}

\subsection{From Simple MDDLog to Generalized CSP}
\label{app:step2tosimple}

Let $\Pi$ be a simple MDDLog program over EDB schema $\Sbf_E$ and with
IDB schema~$\Sbf_I$.  For $i \in \{0,1\}$, an \emph{i-type} is a set
$t$ of relation symbols from $\Sbf_I$ of arity at most $i$ that does
not contain $\mn{goal}()$ and that satisfies all rules in $\Pi$ which
use only IDB relations of arity at most $i$ and do not involve any EDB
relations.

We build a template $T_\theta$ for each 0-type $\theta$. The elements
of $T_\theta$ are exactly the 1-types that agree with $\theta$ on
nullary IDB relations. $T_\theta$ consists of the following facts:
\begin{enumerate}

\item $P()$ for each nullary $P \in \theta$. 
\item $P(t)$ for each 1-type $t$ and each monadic $P \in t$;
\item $R(t_1,\dots,t_n)$ for each relation $R \in \Sbf_E$ and all 1-types $t_1,\dots,t_n$ such that $\Pi$ does not contain a rule 
  \[ 
  P(x_i) \leftarrow R(x_1,\dots,x_n) \wedge P_1(x_{i_1}) \wedge
  \cdots \wedge P_n(x_{i_n})
  \]
  such that $P_j \in t_{i_j}$ for $1 \leq j \leq n$, and $P \notin t_i$.

\end{enumerate}
%
%Let $S_\Pi = \{ T_\theta \mid \theta \text{ a 0-type} \}$. 
The
following was observed in
\cite{DBLP:journals/siamcomp/FederV98}.
\begin{lem}
  For any $\Sbf_E$-instance $I$, we have $I \models \Pi$ iff
  $I \not\rightarrow T_\theta$ for all $0$-types $\theta$.
%  $I \not\models S_\Pi$. 
\end{lem}

\section{MDLog-Rewritability of Generalized CSP}
\label{app:mdlogexp}

In the proof of the subsequent theorem, we use obstructions of CSPs,
which are defined in Section~\ref {sect:shapeObstrExplosion}.
%
% In the proof of the subsequent theorem, we are going to use
% obstructions of generalized CSPs. Obstructions of CSPs are defined in
% Section~\ref{sect:shapeObstrExplosion}. The definition for generalized
% CSPs is as expected: An \emph{obstruction set} $\Omc$ for a set of CSP
% templates $S$ over schema $\Sbf_E$ is a set of instances over the same
% schema such that for any $\Sbf_E$-instance $I$, we have $I
% \not\rightarrow T$ for all $T \in S$ iff $O \rightarrow I$ for some $O
% \in \Omc$.
%
\begin{thm}
  Given a finite set of templates $S$, it can be decided in \ExpTime
  whether coCSP$(S)$ is MDLog-rewritable.
%
  % There is an algorithm that, given a finite set of templates $S$,
  % decides whether coCSP$(S)$ is MDLog-rewritable and that runs in time
  % $2^{p(n)} \times p(m)$, where $n=\sum_{T \in S} |\adom(T)|$,
  % $m=|T|$, and $p$ is a polynomial.
\end{thm}
\begin{proof}\ Consider coCSP$(S)$ over schema $\Sbf_E$ with $S :=
  \{T_1,\dots,T_n\}$. We start with observing that we can assume that
  the templates in $S$ are mutually homomorphically incomparable: if
  this is not the case, we remove templates that are not
  homomorphically minimal and further remove templates so that none of
  the remaining templates are homomorphically equivalent. Clearly,
  this is equivalence preserving and can be done in \ExpTime.

  We aim to show that coCSP$(S)$ is MDLog-rewritable if and only if
  coCSP$(T_i)$ is for all $i \in \{1,\dots,n\}$, which gives the
  desired \ExpTime upper bound.  The ``if'' direction is immediate
  since the union of MDLog programs is expressible as an MDLog
  program.  For the ``only if'' direction, assume that coCSP($S$) is
  MDLog-rewritable, and let $\Gamma$ denote a concrete rewriting.
  Consider a template $T_j$ and let $\mathcal{O}(T_j)$ denote the set
  of all finite $\Sbf_E$-instances of treewidth $(1,k)$ that do not
  homomorphically map to $T_j$ where $k$ is the maximum number of
  variables that occur in a single rule of $\Gamma$.  We will show
  that $\mathcal{O}(T_j)$ is an obstruction set for~$T_j$. It then
  follows from Theorem~23 of \cite{DBLP:journals/siamcomp/FederV98},
  which says that the existence of an obstruction set of treewidth
  $(1,k)$ for some fixed $k$ implies MDLog-rewritability, that
  coCSP$(T_j)$ is MDLog-rewritable.

  By definition of $\mathcal{O}(T_j)$, it is immediate that if $O
  \rightarrow I$ for some $O \in \Omc(T_j)$ and $\Sbf_E$-instance~$I$,
  then $I \not\rightarrow T_j$. We now establish the converse.  Assume
   that 
   %there is an $\Sbf_E$-instance $I$ such that $O \not\rightarrow I$
  % for all $O \in \Omc(T_j)$. 
   $I \not\rightarrow T_j$. Consider the disjoint union $U$ of $I$ and
   $T_j$. Since the templates in $S$ are homomorphically incomparable,
   $U \not\rightarrow T_i$ for all $i \in \{1,\dots,n\}$. Thus $U
   \models \Gamma$ and there is a proof tree for $\mn{goal}()$ from
   $U$ and~$\Gamma$. From that tree, we can read off an
   $\Sbf_E$-instance $J$ such that $J \rightarrow U$, $J$ has
   treewidth $(1,k)$, and $J \models \Gamma$. From the latter, we get
   $J \not \rightarrow T_j$. There must thus also be a connected
   component $O$ of $J$ with $O \not\rightarrow T_j$. We clearly have
   $O \in \mathcal{O}(T_j)$. Since $O \rightarrow U$, $O
   \not\rightarrow T_j$, and $O$ is connected, we moreover get $O
   \rightarrow I$ which finishes the proof.
\end{proof}

\section{Proof of Theorem~\ref{TH:REWRSHAPE2}}\label{appc}

%{
%\renewcommand{\thetheorem}{\ref{TH:REWRSHAPE2}}

 \threwrshapet*

\begin{proof}\ We treat the two cases, FO-rewritability and
  MDLog-rewritability with parameters, in parallel in a uniform
  way. To achieve uniformity, recall that FO-rewritability coincides
  with UCQ-rewritability by Proposition~\ref{prop:ross} and observe
  that a UCQ-rewriting of treewidth $(1,k)$ with $n$ parameters can be
  converted into a non-recursive MDLog-rewriting with $n$ parameters
  of diameter $k$ and vice versa. We work with the latter.
 
Assume that an $n$-ary MDDLog program $\Pi$ over EDB schema $\Sbf_E$
is rewritable into (non-recursive) MDLog with $n$ parameters. We can
convert
\begin{enumerate}

\item $\Pi$ into Boolean MDDLog programs $\Pi_1,\dots,\Pi_k$ with
  constants (Lemma~\ref{lem:nonBoolRed1}),

\item $\Pi_1,\dots,\Pi_k$ into Boolean MDDLog programs
  $\Pi'_1,\dots,\Pi'_k$ without constants (Lemma~\ref{lem:nonBoolRed2}),

\item $\Pi'_1,\dots,\Pi'_k$ into simple Boolean MDDLog programs
  $\Pi''_1,\dots,\Pi''_k$ (Theorem~\ref{th:mmsnpToSimple}), and

\item $\Pi''_1,\dots,\Pi''_k$ into CSP templates $T_1,\dots,T_k$
  (Theorem~\ref{th:simpleToCsp})

\end{enumerate}
such that all these programs and (complements of) templates are
rewritable into (non-recursive) MDLog. Moreover, in the proofs of the
mentioned lemmas and theorems, it is shown how to construct
(non-recursive) MDLog-rewritings of $\Pi''_1,\dots,\Pi''_k$ from given
ones of $T_1,\dots,T_k$, for $\Pi'_1,\dots,\Pi'_k$ from given ones of
$\Pi''_1,\dots,\Pi''_k$, and so on. We are going to analyze these
constructions in more detail.

We first note that for any (non-recursive) MDLog-rewritable CSP, there
is a (non-recursive) MDLog-rewriting where every rule body has at most
one EDB atom that contains all variables which occur in the rule body.
Since each program $\Pi''_i$ is actually \emph{equivalent} to the
complement of the CSP template $T_i$ in Step~4, the same is true for
the programs $\Pi''_i$. Thus, there is a (non-recursive)
MDLog-rewriting $\Gamma''_i$ of $\Pi''_i$ in which 
\begin{itemize}[align=left]

\item[($\dagger$)] each rule body has at most one EDB atom that contains all
  variables.

\end{itemize}

The translation of $\Pi'_i$ into $\Pi''_i$ in Step~3 involves
replacing the EDB schema $\Sbf_E$ with an aggregation schema
$\Sbf'_E$. More precisely, $\Sbf'_E$ consists of relations $R_{q(x)}$
where $q(x)$ is obtained from a rule body in $\Pi'_i$ by first
contracting, then splitting up the body into biconnected components,
and finally dropping all IDB relations. When translating the rewriting
$\Gamma''_i$ of $\Pi''_i$ into a rewriting $\Gamma'_i$ of $\Pi'_i$,
this change in schema is reverted. By ($\dagger$), the diameter of
$\Gamma'_i$ is thus bounded by the arity of relations in $\Gamma''_i$
and that arity, in turn, is bounded by the diameter of~$\Pi'_i$.
% While Step~1 preserves diameter, Step~2 only guarantees that the
% diameter of $\Pi'_i$ is bounded by the rule size of $\Pi_i$. It thus
% seems that we only obtain a weaker version of
% Theorem~\ref{th:rewrshape2} where $k$ is the rule size of $\Pi$
% instead of its diameter. 
%
What's more important, though, is that we actually know what the rule
bodies in $\Gamma'_i$ look like: % since every rule in $\Gamma''_i$
% contains a single EDB atom that contains all variables and due to the
% way in which $\Sbf'_E$ related to $\Sbf_E$, 
%
\begin{itemize}[align=left]

\item[($\ddagger$)] every rule body in $\Gamma'_i$ is obtained from a
  rule body in $\Pi'_i$ by first contracting, then splitting
  up the body into biconnected components, then dropping all IDB
  relations, and finally decorating with some fresh IDB relations
  without introducing fresh variables.

\end{itemize}
Now consider the translation of $\Pi_i$ into $\Pi'_i$ in Step~2 and
the corresponding translation of $\Gamma'_i$ into a rewriting
$\Gamma_i$ of $\Pi_i$. In the former, we dejoin rule bodies by
(sometimes) replacing different occurrences of the same variable $x$
with different variables $x_1,x_2$ and adding the atoms $R_j(x_1)$ and
$R_j(x_2)$ for some $j$, thus increasing the diameter. In the latter,
we rejoin the dejoined rules in $\Gamma'_i$ in the sense that we
replace variables $x,y$ with the same constant $c_j$ whenever the rule
body contains the (EDB) atoms $R_j(x)$ and $R_j(y)$. It can be
verified that rejoining any rule body of the form ($\ddagger$) results
in a rule body whose diameter is bounded by the diameter of
$\Pi'_i$. This gives the desired result since Step~1 preserves
diameter.
\end{proof}

\section{Proof of Theorem~\ref{THM:ALCI}}
\label{app:fromOMQtoMDDLog}

\thmalci*

\noindent

\begin{proof}\ Let $Q=(\Tmc,\Sbf_E,q_0)$ be an OMQ from
  $(\mathcal{SHI},\text{UCQ})$ and let $n_{q_0}$ and $n_\Tmc$ be the
  size of $q_0$ and \Tmc, respectively. We use $\mn{sub}(\Tmc)$ to
  denote the set of subconcepts of (concepts occurring in)
  \Tmc. Moreover, let $\Gamma$ be the set of all tree CQs that
  can be obtained from a CQ in $q_0$ by first existentially
  quantifying all answer variables, then contracting, and then taking
  a subquery.  % A conjunctive query $q$ is \emph{tree-shaped} if
  %
  % \begin{enumerate}
  % \item the undirected graph $(V,\{\{ x,y \} \mid r(x,y) \in q
  %   \})$ is a tree (where $V$ is the set of variables in~$q$),
%
  %  \item $r_1(x,y),r_2(x,y) \in q$ implies $r_1=r_2$, and
%
  %  \item $r(x,y) \in q$ implies $s(y,x) \notin q$ for all
  %    $s$. 
%
  % \end{enumerate}
  %
  Every $q \in \Gamma$ can be viewed as a $\mathcal{ALCI}$-concept
  provided that we additionally choose a root $x$ of the tree.  We
  denote this concept with $C_{q,x}$. For example, the tree CQ
  $q=\exists x \exists y \exists z \, r(x,y) \wedge A(y) \wedge
  s(x,z)$ and choice of $x$ as the root yields the
  $\mathcal{ALCI}$-concept $C_{q,x} = \exists r . A \sqcap \exists s
  . \top$. Let $\mn{con}(q_0)$ be the set of all these concepts
  $C_{q,x}$ and let $\Sbf_I$ be the schema that consists of monadic
  relation symbols $P_C$ and $\overline{P}_C$ for each $C \in
  \mn{sub}(\Tmc) \cup \mn{con}(q_0)$ and nullary relation symbols
  $P_q$ and $\overline{P}_q$ for each $q \in \Gamma$. We are going to
  construct an MDDLog program $\Pi$ over EDB schema $\Sbf_E$ and IDB
  schema $\Sbf_I$ that is equivalent to $Q$.

  By a \emph{diagram}, we mean a conjunction $\delta(\vect{x})$ of
  atoms over the schema $\Sbf_E \cup \Sbf_I$. For an interpretation
  \Imc, we write $\Imc \models \delta(\vect{x})$ if there is a
  homomorphism from $\delta(\vect{x})$ to \Imc, that is, a map
  $h:\vect{x} \rightarrow \Delta^\Imc$ such that:
  \begin{enumerate}

  \item $A(x) \in \delta$ with $A \in \Sbf_E$ implies $h(x) \in A^\Imc$;

  \item $r(x,y) \in \delta$ with $r \in \Sbf_E$ implies $(h(x),h(y)) \in A^\Imc$;

  \item $P_q() \in \delta$ implies $\Imc \models q$ and
    $\overline{P}_q() \in \delta$ implies $\Imc \not \models q$;

  \item $P_C(x) \in \delta$ implies $h(x) \in C^\Imc$ and
    $\overline{P}_C() \in \delta$ implies $h(x) \notin C^\Imc$.

  \end{enumerate}
  We say that $\delta(\vect{x})$ is \emph{realizable} if there is an
  model \Imc of \Tmc with $\Imc \models \delta(\vect{x})$.  A diagram
  $\delta(\vect{x})$ \emph{implies} a CQ $q(\vect{x}')$, with
  $\vect{x}'$ a tuple of variables from $\xbf$, if every homomorphism
  from $\delta(\vect{x})$ to some model \Imc of \Tmc is also a
  homomorphism from $q(\vect{x}')$ to \Imc. The MDDLog program $\Pi$
  consists of the following rules:
  \begin{enumerate}

  \item the rule $P_q() \vee \overline{P}_q() \leftarrow \mn{true}(x)$
    for each $q \in \Gamma$;

  \item the rule $P_C(x) \vee \overline{P}_C(x) \leftarrow \mn{true}(x)$
    for each $C \in \mn{sub}(\Tmc) \cup \mn{con}(q_0)$;

   % \item the rule $\delta(x) \rightarrow \bot$ for each 
   %     non-realizable diagram $\delta(\vect{x})$ that contains at 
   %     most a single variable $x$;

  \item the rule $\bot \leftarrow \delta(x)$ for each 
    non-realizable diagram $\delta(x)$ that contains a single
    variable $x$ and only atoms of the form $P_C(x)$, $C \in
    \mn{sub}(\Tmc) \cup \mn{con}(q_0)$;

  \item the rule $\bot \leftarrow \delta(\vect{x})$ for each 
    non-realizable connected diagram $\delta(\vect{x})$ that contains 
    at most two variables and at most three atoms;

  \item the rule $ \mn{goal}(\vect{x}') \leftarrow \delta(\vect{x})$
    for each diagram $\delta(\vect{x})$ that implies $q_0(\vect{x})$,
    has at most $n_{q_0}$ variable occurrences, and uses only relations
    of the following form: $P_q$, $P_C$ with $C$ a concept name that
    occurs in $q_0$, and role names from $\Sbf_E$ that occur in $q_0$.
 % and
 %       contains no relations of the form $P_C$ or $\overline{P}_C$.
    
  \end{enumerate}
  %
  % Let $\Sbf_\Tmc$ be the set of symbols used in \Tmc and $q$. 
  To understand $\Pi$, a good first intuition is that rules of type~1
  and~2 guess an interpretation \Imc, rules of type~3 and~4 take care
  that the independent guesses are consistent with each other, with
  the facts in $I$ and with the inclusions in the TBox~\Tmc, and rules
  of type~5 ensure that $\Pi$ returns the answers to $q_0$ in~$\Imc$.

  However, this description is an oversimplification. Guessing \Imc is
  not really possible since \Imc might have to contain additional
  domain elements to satisfy existential quantifiers in \Tmc which may
  be involved in homomorphisms from (a CQ in) $q_0$ to $\Imc$, but new
  elements cannot be introduced by MDDLog rules. Instead of
  introducing new elements, rules of type~1 and~2 thus only guess the
  tree CQs that are satisfied by those
  elements. Tree CQs suffice because $\mathcal{SHI}$ has a
  tree-like model property and since we have disallowed the use of
  roles in the query that have a transitive subrole. The notion of
  `diagram implies query' used in the rules of type~5 takes care that
  the guessed tree CQs are taken into account when looking
  for homomorphisms from $q_0$ to the guessed model. This construction
  is identical to the one used in the proof of Theorem~1 of
  \cite{DBLP:conf/pods/BienvenuCLW13}, with two exceptions. First, we
  use predicates $P_C$ and $\overline{P}_C$ for every concept $C \in
  \mn{sub}(\Tmc) \cup \mn{con}(q_0)$ while the mentioned proof uses a
  predicate $P_t$ for every subset $t \subseteq \mn{sub}(\Tmc) \cup
  \mn{con}(q_0)$. And second, our versions of Rules~3-5 are formulated
  more carefully. % we only admit $|q_0|$ atoms of the form
  % $P_q$ or $\overline{P}_q$ in rules of type~4 instead of an
  % unrestricted number.
  It can be verified that the correctness proof
  given in \cite{DBLP:conf/pods/BienvenuCLW13} is not affected by
  these modifications.  The modifications do make a difference
  regarding the size of $\Pi$, though, which we analyse next.

  It is not hard to see that, for some polynomial $p$, the number of
  rules of type~1 is bounded by $2^{p(n_{q_0})}$, the number of rules
  of type~2 and of type~4 is bounded by $2^{p(n_{q_0}\cdot \mn{log}n_\Tmc)}$, the
  number of rules of type~3 is bounded by $2^{2^{p(n_{q_0} \cdot
      \mn{log} n_\Tmc)}}$, and
  the number of rules of type~5 is bounded by $2^{2^{p(n_{q_0})}}$.
  Consequently, the overall number of rules is bounded by
  $2^{2^{p(n_{q_0} \cdot \mn{log}n_\Tmc)}}$ and so is the size of
  $\Pi$. The bounds on the size of the IDB schema and number of rules
  in $\Pi$ stated in Theorem~\ref{THM:ALCI} are easily verified.  It
  remains to argue that the construction can be carried out in double
  exponential time. It suffices to observe two facts. First,
  consistency of a given diagram $\delta(\vect{x})$ can be decided in
  \ExpTime since the satisfiability of $\mathcal{SHI}$ concepts
  w.r.t.\ TBoxes is in \ExpTime \cite{DBLP:phd/dnb/Tobies01}. And
  second, for a given diagram $\delta(\vect{x})$ and CQ $q(\vect{x}')$
  with $\vect{x}'$ a tuple of variables from $\xbf$, it can be decided
  in time single exponential in the size of $\delta(\vect{x})$ and of
  \Tmc and double exponential in the size of $q(\vect{x}')$ whether
  $\delta(\vect{x})$ implies $q(\vect{x}')$.  This is a consequence of
  the fact that, in $\mathcal{SHI}$, given an ABox $\Amc$ that may
  contain compound concepts (in place of concept names), a TBox
  $\Tmc$, a CQ $q(\xbf)$ and a candidate answer $\abf$, it can be
  decided in time single exponential in the size of $\Amc$ and \Tmc
  and double exponential in the size of $q$ whether $a$ is a certain
  answer to $q$ on \Amc w.r.t.~\Tmc~\cite{GliHoLuSa-JAIR08}.
\end{proof}

\section{Proof of Lemma~\ref{LEM:OMQTOSIMPLE}}
\label{app:OMQtosimple}

\lemomqtosimple*

\begin{proof}\ We convert $Q$ into an MDDLog program $\Pi_0$ as per
  Theorem~\ref{THM:ALCI} and then remove the answer variables
  according to the constructions in the proofs of
  Lemmas~\ref{lem:nonBoolRed1} and~\ref{lem:nonBoolRed2}, which gives
  programs $\Pi_1$ and $\Pi_2$.  Analyzing the latter constructions
  reveals that the number of rules on $\Pi_1$ is bounded by
  $r \cdot a^a$ rules where $r$ is the number of rules in $\Pi_0$ and
  $a$ is its arity. Moreover, the rule size %and variable width
  does not increase and neither the IDB schema nor the EDB schema
  changes. The latter construction produces a program with
  $r' \cdot s^s$ rules where $r'$ is the number of rules in $\Pi_1$
  and $s$ is the rule size of $\Pi_1$. Moreover, the IDB schema is not
  changed and the rule size at most doubles.  The EDB schema of the
  new program comprises $a$ fresh monadic relation symbols.  It can
  thus be verified that the obtained Boolean MDDLog program $\Pi_2$
  still satisfies Conditions~1-3 of Theorem~\ref{THM:ALCI} except that
  $n_q$ in the last point has to be replaced by $2n_q$. We make this
  explicit for the reader's convenience:
  \begin{enumerate}

  \item the size of $\Pi_2$ is bounded by $2^{2^{p(n_q \cdot \mn{log}n_\Tmc)}}$;

  \item the IDB schema of $\Pi_2$ is of size at most $2^{p(n_q \cdot \mn{log}n_\Tmc)}$;

  \item the rule size of $\Pi_2$ is bounded by $2n_q$.

  \end{enumerate}
  We next convert $\Pi_2$ into a simple Boolean MDDLog program $\Pi_Q$
  according to Theorem~\ref{th:mmsnpToSimple}. Let us analyze the
  construction in detail to understand the size of $\Pi_Q$, of its EDB
  schema $\Sbf'_E$, and of its IDB schema $\Sbf'_I$. 

  The initial variable identification step can be ignored. In fact, we
  start with at most $2^{2^{p(n_q\cdot \mn{log}n_\Tmc)}}$ rules, each
  of size at most $2n_q$. Thus variable identification results in a
  factor of $(2n_q)!$ regarding the program size and rule number,
  which is absorbed by $2^{2^{p(n_q\cdot \mn{log}n_\Tmc)}}$, and the
  other relevant parameters do not change; in particular, the IDB
  schema remains unchanged.
 
  The next and central step is to make rules biconnected. Given that the rule size is at most $2n_q$, this can split up each rule into at most $2n_q$ rules. This is absorbed by the bounds on program size and rule number. However, on first glance it might seem that we end up with a double exponentially large IDB schema since we might have to split up a double exponential number of rules, each time introducing at least one fresh IDB relation. To argue that this is actually not the case, we distinguish rules of type~1-2 and~4-5 from the construction of $\Pi_1$ (proof of Theorem~\ref{THM:ALCI}); note that the constructions in the proofs of Lemmas~\ref{lem:nonBoolRed1}
  and~\ref{lem:nonBoolRed2} modify the rules only in a very mild way and thus for every rule in $\Pi_2$ it is still clear which type it has.

We need not worry about rules of Type~1-2 and~4-5 since there are only $2^{p(|n_q| \cdot \mn{log}|n_\Tmc|)}$ many such rules, each of
  size at most $2n_q$, and thus the number of additional IDB relations
  introduced for making them biconnected is also bounded by
  $2^{p(n_q \cdot \mn{log}n_\Tmc)}$. Rules of type~3 in $\Pi_0$, on
  the other hand, are of a very restricted form, namely
  \[
    \bot \leftarrow P_{C_1}(x) \wedge \cdots \wedge P_{C_n}(x)
  \]
  with $C_1,\dots,C_n \in \mn{sub}(\Tmc) \cup \mn{con}(q_0)$.  These
  rules are biconnected and thus we are done when $Q$ is Boolean. In
  the non-Boolean case, rules of the above lead to the introduction of
  additional rules in the construction in the proof of
  Lemma~\ref{lem:nonBoolRed2}. This results in rules in $\Pi_2$ that
  are of the form
  \[
    \bot \leftarrow P_{C_1}(x_1) \wedge R_i(x_1) \wedge \cdots \wedge
    P_{C_n}(x_n) \wedge R_i(x_n)
  \]
  where $R_a$ is one of the fresh IDB relations introduced in the
  mentioned construction.  The latter rules have to be split up to be
  made biconnected. This will result in rules of the form
  \[
    \bot \leftarrow Q_1() \wedge \cdots \wedge Q_n() \quad
    \text{ and }
    \quad
    Q_i() \leftarrow P_C(x) \wedge R_a(x)
  \]
  Clearly, there are only $2^{p(n_q \cdot \mn{log}n_\Tmc)}$ many rule
  bodies of the latter form and thus it suffices to introduce at most
  the same number of fresh IDB relations $Q_i$. Thus, the size of the
  IDB schema of $\Pi_Q$ is bounded by
  $2^{p(n_q \cdot \mn{log}n_\Tmc)}$. Also note that, at this Point,
  the rule size has (potentially) decreased and is bounded by
  $\max\{n_q,2\}$. This is obvious for rules of Type~1-2 and~4-5,
  and also for the rules obtained from making rules of Type~3
  biconnected, see above.

  The last step is the change of EDB schema. It involves no blowups
  and we thus obtain the bounds stated in Lemma~\ref{LEM:OMQTOSIMPLE}.
  In particular, the arity of relations in $\Sbf'_E$ is bounded by
  $\max\{n_q,2\}$ since it is bounded by the rule size of the program
  that we had obtained before changing the EDB schema.
\end{proof}

\end{document}